\documentclass[letterpaper]{article} % DO NOT CHANGE THIS

\usepackage{palatino}  % DO NOT CHANGE THIS
\usepackage[hyphens]{url}  % DO NOT CHANGE THIS
\usepackage{graphicx} % DO NOT CHANGE THIS
\usepackage{graphicx}  % DO NOT CHANGE THIS
\usepackage[margin=1in]{geometry}
\usepackage{dsfont}
\usepackage{wrapfig}
\usepackage{multirow}
\usepackage{bbold}
\usepackage{url}
\usepackage{hyperref}
\usepackage[stable]{footmisc}
\usepackage[normalem]{ulem}

\usepackage[table, dvipsnames]{xcolor}
\usepackage{changes}
\definechangesauthor[color=orange]{GK}

\title{
  Fast quantum state reconstruction \\ via accelerated non-convex programming}
\author{\large Junhyung Lyle Kim\textsuperscript{\rm 1}, George Kollias\textsuperscript{\rm 2}, Amir Kalev\textsuperscript{\rm 3}, Ken X. Wei\textsuperscript{\rm 2}, Anastasios Kyrillidis\textsuperscript{\rm 1} \\
\large \textsuperscript{\rm 1} Computer Science, Rice University, Houston, TX 77098, USA \\ %If you have multiple authors and multiple affiliations
\large \textsuperscript{\rm 2} IBM Quantum, IBM T.J. Watson Research Center, Yorktown Heights, NY 10598, USA\\ %If you have multiple authors and multiple affiliations
\large \textsuperscript{\rm 3 }Information Sciences Institute, University of Southern California, Arlington, VA 22203, USA %If you have multiple authors and multiple affiliations
}

\usepackage{amsmath,amsfonts,amsthm,amssymb,xspace,bm}
\usepackage{verbatim,dsfont,mathtools,algorithm,algorithmic,url}
\usepackage{booktabs}

\usepackage{multirow}
\usepackage{xfrac}
\usepackage{wrapfig}
\usepackage{varwidth}
\usepackage{colortbl}
\usepackage{epstopdf}
\usepackage{appendix}
\usepackage{enumerate}
\usepackage{pifont}
\usepackage[table]{xcolor}
\usepackage{color}

\theoremstyle{plain}
\newtheorem{corollary}{Corollary}
\newtheorem{definition}{Definition}

\DeclareMathOperator{\trace}{\textsc{Tr}}

\newcommand{\ket}[1]{ \left| #1 \right\rangle}
\newcommand{\bra}[1]{ \left\langle #1\right|}

\newcommand{\tr}{\texttt{Tr}}
\usepackage{enumitem}
\usepackage{color}

\newtheorem{theorem}{Theorem}

\newtheorem{lemma}{Lemma}

\begin{document}

\maketitle

\begin{abstract}
%Given that cutting-edge scientific studies on quantum algorithms are relying on fault-tolerant quantum computers, the behavior of current harware implementations needs to be characterized, verified, and rigorously certified, before their widespread commercial use. Quantum state tomography (QST) is the canonical procedure for this task: QST is a generic protocol for estimating the density matrix of a quantum system's state from measured data; from this procedure, one can exactly identify the nature of imperfections and deviations in hardware implementation. The estimation is completed by solving a numerical program that searches for the density matrix that is most consistent with the observed data, under a properly chosen cost function. However, as density matrices grow exponentially in the number of subsystems, one must introduce new numerical methods to the toolbox of QST procedures. [AK: I think this can be said in the intro, and it makes the abs too long.]

We propose a new quantum state reconstruction method that combines ideas from compressed sensing, non-convex optimization, and acceleration methods. 
The algorithm, called Momentum-Inspired Factored Gradient Descent (\texttt{MiFGD}), extends the applicability of quantum tomography for larger systems.
% Despite being a non-convex method, \texttt{MiFGD} converges \emph{provably} to the true density matrix at a linear rate, in the absence of experimental and statistical noise, and under common assumptions.
Despite being a non-convex method, \texttt{MiFGD} converges \emph{provably} close to the true density matrix at an accelerated linear rate, in the absence of experimental and statistical noise, and under common assumptions.
With this manuscript, we present the method, prove its convergence property and provide Frobenius norm bound guarantees with respect to the true density matrix. 
From a practical point of view, we benchmark the algorithm performance with respect to other existing methods, in both synthetic and real experiments performed on an IBM's quantum processing unit. 
We find that the proposed algorithm performs orders of magnitude faster than state of the art approaches, with the same or better accuracy. 
In both synthetic and real experiments, we observed accurate and robust reconstruction, despite experimental and statistical noise in the tomographic data.
Finally, we provide a ready-to-use code for state tomography of multi-qubit systems. 
% We also discuss its implementation in the context of quantum process tomography. 
% \amir{This is just a proposal for things we can include. Tasos: yes this is something we can discuss, but proceess tomorgraphy might be even more diffucult to handle and there will be no theory for it.}
\end{abstract}

\section{Introduction}
Quantum tomography is one of the main procedures to identify the nature of imperfections and deviations in quantum processing unit (QPU) implementation \cite{altepeter20044, eisert2019quantum}. %The idea is simple: One can imagine a situation in which a person, say Bob, designs and constructs a quantum computer that presumably operates as promised. Not confident with Bob's description, Alice may wish to test the performance of the implementation by measuring how close the estimated output---represented as a density matrix, after executing a circuit that represents the testing algorithm---is to the expected output. \amir{One can argue  ``then just do fidelity estimation instead of tomography''} 
Generally, quantum tomography is composed of two main parts: $i)$ measuring the quantum system, and $ii)$ analyzing the measurement data to obtain an estimation of the density matrix (in the case of state tomography \cite{altepeter20044}), or of the quantum process (in the case of process tomography \cite{mohseni2008quantum}). 
In this manuscript, we focus on the case of state tomography.

As the number of free parameters that define quantum states and processes scale exponentially with the number of subsystems, generally quantum tomography is a non-scalable protocol \cite{gross2010quantum}. 
In particular, quantum state tomography (QST) suffers from two bottlenecks related to its two main parts. 
The first concerns with the large data one needs to collect to perform tomography; the second concerns with numerically searching in an exponentially large space for a density matrix that is consistent with the data.  

There have been various approaches over the years to improve the scalability of QST, as compared to full QST \cite{vogel1989determination, jevzek2003quantum, banaszek2013focus}. 
Focusing on the data collection bottleneck, to reduce the resources required, prior information about the unknown quantum state is often assumed. 
For example, in compressed sensing QST \cite{gross2010quantum, kalev2015quantum}, it is assumed that the density matrix of the system is low-rank. 
In neural network QST \cite{Torlai2018Neural, beach2019qucumber, torlai2019machine}, one assumes real and positive wavefunctions, which occupy a restricted place in the landscape of quantum states.
Extensions of neural networks to complex wave-functions, or the ability to represent density matrices of mixed states, have been further considered in the literature, after proper reparameterization of the Restricted Boltzmann machines \cite{Torlai2018Neural}.
The prior information considered in these cases is that they are characterized by structured quantum states, which is the reason for the very high performances of neural network QST  \cite{Torlai2018Neural}.\footnote{\cite{Torlai2018Neural} considers also the case of a completely unstructured case and test the limitation of this technique, which does not perform as expected due to lack of structure.}
%the quantum state has a low-depth (i.e., efficient) neural-network description. 
Similarly, in matrix-product-state tomography \cite{Cramer2010Efficient, lanyon2017efficient}, one assumes that the state-to-be-estimated can be represented with low bond-dimension matrix-product state. 

Focusing on the computational bottleneck, several works introduce sophisticated numerical methods to improve the efficiency of QST. 
Particularly, variations of gradient descent convex solvers---e.g., \cite{gonccalves2016projected, bolduc2017projected, shang2017superfast, HU2019RECONSTRUCTING}---are time-efficient 
%to estimate up to a 14-qubit state in a few hours, 
in idealized (synthetic) scenarios \cite{HU2019RECONSTRUCTING}, and only after a proper distributed system design \cite{hou2016full}.
The problem is that achieving such results seems to require utilizing special-purpose hardware (like GPUs).
Thus, going beyond current capabilities requires novel methods that efficiently search in the space of density matrices under more realistic scenarios. 
Importantly, such numerical methods should come with  guarantees on their performance and convergence. % of the proposed method. %We refer the reader to the Related Work section below for a more detailed discussion about existing numerical methods for quantum states tomography. 

%In the noisy-intermediate scale quantum (NISQ) era,  quantum information processors are big enough so that traditional numerical estimation tools  for state tomography becoming inefficient, while small enough so that novel numerical tools can improve state estimation capabilities \cite{}. 

The setup we consider here is that of an $n$-qubit state, under the prior assumption that the state is close to a pure state, and thus its density matrix is of low-rank.
This assumption is justified by state-of-the-art experiments, where our aim is to manipulate the pure states by unitary maps. 
From a theoretical perspective, the low-rank assumption means that we can use compressed sensing techniques, which allow the recovery of the density matrix from relatively few measurement data \cite{liu2011universal}. 

Indeed, by now, compressed sensing QST is widely used for estimating highly-pure quantum states, e.g., \cite{riofrio2017experimental, kliesch2019guaranteed, flammia2012quantum, gross2010quantum}. 
However, compressed sensing QST usually relies on convex optimization for the estimation part \cite{kalev2015quantum}; this limits the applicability to relatively small system sizes \cite{gross2010quantum}. 
On the other hand, non-convex optimization can preform much faster than its convex counterpart \cite{kyrillidis2017provable}. 
% Though typical non-convex optimization lack of convergence guarantees, 
Although non-convex optimization typically lacks convergence guarantees, it was recently shown that one can formulate compressed sensing QST as a non-convex problem and solve it with rigorous convergence guarantees (under certain but generic conditions), allowing state estimation of larger system sizes \cite{kyrillidis2017provable}. 

%In this paper, we introduce another layer to non-convex compressed sensing quantum tomography by considering momentum acceleration methods. We prove that non-convex momentum acceleration provides {\it linear} convergence rates in iterate distance, under common assumptions. In practice, this result enables quantum state tomography of system sizes larger than possible by state-of-the-art algorithms.  

Following the non-convex path, we introduce a new algorithm to the toolbox of QST---the Momentum-Inspired Factored  Gradient Descent (\texttt{MiFGD}). 
Our approach combines ideas from compressed sensing, non-convex optimization, and acceleration techniques, to allow pushing QST beyond current capabilities. 
% \texttt{MiFGD} includes acceleration motions per iteration, that non-trivially complicate theoretical justification of its performance in non-convex scenarios, but nevertheless boosts the practical efficiency of the algorithm. 
\texttt{MiFGD} includes acceleration motions per iteration, that non-trivially complicate theoretical convergence analysis. Nevertheless, we justify the efficacy of the algorithm both in theory -- by achieving an accelerated linear rate -- and in practice.
% justification of its performance in non-convex scenarios, but nevertheless boosts the practical efficiency of the algorithm. 

The contributions of the paper are summarized as follows: \vspace{-0.0cm}
\begin{itemize}[leftmargin=0.7cm]
% \item [$i)$] We prove that the non-convex \texttt{MiFGD} algorithm has indeed favorable {\it linear} convergence rate, in terms of iterate distance, in the noiseless measurement data case and under common assumptions. \vspace{-0.2cm}
\item [$i)$] We prove that the non-convex \texttt{MiFGD} algorithm has indeed {\it accelerated linear} convergence rate, in terms of iterate distance, in the noiseless measurement data case and under common assumptions. \vspace{-0.2cm}
\item [$ii)$] We provide QST results based on real data from IBM's quantum computers up to 8-qubits, contributing to recent efforts on testing QST algorithms in real quantum data \cite{riofrio2017experimental}. Our synthetic examples scale up to 12-qubits effortlessly, leaving the space for an efficient, hardware-aware implementation open for future work. \vspace{-0.2cm}
\item [$iii)$] We show in practice that \texttt{MiFGD} allows faster estimation of quantum states compared to state-of-the-art convex and non-convex algorithms, including recent deep learning approaches \cite{Torlai2018Neural, beach2019qucumber, torlai2019machine, gao2017efficient}, even in the presence of statistical noise in the measurement data. \vspace{-0.2cm}
%\item [$iv)$] We showcase the capabilities of our algorithm to be easily implemented on a distributed system, highlighting the efficiency our approach provides when one considers larger quantum system sizes. \vspace{-0.2cm}
\item [$iv)$] We exploit parallel computations in \texttt{MiFGD} by extending its implementation to enable efficient, parallel execution over shared and distributed memory systems. This way, we experimentally showcase the scalability of this work, which is particularly critical for tackling larger quantum system sizes. \vspace{-0.2cm}
\item [$v)$] We provide implementation of our approach, compatible to the open-source software Qiskit \cite{qiskit}, at \href{https://github.com/gidiko/MiFGD}{https://github.com/gidiko/MiFGD}. \vspace{-0.2cm}
\end{itemize}

\section{Results}
\subsection{Setup}
%The matrix sensing problem has been studied  extensively in the literature; see \cite{recht2010guaranteed} and references to it. %It has numerous applications including video background subtraction~\cite{waters2011sparcs}, system approximation~\cite{fazel2001rank} and identification~\cite{liu2009interior}, robust PCA~\cite{candes2011robust}, and quantum state tomography~\cite{flammia2012quantum}.

We consider the estimation of a low-rank density matrix $\rho^\star \in \mathbb{C}^{d \times d}$ on an $n$-qubit Hilbert space, $d=2^n$, through the following $\ell_2$-norm reconstruction objective: %\cite{bhojanapalli2016dropping,bhojanapalli2016global,li2017algorithmic}:\footnote{The rectangular case can be derived, after proper transformations, using ideas from \cite{park2016finding}.}
\begin{equation}\label{eq:obj}
\begin{aligned}
& \min_{\rho \in \mathbb{C}^{d \times d}}
& & f(\rho) := \tfrac{1}{2} \|\mathcal{A}(\rho) - y\|_2^2 \\
& \text{subject to}
& & \rho \succeq 0, ~\texttt{rank}(\rho) \leq r.
\end{aligned}
\end{equation}
Here, $y \in \mathbb{R}^m$ is the measured data\footnote{Specific description on how $y$ is generated and what it represents will follow.} (observations), and $\mathcal{A}(\cdot): \mathbb{C}^{d \times d} \rightarrow \mathbb{R}^m$ is the linear sensing map, where $m \ll d^2$. 
The sensing map relates the density matrix $\rho^\star$ to (expected, noiseless) observations through the Born rule,  $\left(\mathcal{A}(\rho) \right)_i = \texttt{Tr}(A_i \rho)$, where $A_i \in \mathbb{C}^{d \times d}$, $i=1,\dots,m$, are matrices closely related to the measured observable or the POVM elements of appropriate dimensions. % $A_i \in \mathbb{C}^{d \times d}$.
For concreteness, we focus on the least squares objective function.\footnote{Our results rely on standard optimization assumptions (restricted smoothness and restricted strong convexity assumptions \cite{negahban2012restricted}).}
%\gdk{Could we use something stronger that "conjecture" here?}
The constraint that a density matrix is a non-negative matrix, $\rho \succeq 0$, is a convex constraint. In contrast, the constraint on its rank, $\texttt{rank}(\rho) \leq r$, is a non-convex constraint that promotes a low-rank solution.
Following compressed sensing QST results~\cite{kalev2015quantum}, the constraint $\tr(\rho)=1$ (that should be satisfied, by definition, by any density matrix)  can be ignored, without affecting the scaling of the precision of the final estimation.

A pivotal assumption is that the linear map $\mathcal{A}$ satisfies the \emph{restricted isometry property}:
\begin{definition}[Restricted Isometry Property (RIP) \cite{recht2010guaranteed}]\label{def:rip}
A linear operator $\mathcal{A} :~\mathbb{C}^{d \times d} \rightarrow \mathbb{R}^m$ satisfies the RIP on rank-$r$ matrices, with parameter  $\delta_{r} \in (0, 1)$, if the following holds for any rank-$r$ matrix $X\in \mathbb{C}^{d \times d}$, with high probability:
\begin{align}
(1 - \delta_r) \cdot \|X\|_F^2 \leq \|\mathcal{A}(X)\|_2^2 \leq (1 + \delta_r) 
\cdot \|X\|_F^2.
\end{align}
\end{definition}
Such maps (almost) preserve the Frobenius norm of low-rank matrices, and, as an extension, of low-rank Hermitian matrices.
The intuition behind RIP is that $\mathcal{A}(\cdot)$ behaves as almost a bijection between the subspaces $\mathbb{C}^{d \times d}$ and $\mathbb{R}^m$, when we focus on low rank matrices.
  
Following recent works~\cite{kyrillidis2017provable},  instead of solving Eq.~(\ref{eq:obj}), we propose to solve a factorized version of it:
\begin{equation}\label{eq:factobj}
\min_{U \in \mathbb{C}^{d \times r}} ~\tfrac{1}{2} \|\mathcal{A}(UU^\dagger) - y\|_2^2,
\end{equation}
where $U^\dagger \in \mathbb{C}^{r \times d}$ denotes the adjoint of $U$.
The motivation for this reformulation is as follows: instead of representing  the density matrix $\rho$ as a ${d \times d}$ Hermitian matrix, and imposing the low-rank constraint as in Eq.~(\ref{eq:obj}), we work in a space where low-rank density matrices are %compressed to their information theoretic limit and 
represented through factors $U \in \mathbb{C}^{d \times r}$. 
The low-rankness of $\rho$ is enforced through the factorization of the density matrix into a outer product of such a rectangular matrix representation $U \in \mathbb{C}^{d \times r}$ with its Hermitian conjugate, where $d \gg r$.
%\gdk{Is the intention here to use "rectangular" as opposed to "square" matrix representation? Is this commonly used in this context?}
By rewriting $\rho = UU^\dagger$, for $U \in \mathbb{C}^{d \times r}$, both the PSD constraint ($\rho \succeq 0$) and the low-rankness constraint ($\texttt{rank}(\rho) \leq r$) are directly satisfied, leading to the non-convex formulation~\eqref{eq:factobj}. 
Working in the factored space was shown~\cite{kyrillidis2017provable, park2016finding, park2016provable, tu2016low, zhao2015nonconvex, zheng2015convergent} to improve time and space complexities. %see also Section~\ref{sec:related} for a subset of references on the subject. 
%Since any matrix $\rho \succeq 0$ with $\texttt{rank}(\rho) \leq r$, can be written as $\rho = UU^\dagger$, for $U \in \mathbb{C}^{d \times r}$, this re-parameterization encapsulates both constraints in \eqref{eq:obj} leading to the non-convex formulation \eqref{eq:factobj}.

A common approach to solve \eqref{eq:factobj} is to use gradient descent on the parameter $U$, with iterates generated by the rule:\footnote{We assume cases where $\nabla f(\cdot) = \nabla f(\cdot)^\dagger$. 
If this does not hold, the theory still holds by carrying around 
$\nabla f(\cdot) + \nabla f(\cdot)^\dagger$ instead of just $\nabla f(\cdot)$,
after proper scaling.}
\begin{align}
U_{i+1} &= U_{i} - \eta \nabla f(U_i U_i^\dagger) \cdot U_i \label{eq:fgd} \\ &= U_{i} - \eta \mathcal{A}^\dagger \left(\mathcal{A}(U_i U_i^\dagger) - y\right) \cdot U_i.
\end{align}
Here, $U_i \in \mathbb{C}^{d \times r}, ~\forall i$.
The operator $\mathcal{A}^\dagger :~\mathbb{R}^{m} \rightarrow \mathbb{C}^{d \times d}$
is the adjoint of $\mathcal{A}$, defined as $\mathcal{A}^\dagger(x) = \sum_{i = 1}^m x_i A_i$, for $x \in \mathbb{R}^m$. 
%Similarly, $U^\dagger$ is the adjoint of matrix $U$.
The hyperparameter $\eta>0$ is the step size.
This algorithm has been studied in \cite{bhojanapalli2016dropping,zheng2015convergent,tu2016low,park2016non,ge2017no,hsieh2017non}. 
We will refer to the above iteration as the \emph{Factored Gradient Descent} (FGD) algorithm, as in \cite{park2016finding}. 
%\gdk{Should we perhaps use "iteration" rather than "recursion" here and in what follows for referring to the "structure" of the algorithm?}
%None of the above works have considered momentum acceleration on $U$. 
In what follows, we will study the \texttt{MiFGD} algorithm, a momentum-inspired factored gradient descent. 

\subsection{Momentum-Inspired Factored  Gradient Descent and Main Results}
The \texttt{MiFGD} algorithm is given in the Methods section. It is a two-step variant of FGD:
\begin{align}
U_{i+1} &= Z_{i} - \eta \mathcal{A}^\dagger \left(\mathcal{A}(Z_i Z_i^\dagger) - y\right) \cdot Z_i, \label{eq:MiFGD1}\\
Z_{i+1} &= U_{i+1} + \mu \left(U_{i+1} - U_i\right). \label{eq:MiFGD2}
\end{align}
Here, $Z_i$ is a rectangular matrix (with the same dimension as $U_i$) that accumulates the ``momentum" of the iterates $U_i$. 
$\mu$ is the momentum parameter that weighs how the previous estimates $U_i$ will be mixed with the current estimate $U_{i+1}$ to generate $Z_{i+1}$. The above iteration is an adaptation of Nesterov's accelerated first-order method for convex problems \cite{nesterov1983method}. %: briefly, consider the generic convex optimization problem $\min_{x \in \mathbb{R}^d} g(x)$, where $g$ is a convex function that satisfies standard Lipschitz gradient continuity assumptions with Lipschitz constant $L$. Nesterov's accelerated method is given by the recursion: $x_{i+1} = y_i - \tfrac{1}{L} \nabla g(y_i), ~ \text{and} ~ y_{i+1} = x_{i+1} + \mu_{i} (x_{i+1} - x_i)$, where the parameters $\mu_i$ are chosen to obey specific rules (see \cite{nesterov2013introductory} for more details).
% \footnote{This is one of the most famous sequence of $\mu$ values; in \cite{nesterov2013introductory}, Nesterov proposes and analyzes other $\mu$ sequences that also lead to similar convergence rates.} rule:
%\begin{scriptsize}$t_0 = 0, ~~ t_i =(1 + \sqrt{1 + 4t_{i-1}^2})/2$\end{scriptsize}, and \begin{scriptsize}$\mu_i = \tfrac{1 - t_i}{t_{i+1}}$\end{scriptsize}. 
We borrow this momentum formulation, and we study how constant $\mu$ selections behave in non-convex problem formulations, such as in \eqref{eq:factobj}.
\emph{We note that the theory and algorithmic configurations in \cite{nesterov1983method} do not generalize to non-convex problems,} which is one of the contributions of this work.
Albeit being a non-convex problem, we show that \texttt{MiFGD} converges at an accelerated linear rate around a neighborhood of the optimal value, akin to convex optimization problems \cite{nesterov1983method}.

An important observation is that the factorization $\rho = UU^\dagger$ is not unique. For instance, suppose that $U^\star$ is an optimal solution for~\eqref{eq:factobj}; then, for any rotation matrix $R\in\mathbb{C}^{r\times{r}}$ satisfying $RR^\dagger=I$, the matrix $\widehat{U}=U^\star{R}$ is also optimal for \eqref{eq:factobj}. 
\footnote{To see this, observe that $\rho^\star = U^\star U^{\star\dagger} = U^\star I U^{\star\dagger} = U^\star R R^\dagger U^{\star\dagger} = \widehat{U}\widehat{U}^\dagger$.} 
To resolve this ambiguity, we use the distance between a pair of matrices as the minimum distance $\min_{R \in \mathcal{O}} \left \| U  - U^\star R \right\|_F$ up to rotations, where $\mathcal{O}=\{R\in\mathbb{C}^{r\times{r}}\;|\;RR^\dagger=I\}$.
In words, we want to track how close an estimate $U$ is to $U^\star$, up to the minimizing rotation matrix.

%We do not know \emph{a priori} $\tau(\rho^\star)$ and $\texttt{srank}(\rho^\star)$ to compute $\gamma'$, but they can be approximated depending on the problem at hand; \emph{e.g.}, in the quantum state tomography case, the rank could be $r=1$ (for pure quantum states), and we know apriori that $\tau(\rho^\star) = \texttt{srank}(\rho^\star) = 1$, by construction.
%\medskip
%\noindent \textbf{Main Theorem.}
We are now ready to state the main theorem regarding the \texttt{MiFGD} algorithm: 
%
% \begin{theorem}[\texttt{MiFGD} convergence rate (Informal)]{\label{thm_inf:00}}
% Assume that $\mathcal{A}$ satisfies the RIP for some constant $0 < \delta_{2r} < 1$.
% Let $y = \mathcal{A}(\rho^\star)$ denote the set of measurements, by measuring the quantum density matrix $\rho^\star$. 
% Given a good initialization point $U_0$ and setting step size $\eta$ and momentum $\mu$ appropriately,
% \texttt{MiFGD} converges with a linear rate to a region---with radius that depends on $O(\mu)$---around the global solution $\rho^\star$.
% \end{theorem}
\begin{theorem}[\texttt{MiFGD} convergence rate (Informal)]{\label{thm_inf:00}}
\textbf{Assume that $\mathcal{A}$ satisfies the RIP for some constant $0 < \delta_{2r} < 1$.
Let $y = \mathcal{A}(\rho^\star)$ denote a data set obtained by measuring a quantum system in a state $\rho^\star$. 
Given a good initialization point $U_0$ and setting step size $\eta$ and momentum $\mu$ appropriately,
\texttt{MiFGD} converges with an accelerated linear rate to a region---with radius that depends on $O(\mu)$---around $\rho^\star$.
% the global solution $\rho^\star$.
}
\end{theorem}
``Accelerated linear rate'' intuitively means that \texttt{MiFGD} (provably) enjoys smaller contraction factor compared to that of vanilla FGD.
We refer to Theorem~\ref{thm:00} of the Methods section for a formal statement. 
There, we state the conditions under which the simple \texttt{MiFGD} recursion in Eqs. \eqref{eq:MiFGD1}-\eqref{eq:MiFGD2} has an accelerated linear convergence rate in iterate distance, up to a constant error level proportional to the  momentum parameter $\mu$. 
The theorem assumes that the observations are noiseless; that is, the observed  data is $y = \mathcal{A}(\rho^\star)$, where $\rho^\star$ is the state of the system. 
Nevertheless, our %numerical 
experiments suggest that \texttt{MiFGD} is robust to statistical errors and noise in the data. 
The formal analysis of robustness to noisy data can be derived from our analysis and considered future work; here, for clarity, we consider this work as the basis for that analysis.

\subsection{Experimental setup}

\subsection{$\rho^\star$ density matrices and quantum circuits\footnote{The content in this subsection is implemented in the \texttt{states.py} component of our complementary software package.}}

In our numerical and real experiments, we have considered (different subsets of) the following $n$-qubit pure quantum states: \vspace{-0.2cm}
\begin{enumerate}
    \item The (generalized) GHZ state: \vspace{-0.2cm}
    \begin{align*}
    |\texttt{GHZ}(n) \rangle = \frac{|0\rangle^{\otimes n} + |1\rangle^{\otimes n}}{\sqrt{2}}, ~~n > 2.
\end{align*} \vspace{-0.2cm}
\item The (generalized) GHZ-minus state: \vspace{-0.2cm}
\begin{align*}
    |\texttt{GHZ}_{-}(n) \rangle = \frac{|0\rangle^{\otimes n} - |1\rangle^{\otimes n}}{\sqrt{2}}, ~~n > 2.
\end{align*} \vspace{-0.2cm}
    \item The Hadamard state:
    \begin{align*}
    |\texttt{Hadamard}(n) \rangle = \left(\frac{|0\rangle + |1\rangle}{\sqrt{2}}\right)^{\otimes n}.
\end{align*}  \vspace{-0.2cm}
\item A random state $|\texttt{Random}(n) \rangle$. \vspace{-0.2cm}
\end{enumerate}
We have implemented these states (on the IBM quantum simulator and/or the IBM's QPU) using the following circuits. The GHZ state $|\texttt{GHZ}(n) \rangle$ is generated by applying the Hadamard gate to one of the qubits, and then applying  $n-1$ CNOT gates between this qubit (as a control) and the remaining $n-1$ qubits (as targets). 
The GHZ-minus state $|\texttt{GHZ}_{-}(n) \rangle$ is generated by applying the $X$ gate to one of the qubits (e.g., the first qubit) and the Hadamard gate to the remaining $n-1$ qubits, followed by applying  $n-1$ CNOT gates between the first qubit (as a target) and the other $n-1$ qubits (as controls). Finally, we apply the Hadamard gate to all of the qubits. 
The Hadamard state $|\texttt{Hadamard}(n) \rangle$ is a separable state, and it is generated by applying the Hadamard gate to all of the qubits. 
The random state $|\texttt{Random}(n) \rangle$ is generated by a random quantum gate selection: In particular, for a given circuit depth, we uniformly select among generic single-qubit rotation gates with 3 Euler angles, and controlled-X gates, for every step in the circuit sequence. 
For the rotation gates, the qubits involved are selected uniformly at random, as well as the angles from the range $[0, 1]$.
For the controlled-X gates, the source and target qubits are also selected uniformly at random.

We generically denote the density matrix that correspond to pure state $\ket{\psi}$ as $\rho^\star=\ket{\psi}\bra{\psi}$. 
For clarity, we will drop the bra-ket notation when we refer to $|\texttt{GHZ}(n) \rangle$, $|\texttt{GHZ}_{-}(n) \rangle$, $|\texttt{Hadamard}(n) \rangle$ and $|\texttt{Random}(n) \rangle$. 
While the density matrices of the $\texttt{GHZ}(n)$ and $\texttt{GHZ}_{-}(n)$ are sparse in the $\{\ket{0},\ket{1}\}^n$ basis, the density matrix of $\texttt{Hadamard}(n)$ state is fully-dense in this basis, and the sparsity of the density matrix that of $\texttt{Random}(n)$ may be different form one state to another.

\subsection{Measuring quantum states\footnote{The content in this subsection is implemented in the \texttt{measurements.py} component of our complementary software package.}}

\textit{The quantum measurement model.} 
In our experiments (both synthetic and real) we measure the qubits in the Pauli basis \cite{achuthan1958general}.\footnote{This is the non-commutative analogue of the Fourier basis, for the case of sparse vectors \cite{rudelson2008sparse, candes2006near}.}
A Pauli basis measurement on an $n$-qubit system has $d=2^n$ possible outcomes. The Pauli basis measurement is uniquely defined by the {\it measurement setting}. A Pauli measurement is a string of $n$ letters $\alpha:=(\alpha_1,\alpha_2,\ldots,\alpha_n)$ such that $\alpha_k\in \{x,y,z\}$ for all $k\in[n]$. 
% %%%%%%%
% \jlk{Maybe $k \in \{1, \dots, n\}$ instead?}
% %%%%%%%
Note that there are at most $3^n$ distinct Pauli strings. To define the Pauli basis measurement that associated with a given measurement string $\alpha$, we first define the the following three bases on $\mathbb{C}^{2 \times 2}$:
\begin{align*}
   {\cal B}_x &= \left\{\ket{x,0}:=\frac1{\sqrt{2}}(\ket{0}+\ket 1), \;\ket{x,1}:=\frac1{\sqrt{2}}(\ket{0}-\ket 1)\right\},\\
   {\cal B}_y &= \left\{\ket{y,0}:=\frac1{\sqrt{2}}(\ket{0}+ i\ket 1),\;\ket{y,1}:=\frac1{\sqrt{2}}(\ket{0}- i\ket 1)\right\},\\ 
   {\cal B}_z &= \left\{\ket{z,0}:=\ket{0},\;\ket{z,1}:=\ket{1}\right\}.
\end{align*}
These are the eigenbases of the single-qubit Pauli operators,  $\sigma_x,\sigma_y$, and $\sigma_z$, whose $2\times2$ matrix representation is given by:
\begin{align*}
   \sigma_x = \begin{bmatrix} 0 & 1 \\ 1 & 0 \end{bmatrix}, \quad \sigma_y = \begin{bmatrix} 0 & -i \\ i & 0 \end{bmatrix}, \quad \sigma_z = \begin{bmatrix} 1 & 0 \\ 0 & -1 \end{bmatrix}.
\end{align*}
Given a Pauli setting $\alpha$, the Pauli basis measurement $\Pi_\alpha$ is defined by the $2^n$ projectors:
\begin{align*}
   \Pi_\alpha = \left\{\ket{v_\ell^{(\alpha)}}\bra{v_\ell^{(\alpha)}}=\bigotimes_{k=1}^n\ket{\alpha_k,\ell_k}\bra{\alpha_k,\ell_k}: \ell_k\in\{0,1\} \;\forall k\in[1,n] \right\},
\end{align*}
where $\ell$ denotes the bit string $(\ell_{k_1}, \ell_{k_2},\ldots, \ell_{k_n})$.
Since there are $3^n$ distinct Pauli measurement settings, there are the same number of possible Pauli basis measurements.

Technically, this set forms a positive operator-valued measure (POVM). The projectors that form $\Pi_\alpha$ are the measurement outcomes (or POVM elements) and the probability to obtain an outcome $|v_\ell^{(\alpha)}\rangle \langle v_\ell^{(\alpha)}|$ --when the state of the system is $\rho^\star$-- is given by the Born rule: $\langle v_\ell^{(\alpha)} | \rho^\star | v_\ell^{(\alpha)}\rangle = \texttt{Tr}\big(|v_\ell^{(\alpha)}\rangle \langle v_\ell^{(\alpha)}| \cdot \rho^\star \big )$. 

\medskip
\noindent \textit{The RIP and expectation values of Pauli observables.} 
Starting with the requirements of our algorithm, the sensing mapping $\mathcal{A}(\cdot): \mathbb{C}^{d \times d} \rightarrow \mathbb{R}^m$ we consider is comprised of a collection of $A_i \in \mathbb{C}^{d \times d}, ~i = 1, \dots, m$ matrices, such that $y_i = \texttt{Tr}(A_i \rho^\star)$. 
We denote the vector $(y_1,\ldots,y_m)$ by $y$.  

When no prior information about the quantum state is assumed, to ensure its (robust) recovery, one must choose a set $m$ sensing matrices $A_i$, so that $d^2$ of them are linearly independent.
%(that is, that they span the space of matrices in $\mathbb{C}^{d \times d}$). 
One example of such choice is the POVM elements of the $3^n$ Pauli basis measurements.

Yet, when it is known that the state-to-be-reconstructed is of low-rank, theory on low-rank recovery problems suggests that $A_i$ could just be ``incoherent'' enough with respect to $\rho^\star$ \cite{gross2011recovering}, so that recovery is possible from a limited set of measurements, i.e., with $m\ll d^2$. In particular, it is known~\cite{liu2011universal, gross2011recovering, gross2010quantum} that if the sensing matrices correspond to random {\it Pauli monomials}, then $m = O\left(r \cdot d \cdot \text{poly}\left(\log d\right)\right)$ $A_i$'s are sufficient for a successful recovery of $\rho^\star$, using convex solvers for \eqref{eq:obj}.\footnote{The main difference between \cite{gross2011recovering, gross2010quantum} and \cite{liu2011universal} is that the former guarantees recovery for almost all choices of 
$m = O\left(r \cdot d \cdot \text{poly} \left(\log d\right)\right)$ random Pauli monomials, while the latter proves that there exists a \emph{universal} set of $m = O\left(r \cdot d \cdot \text{poly} \left(\log d\right)\right)$ Pauli monomials $A_i$ that guarantees successful recovery.}  A Pauli monomial $P_i$ is an operator in the set $P_i\in\{\mathbb{1},\sigma_x,\sigma_y,\sigma_z\}^{\otimes n}$, that is, an $n$-fold tensor product of single-qubit Pauli operators (including the identity operator). For convenience we relabel the single-qubit Pauli operators as $\sigma_0:=\mathbb{1},\sigma_1:=\sigma_x,\sigma_2:=\sigma_y$, and $\sigma_3:=\sigma_z$, so that we can also write $P_i=\bigotimes_{k=1}^n\sigma_{i_k}$ with $i_k\in \{0,\dots,3\}$ for all $k\in[n]$.  These results~\cite{liu2011universal, gross2011recovering, gross2010quantum} are feasible since the Pauli-monomial-based sensing map $\mathcal{A}(\cdot)$ obeys the RIP property, as in Definition \ref{def:rip}.\footnote{In particular, the RIP is satisfied for the sensing mechanisms that obeys $\left(\mathcal{A}(\rho^\star)\right)_i = \tfrac{d}{\sqrt{m}} \texttt{Tr}(A_i^* \rho^\star)$, $i = 1, \dots, m$. Further, the case considered in \cite{liu2011universal} holds for a slightly larger set than the set of rank-$r$ density matrices: for all $\rho \in \mathbb{C}^{d \times d}$ such that $\|\rho\|_* \leq \sqrt{r} \|\rho\|_F$.}
For the rest of the text, we will use the term ``\emph{Pauli expectation value}'' to denote $\texttt{Tr}(A_i \rho^\star) = \texttt{Tr}(P_i \rho^\star)$.

\medskip
\noindent \textit{From Pauli basis measurements to Pauli expectation values.}
While the theory for compressed sensing was proven for Pauli expectation values, in real QPUs, experimental data is obtained from Pauli basis measurements. 
Therefore, to make sure we are respecting the compressed sensing requirements on the sensing map, we follow this protocol: \vspace{-0.2cm}
\begin{enumerate}
    \item[$i)$] We sample $m = O\left(r \cdot d \cdot \text{poly}\left(\log d\right)\right)$ or $m = \texttt{measpc} \cdot d^2$ Pauli monomials uniformly over $\{\sigma_i\}^{\otimes n}$ with $i \in \{0, \dots, 3\}$, where $\texttt{measpc} \in [0, 1]$ represents the percentage of measurements out of full tomography. \vspace{-0.2cm}
    \item[$ii)$] For every monomial, $P_i$, in the generated set, we identify an experimental setting $\alpha(i)$ that corresponds to the monomial. There, qubits, for which their Pauli operator in $P_i$ is the identity operator, are measured, without loss of generality, in the $\sigma_3$ basis. For example, for $n=3$ and  $P_i=\sigma_0\otimes\sigma_1\otimes\sigma_1$, we identify the measurement setting $\alpha(i)=(z,x,x)$. \vspace{-0.2cm}
    \item[$iii)$] We measure the quantum state in the Pauli basis that corresponds to  $\alpha(i)$, and record the outcomes.
\end{enumerate}
To connect the measurement outcomes to the expectation value of the Pauli monomial, we use the relation:
\begin{align}\label{eq:basis2monomial}
\texttt{Tr}(P_i \rho^\star) = \sum_{\ell \in \{0,1\}^n } (-1)^{\chi_{_{f(\ell)}}} \cdot \texttt{Tr}\left(|v_\ell^{(\alpha(i))}\rangle \langle v_\ell^{(\alpha(i))}| \cdot \rho^\star \right),
\end{align}
where $f(\ell):\{0,1\}^n\rightarrow\{0,1\}^n$ is a mapping that takes a bit string $\ell$ and returns a new bit string $\tilde{\ell}$ (of the same size) such that $\tilde{\ell}_k=0$ for all $k$'s for which $i_k=0$ (that is, the locations of the  identity operators in $P_i$), and  $\chi_{\tilde{\ell}}$ is the parity of the bit string $\tilde{\ell}$.

\subsection{Algorithmic setup   }

%\added[id=GK]{[Description of the overall structure of the code can go here]}

%\added[id=GK]{[(Optional) description of the optimizations in encoding Pauli projectors (construction, application, storage) can also go here and then move to its own section]}

In our implementation, we explore a number of control parameters, including the maximum number of iterations \texttt{maxiters}, the learning rate $\eta$, the relative error from successive state iterates \texttt{reltol}, the acceleration parameter $\mu$, the percentage  of the complete set of measurements (i.e. over all possible Pauli monomials) \texttt{measpc}, and the \texttt{seed}. 
In the sequel experiments we set $\texttt{maxiters}=1000$, 
$\eta=10^{-3}$, 
$\texttt{reltol}=5\times 10^{-4}$ unless stated differently.
Regarding acceleration, $\mu=0$ when acceleration is muted; we experiment over the range of values $\mu \in \{ \tfrac{1}{8}, \tfrac{1}{4}, \tfrac{1}{3}, \tfrac{3}{4} \}$ when investigating the acceleration effect, beyond the theoretically suggested $\mu^\star$. 
In order to explore the dependence of our approach on the number of measurements available, \texttt{measpc} varies over the set of $\{5\%, 10\%, 15\%, 20\%, 40\%, 60\%\}$; \texttt{seed} is used for differentiating repeating runs with all other parameters kept fixed.\footnote{\texttt{maxiters} is \texttt{num\_iterations} in the code; also \texttt{reltol} is \texttt{relative\_error\_tolerance},
 \texttt{measpc} is \texttt{complete\_measurements\_percentage}.}

Denoting $\widehat{\rho}$ the estimate of $\rho^\star$ by \texttt{MiFGD}, we report on outputs including: \vspace{-0.2cm}
\begin{itemize}[leftmargin=0.6cm]
  \item
    The evolution with respect to the distance between $\widehat{\rho}$  and $\rho^\star$:
    $\| \widehat{\rho} - \rho^\star \|_F$, for various $\mu$'s. \vspace{-0.2cm}% (target error list plots).
  \item
    The number of iterations to reach $\texttt{reltol}$ to $\rho^\star$ for various $\mu$'s. \vspace{-0.2cm}
  \item
    The fidelity of $\widehat{\rho}$, defined as $\text{Tr}\big(\rho^\star \widehat{\rho} \big)$ (for rank-1 $\rho^\star$), 
    %\gdk{Note: More general definitions of fidelity are given in the sequel?}
    as a function of the acceleration parameter $\mu$ in the default set. \vspace{-.2cm} %(fidelity list plots).
  \end{itemize}
  In our plots, we sweep over our default set of \texttt{measpc} values, repeat $5$ times for each individual setup, varying supplied seed, and depict their $25$-, $50$- and $75$-percentiles.

\subsection{Experimental setup on quantum processing unit (QPU)}
%We push forward bigger real quantum systems, by developing scalable compressed sensing reconstruction tools and testing them on real quantum hardware. 
%  One of the important differentiators of this work is the reconstruction of quantum states by first realizing them in real quantum processors - IBM Quantum devices - and subsequently probing the hardware devices with suitable pulses for collecting the measurements to serve as input to tomography runs.
%In our experiments, 
We show empirical results on 6- and 8-qubit real data, obtained on the 20-qubit IBM QPU {\texttt{ibmq\_boeblingen}}. 
The layout/connectivity of the device is shown in Figure \ref{fig:real_circuit}. 
The 6-qubit data was from qubits $[0,1,2,3,8,9]$, and the 8-qubit data was from $[0,1,2,3,8,9,6,4]$. 
The $T_1$ coherence times are $[39.1, 75.7, 66.7, 100.0, 120.3, 39.2, 70.7, 132.3]$ $\mu s$, and $T_2$ coherence times are $[86.8,94.8,106.8,63.6,156.5,66.7,104.5,134.8]$ $\mu s$. The circuit for generating 6-qubit and 8-qubit GHZ  states are shown in Fig \ref{fig:real_circuit}. 
The typical two qubit gate errors measured from randomized benchmarking (RB) for relevant qubits are summarized in Table \ref{tab:my_label}.

% \begin{figure}[hb!]
%     \centering
%     \includegraphics[scale=0.5]{boeblingen}
%     \caption{Caption}
%     \label{fig:my_label}
% \end{figure}

\begin{table}[hb!]
    \centering
    \begin{tabular}{ccccccc}
    \hline
         $C_0 X_1$ & $C_1 X_2$ & $C_2 X_3$ & $C_3 X_8$ & $C_8 X_9$ & $C_3 X_4$ & $C_1 X_6$ \\
    \hline
        0.0072 & 0.0062 & 0.0087 & 0.0077 & 0.0152 & 0.0167 & 0.0133
    \end{tabular}
    \caption{Two qubit error rates for the relevant gates used in generating 6-qubit and 8-qubit GHZ states on \texttt{ibmq\_boeblingen}.}
    \label{tab:my_label}
\end{table}

\begin{figure}[!htp]
\vspace{-0.3cm}
\centering
    \begin{minipage}{0.32\linewidth}
    \centering
    \includegraphics[width=0.7\linewidth]{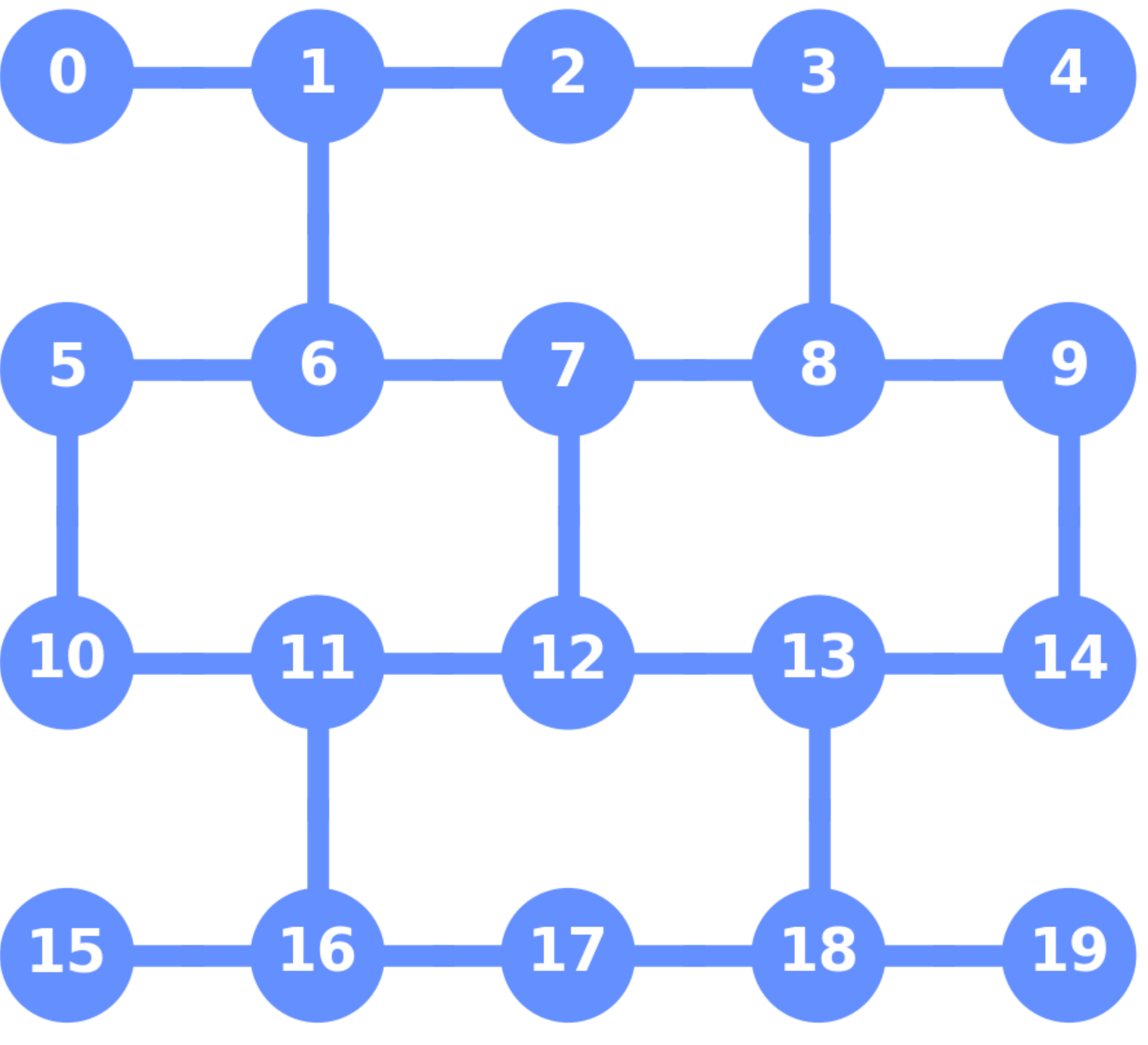}
    % \caption{Caption}
    \label{fig:real_circuit}
    \end{minipage}
    \begin{minipage}{0.32\linewidth}
    \centering
    \includegraphics[width=1.05\linewidth]{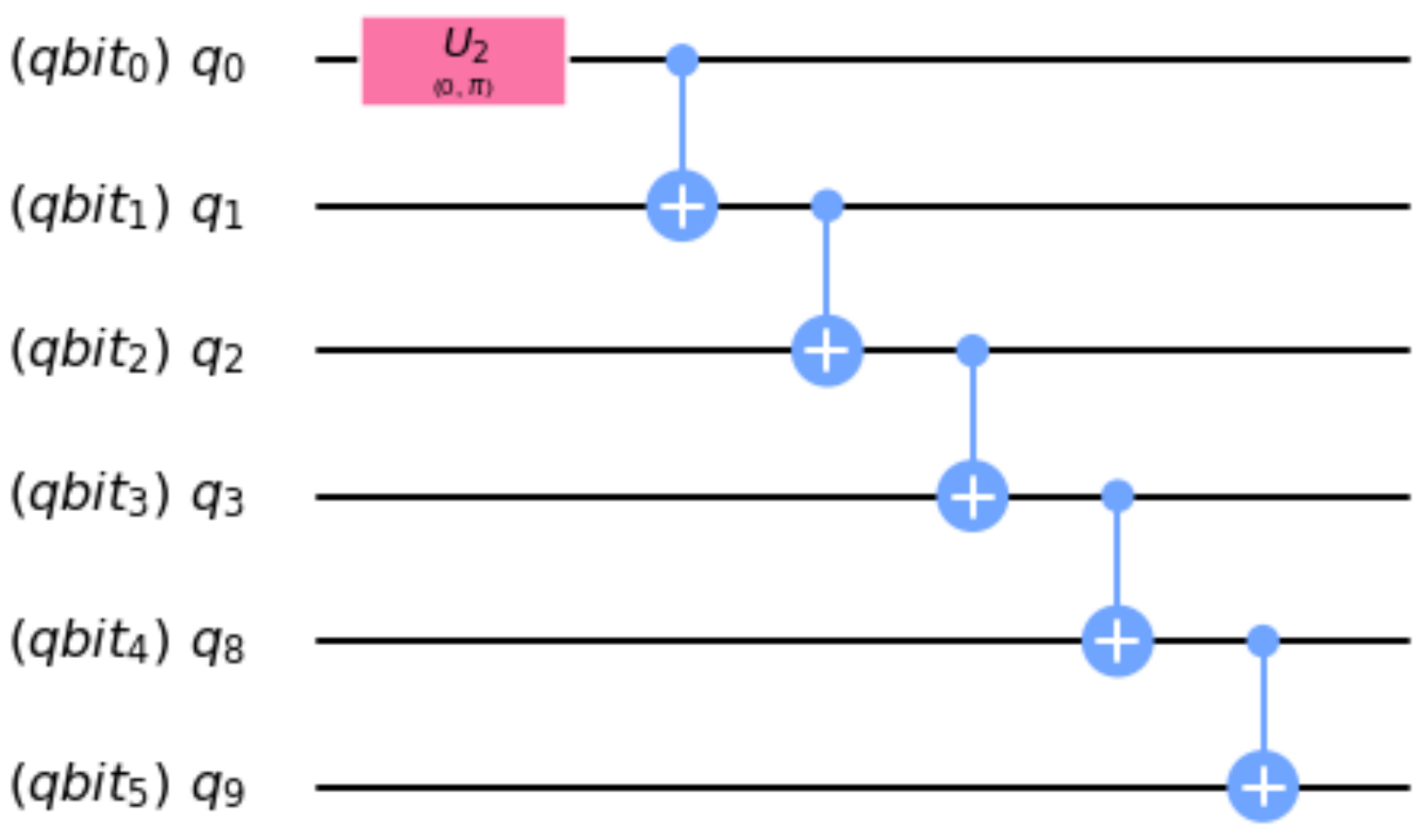}
    % \caption{Caption}
    % \label{fig:my_label}
    \end{minipage}
    \hspace{0.4cm}
    \begin{minipage}{0.32\linewidth}
    \centering
    \includegraphics[width=0.85\linewidth]{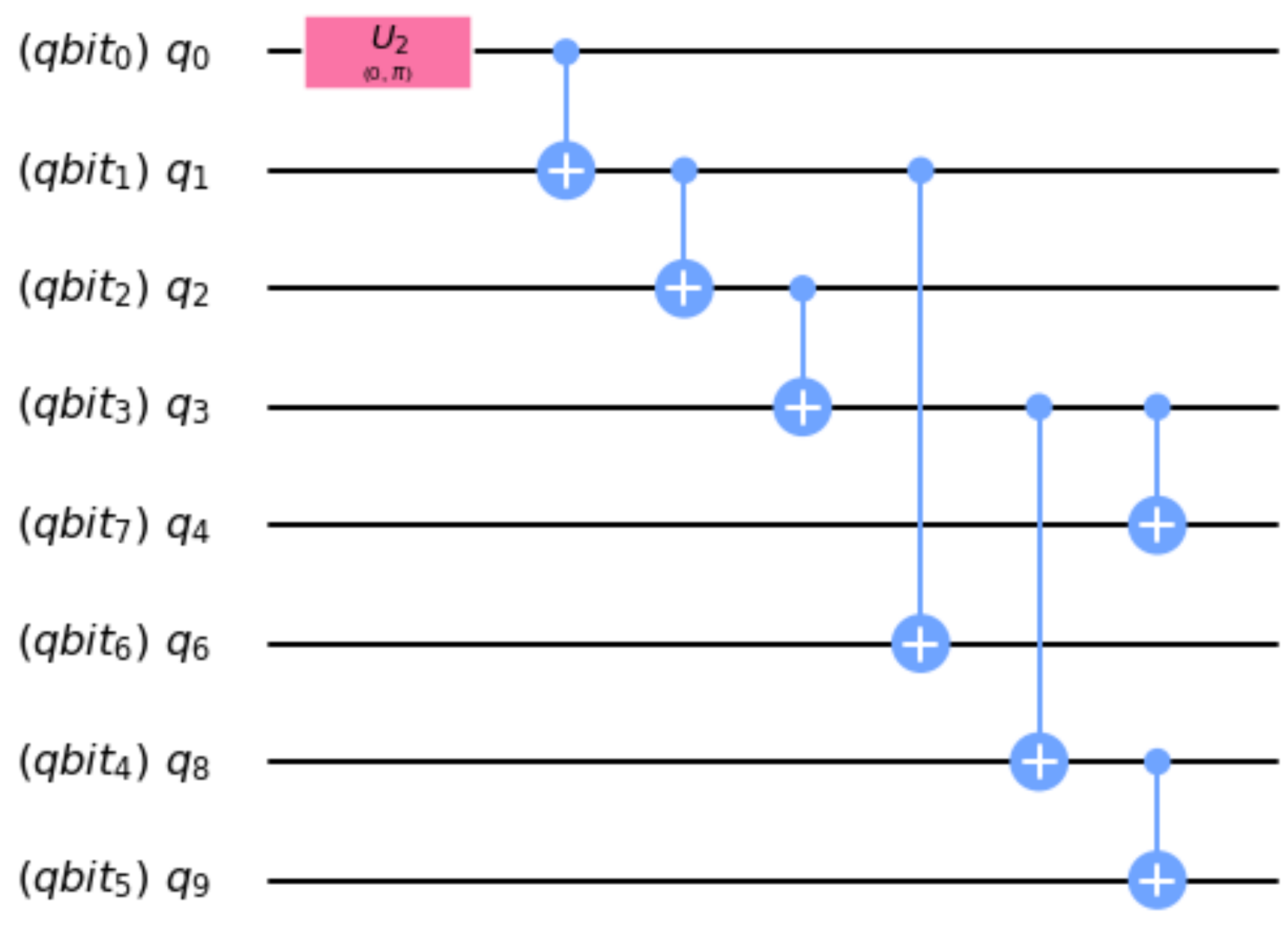}
    % \caption{Caption}
    \end{minipage}
    \vspace{-0.2cm}
    \caption{\textbf{Left panel:} Layout connectivity of IBM backend \texttt{ibmq\_boeblingen}; \textbf{Middle and right panels:} Circuits used to generate 6-qubit state (left) and 8-qubit GHZ state (right). $qbit$ refers to the quantum registers used in qiskit, and $q$ corresponds to qubits on the real device.}
    \label{fig:real_circuit}
\end{figure}
The QST circuits were generated using the tomography module in \texttt{qiskit-ignis}.\footnote{\url{https://github.com/Qiskit/qiskit-ignis}.}
For complete QST of a $n$-qubits state $3^n$ circuits are needed. 
The result of each circuit is averaged over 8192, 4096 or 2048, for different $n$-qubit scenarios.
To mitigate for readout errors, we prepare and measure all of the $2^n$ computational basis states in the computation basis 
to construct a calibration matrix $C$. 
$C$ has dimension $2^n$ by $2^n$, where each column vector corresponds to the measured outcome of a prepared basis state. 
In the ideal case of no readout error, $C$ is an identity matrix. 
We use $C$ to correct for the measured outcomes of the experiment by minimizing the function:
\begin{equation}%\label{eq:obj}
\begin{aligned}
& \min_{v^{\text{cal}} \in \mathbb{R}^d}
& & \|C v^{\text{cal}} - v^{\text{meas}}\|^2
& \text{subject to}
& & \sum_i v^{\text{cal}}_i = 1, ~v^{\text{cal}}_i \geq 0, \forall i = 1, \dots, d
\end{aligned}
\end{equation}
Here $v^{\text{meas}}$ and $v^{\text{cal}}$ are the measured and calibrated outputs, respectively. 
The minimization problem is then formulated as a convex optimization problem and solved by quadratic programming using the package \texttt{cvxopt} \cite{vandenberghe2010cvxopt}.

%\added[id=GK]{[Description of IBM Quantum quantum processors used can go here]}

\subsection{\texttt{MiFGD} on 6- and 8-qubit real quantum data}

We realize two types of quantum states on IBM QPUs, parameterized by the number of qubits $n$ for each case: these are the $\texttt{GHZ}_{-}(n)$ and $\texttt{Hadamard}(n)$ circuits. 
We collected measurements over all possible Pauli settings  by repeating the experiment for each setting  a number of times: these are the number of \texttt{shots} for each setting. 
The (circuit, number of \texttt{shots}) measurement configurations from IBM Quantum devices are summarized in Table \ref{tab:device_measurement_setups}.

\begin{wraptable}{r}{50mm}
    \vspace{-0.2cm}
    \begin{center}
    \begin{tabular}{cc}
      \toprule
      Circuit   &  \# \texttt{shots} \\
      \midrule
      $\texttt{GHZ}_{-}(6)$ &    2048 \\
      $\texttt{GHZ}_{-}(6)$ &    8192 \\
      $\texttt{GHZ}_{-}(8)$ &    2048 \\
      $\texttt{GHZ}_{-}(8)$ &    4096 \\
      $\texttt{Hadamard}(6)$ &    8192 \\
      $\texttt{Hadamard}(8)$ &    4096 \\
      \bottomrule
\end{tabular}
\end{center}
\vspace{-0.3cm}
\caption{QPU settings.}
\label{tab:device_measurement_setups}
\vspace{-0.3cm}
\end{wraptable}
%
%\textcolor{red}{Lyle: We should probably specify the value of }
%
In the Appendix, we provide target error list 
plots for the evolution of $\| \widehat{\rho} - \rho^\star \|_F^2$ for reconstructing all the settings in Table \ref{tab:device_measurement_setups}, both for real data and for simulated scenarios.
Further, we provide plots that relate the effect of momentum acceleration on the final fidelity observed for these cases. 
For clarity, in Figure \ref{fig:target_error_list_ibmq-device_measpc_20}, we summarize the efficiency of momentum acceleration, by showing the reconstruction error only for the following settings: $\texttt{maxiters}=1000$, $\eta=10^{-3}$, $\texttt{reltol}=5\times 10^{-4}$, and $\texttt{measpc} = 20\%$. 
In the plots, $\mu=0$ corresponds to the FGD algorithm in \cite{park2016finding}, $\mu^\star$ corresponds to the value obtained through our theory, while we use $\mu \in \left\{ \tfrac{1}{8}, \tfrac{1}{4}, \tfrac{1}{3}, \tfrac{3}{4} \right\}$ to  study the acceleration effect.
For $\mu^\star$, per our theory, we follow the rule $\mu^\star \approx \varepsilon / 2211$ for $\varepsilon \in (0, 1]$; see also the Methods section for more details.\footnote{For this application, $\sigma_r(\rho^\star) = 1$, $\tau(\rho^\star) = 1$, and $r=1$ by construction; we also approximated $\kappa = 1.223$, which, for user-defined $\varepsilon = 1$, results in $\mu^\star = 4.5 \cdot 10^{-4}$. Note that smaller $\varepsilon$ values result into a smaller radius of the convergence region; however, more pessimistic $\varepsilon$ values result into small $\mu$, with no practical effect in accelerating the algorithm.}
Note that, in most of the cases, the curve corresponding to $\mu=0$ is hidden behind the curve corresponding to $\mu \approx \mu^\star$.
We run each QST experiment for $5$ times for random initializations. 
We record the evolution of the $\|\widehat{\rho} - \rho^\star\|_F^2$ error at each step, and stop when the relative error of successive iterates gets smaller than $\texttt{reltol}$ or the number of iterations exceeds $\texttt{maxiters}$ (whichever happens first). To implement $\texttt{measpc} = 20\%$, we follow the description given above Eq.~\eqref{eq:basis2monomial} with $m = \texttt{measpc} \cdot d^2$. %\amir{Please verify that this is correct.}

To highlight the level of noise existing in real quantum data, in Figure \ref{fig:target_error_list_ibmq-simulator_measpc_20}, we repeat the same setting using the QASM simulator in \texttt{qiskit-aer}. This is a parallel, high performance quantum circuit simulator written in C++ that can support a variety of realistic circuit level noise models.
 
%We report results for the case of $q = 6$ qubits, \emph{i.e.}, $n = 2^n = 64$; higher dimension experiments are left for future work.
%We design $\rho^\star$ as a pure, rank-1 density matrix, such that $\rho \succeq 0$.
%More details on how we construct $\rho^\star$ are provided in Section \ref{sec:app_QST} in the appendix.
%We use the IBM Quantum Information Software Kit (QISKit) \cite{qiskit} to complete the simulations.
%We generate $y$ according to the above, where $m = \lceil{0.6 \cdot n^2} \rceil$. 
%We randomly chose informationally incomplete subsets of the measurements (here $60\%$ of all measurements) and run Algorithm~\ref{alg:algo1} assuming that the reconstructed state $\widehat{\rho}$ is of rank $r=1$ (pure state); we utilize high-performance Pauli projector operators in the iteration. 
%We test the behavior of our algorithm for $\mu \in \{0, \sfrac{1}{8}, \sfrac{1}{4}, \sfrac{1}{3}, \sfrac{3}{4} \}$ values; $\mu=0$ corresponds to the algorithm in \cite{tu2016low, bhojanapalli2016dropping}. We run each QST realization $10$ times. 

\begin{figure*}[!h]
  \begin{center}
    \includegraphics[width=0.32\textwidth]{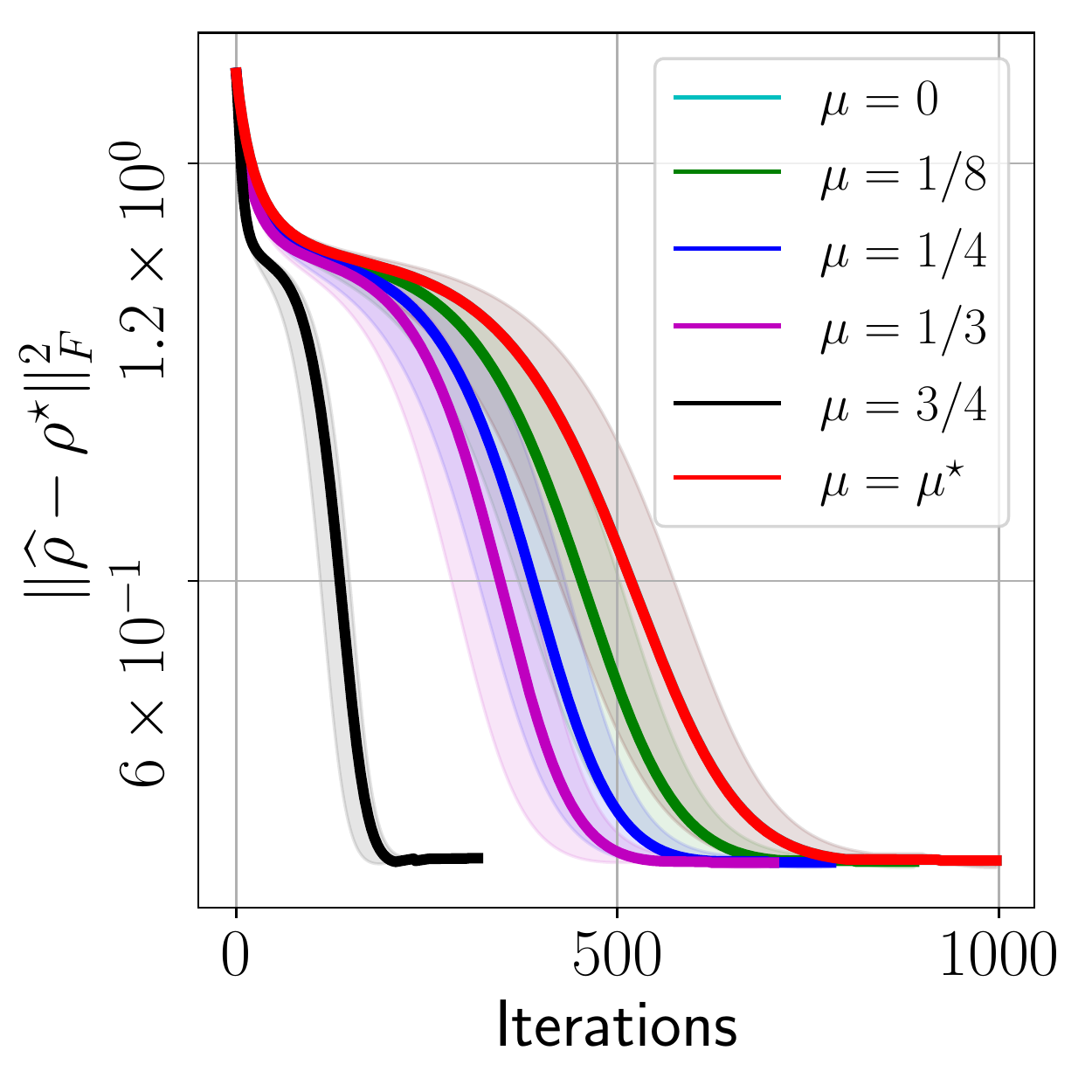}
    \includegraphics[width=0.32\textwidth]{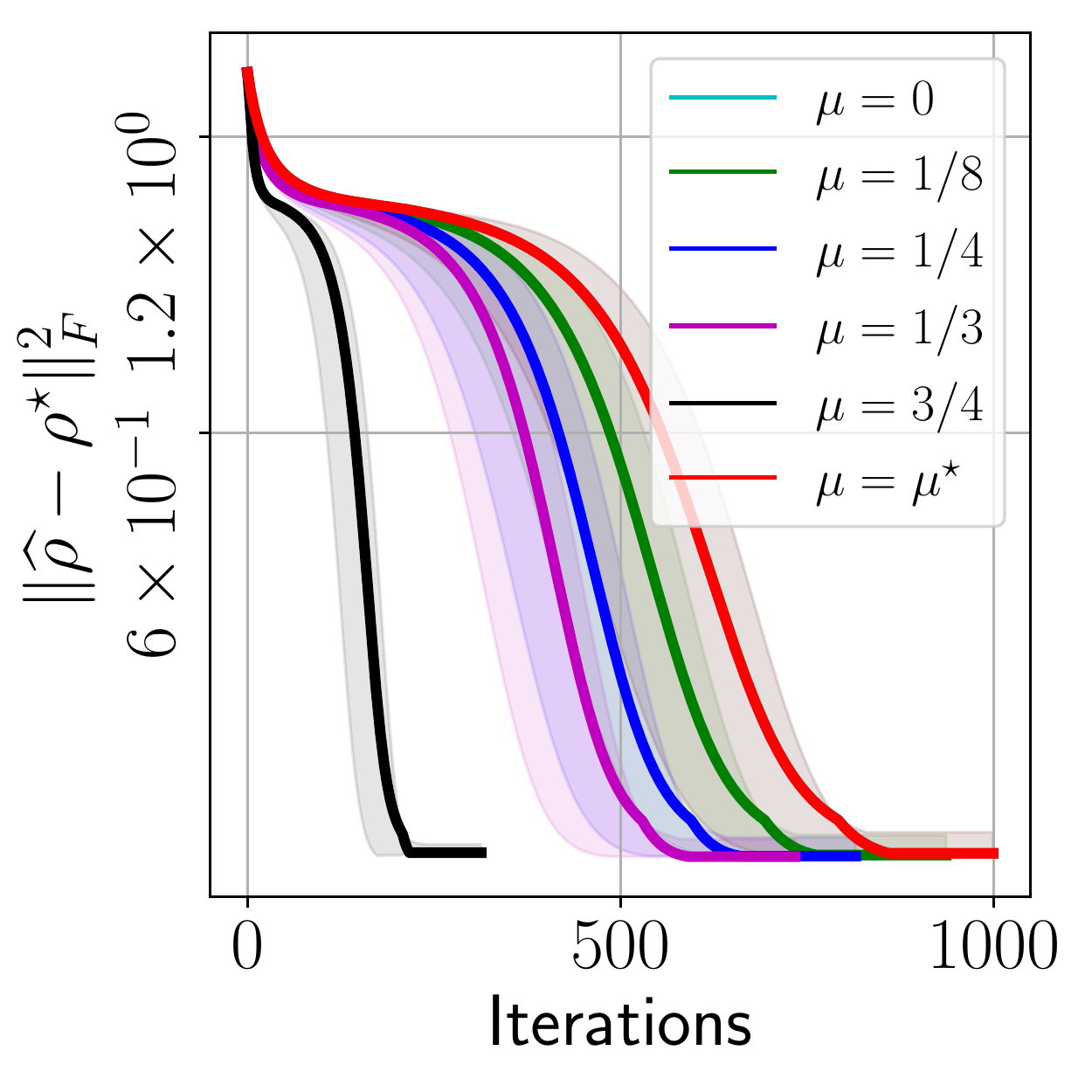} 
    \includegraphics[width=0.32\textwidth]{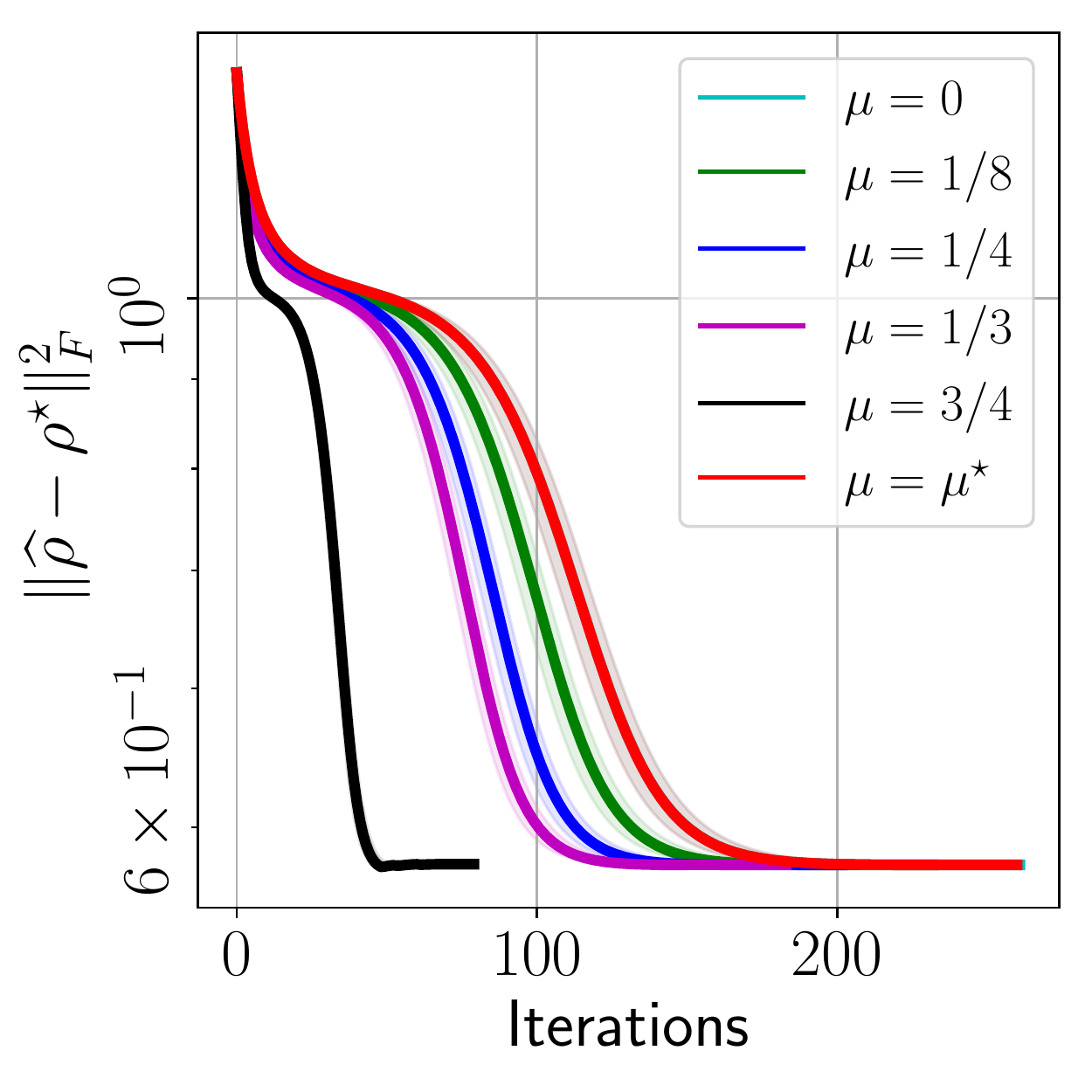}  
    \includegraphics[width=0.32\textwidth]{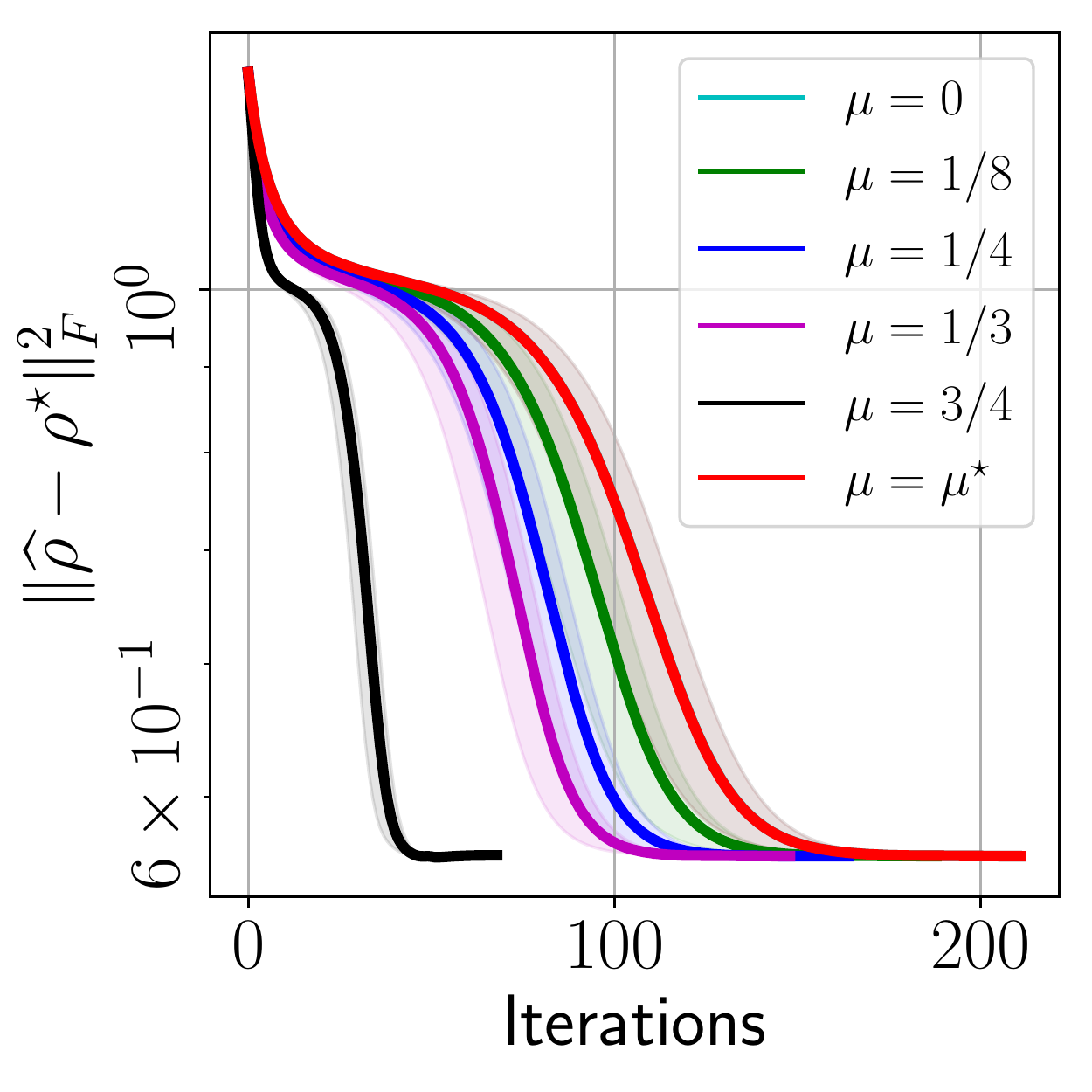}  
    \includegraphics[width=0.32\textwidth]{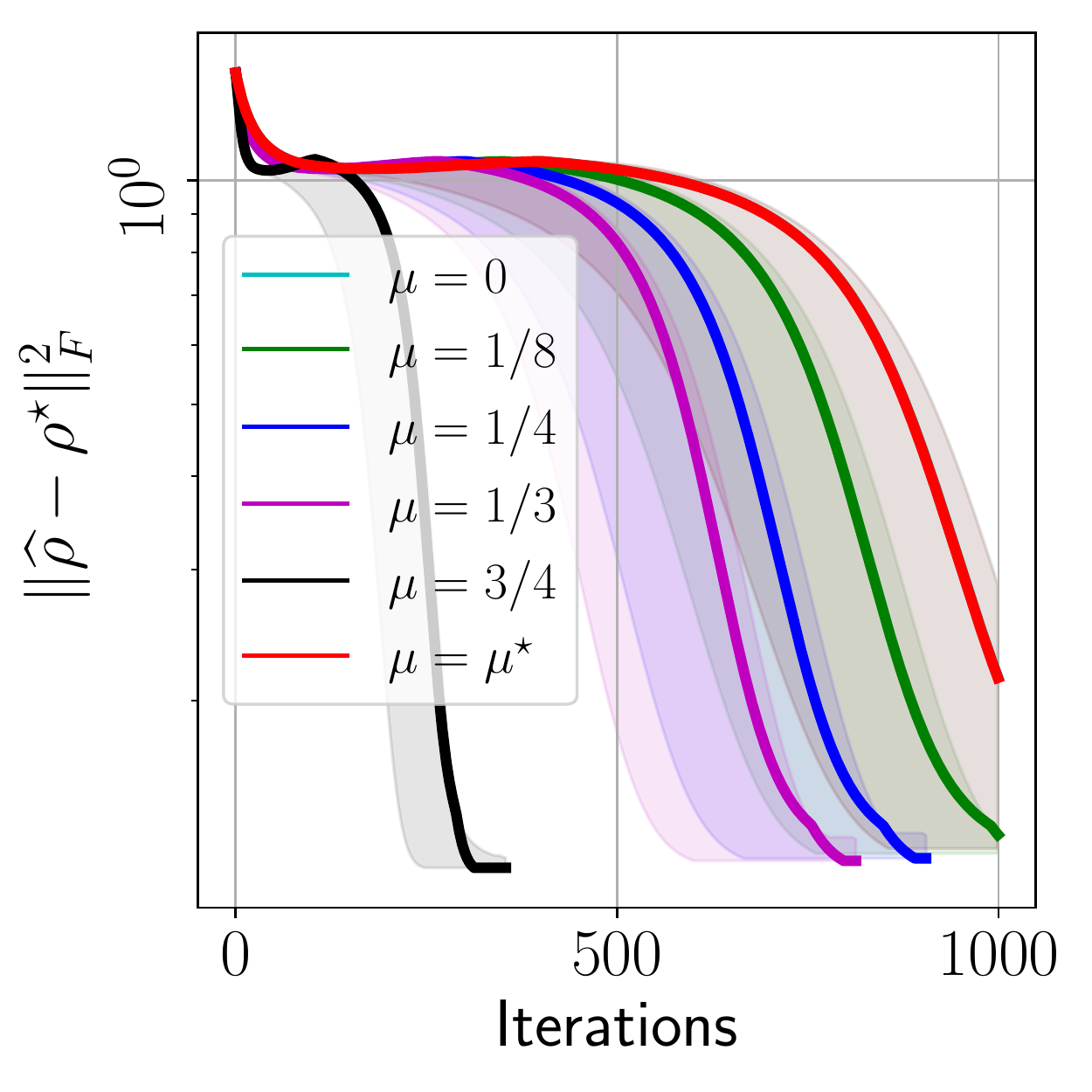}  
    \includegraphics[width=0.32\textwidth]{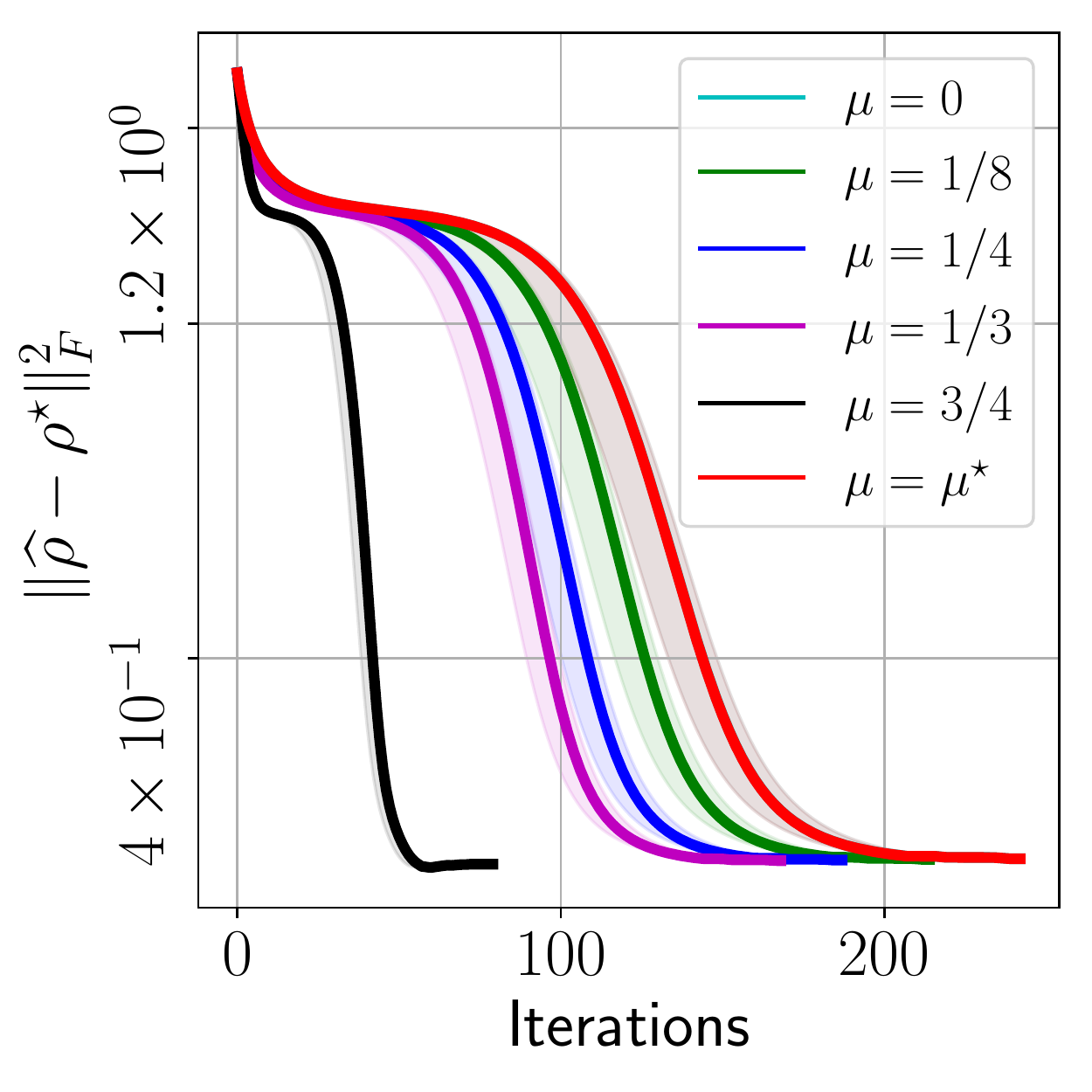}  
    \caption{
      Target error list plots $\| \widehat{\rho} - \rho^\star \|_F^2$ versus method iterations using real IBM QPU data. \textbf{Top-left}: $\texttt{GHZ}_{-}(6)$ with 2048 shots; \textbf{Top-middle}: $\texttt{GHZ}_{-}(6)$ with 8192 shots;  
      \textbf{Top-right}: $\texttt{GHZ}_{-}(8)$ with 2048 shots; \textbf{Bottom-left}: $\texttt{GHZ}_{-}(8)$ with 4096 shots/copies of $\rho^\star$; 
      \textbf{Bottom-middle}: $\texttt{Hadamard}(6)$ with 8192 shots;  
      \textbf{Bottom-right}: $\texttt{Hadamard}(8)$ with 4096 shots.
      All cases have $\texttt{measpc} = 20\%$.
      Shaded area denotes standard deviation around the mean over repeated runs in all cases. %\amir{To make it slightly more readable, I suggest to write  ``$\Vert{\hat{\rho}-\rho^\star}\Vert_F^2$ as a function of the number of iterations'' above all the figures as a title, leave $\Vert{\hat{\rho}-\rho^\star}\Vert_F^2$ and iteration lables on one figure, and remove them from all of the rest figures. Also, the legend can appear only on one figure, and in the caption we say it applies to all. Finally note that some figures the y-axis is not properly labeled (missing numbers).  Finally 2, instead of writing ``\textbf{Top-right}: $\texttt{GHZ}_{-}(8)$ with 2048 shots'' etc for all of the figures in the caption, let's just write $\texttt{GHZ}_{-}(8)$: 2048 shots on top ef the corresponding figure. Then in the caprion, we should discuss briefly, what do we learn form this plot?}
      }
    \label{fig:target_error_list_ibmq-device_measpc_20}
  \end{center}
\end{figure*}

Figure \ref{fig:target_error_list_ibmq-device_measpc_20} summarizes the performance of our proposal on different $\rho^\star$, and for different $\mu$ values on real IBM QPU data.
All plots show the evolution of $\| \widehat{\rho} - \rho^\star \|_F$ across iterations, featuring a steep dive to convergence for the largest value of $\mu$ we tested: we report that we also tested $\mu = 0$, which shows only slight worse performances than $\mu^\star$.
% %
% \textcolor{Blue}{(in most of the cases, the curve corresponding to $\mu=0$ is hidden behind the curve corresponding to $\mu=\mu^\star$.)}
% %
Figure \ref{fig:target_error_list_ibmq-device_measpc_20} highlights the universality of our approach: its performance is oblivious to the quantum state reconstructed, as long as it satisfies purity or it is close to a pure state. 
Our method does not require any additional structure assumptions in the quantum state. 

To highlight the effect of real noise on the performance of \texttt{MiFGD}, we further plot its performance on the same settings but using measurements coming from an idealized quantum simulator. 
Figure \ref{fig:target_error_list_ibmq-simulator_measpc_20} considers the exact same settings as in Figure \ref{fig:target_error_list_ibmq-device_measpc_20}.
It is obvious that \texttt{MiFGD} can achieve better reconstruction performance when data are less erroneous.
This also highlights that, in real noisy scenarios, the radius of the convergence region of \texttt{MiFGD} around $\rho^\star$ is controlled mostly by the the noise level, rather than by the inclusion of momentum acceleration.

%On the right, larger $\mu$'s not only accelerate, but also slightly improve the reconstruction quality. %, in terms of the fidelity metric. %, roughly defined as the inner product of target and converged states; fidelity equal to one implies that $\widehat{\rho} = \rho^\star$.
  
\begin{figure*}[!h]
  \begin{center}
    \includegraphics[width=0.32\textwidth]{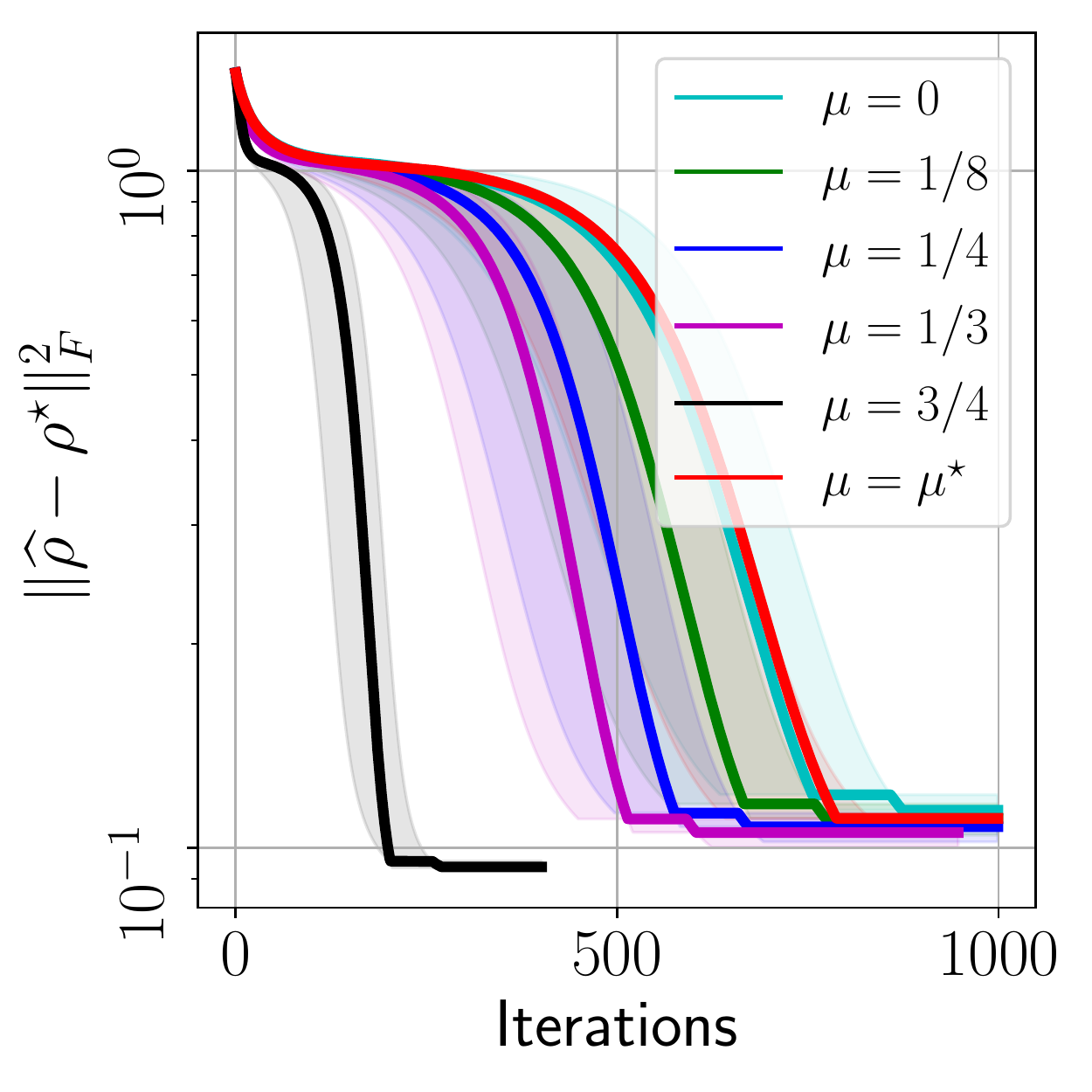}
    \includegraphics[width=0.32\textwidth]{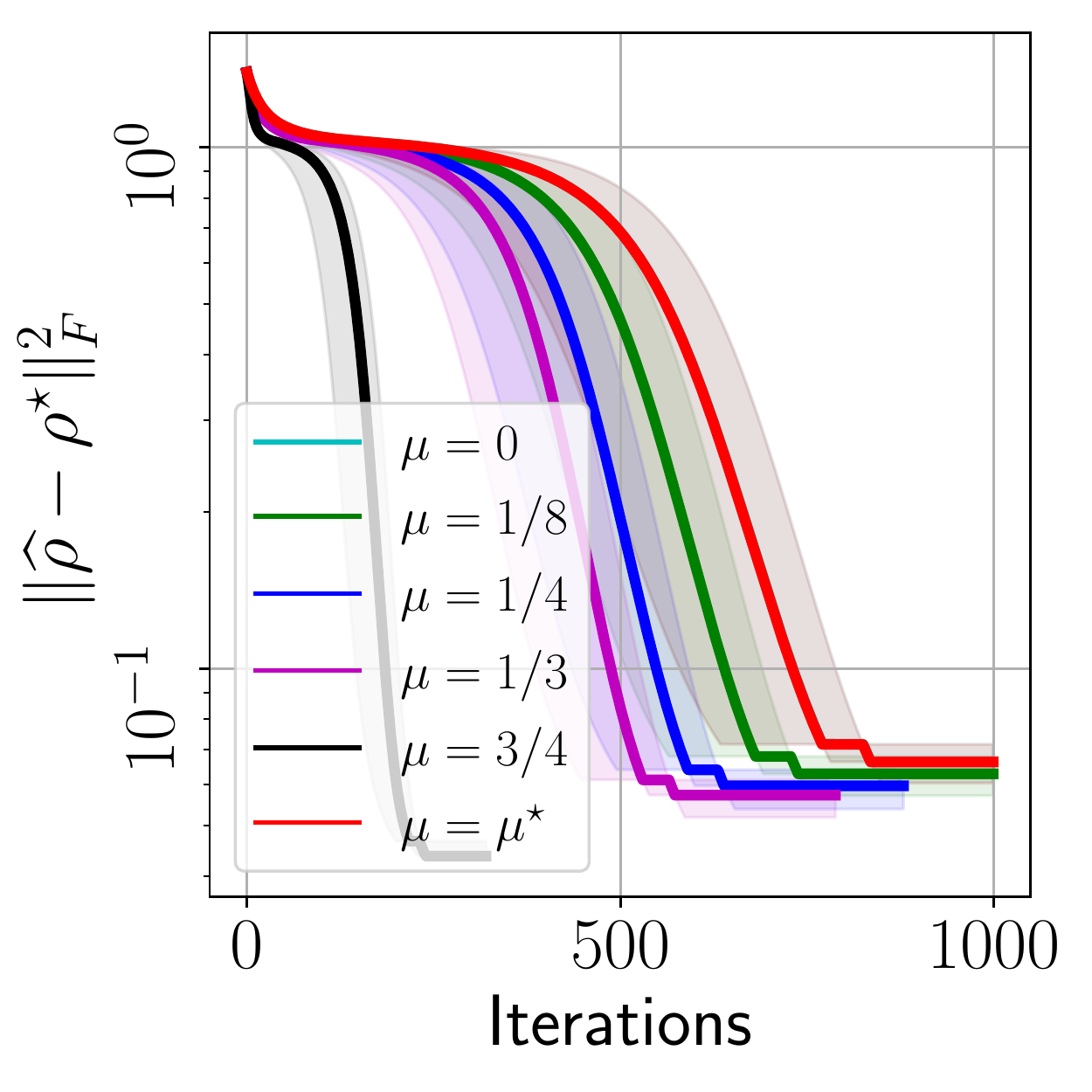}
    \includegraphics[width=0.32\textwidth]{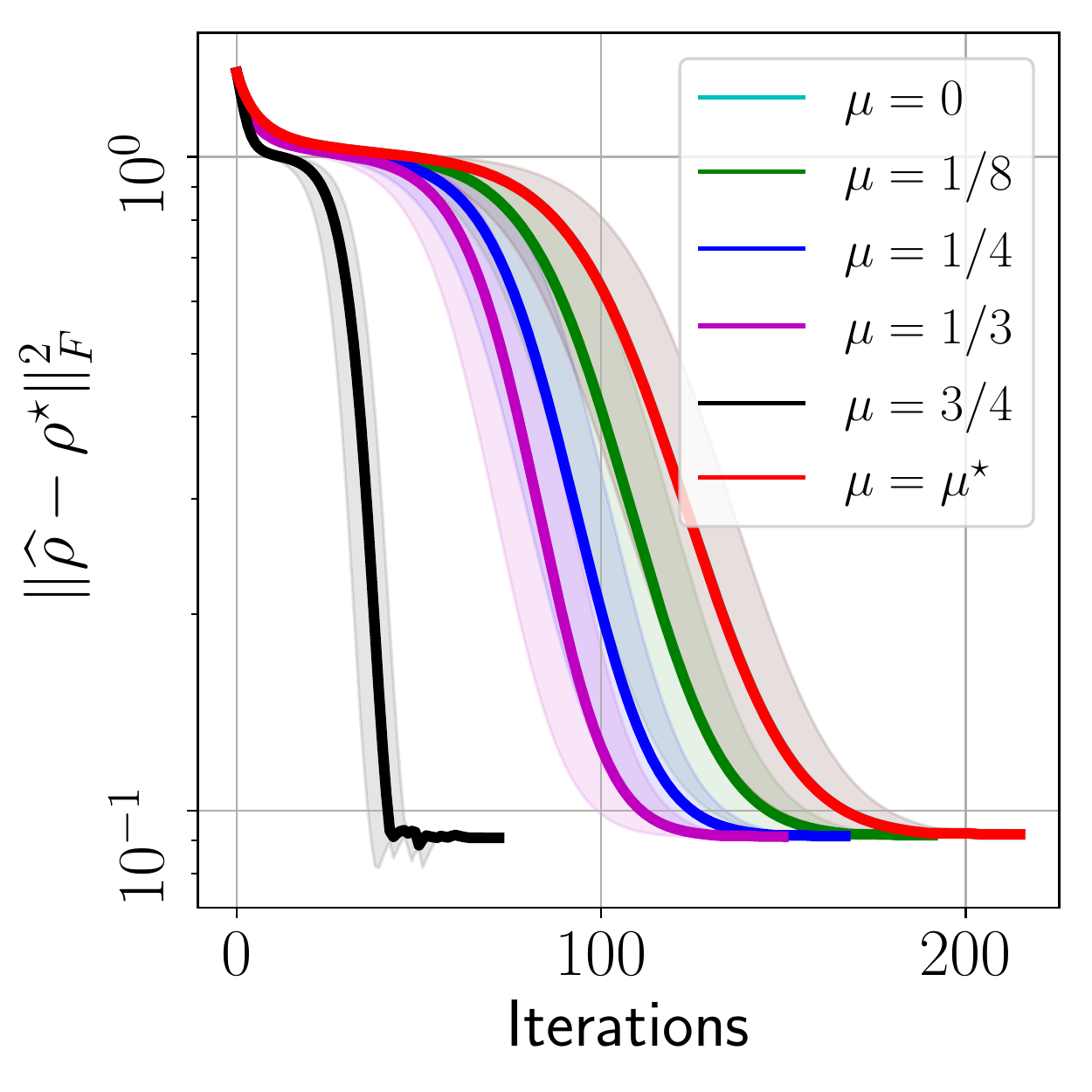}
    \includegraphics[width=0.32\textwidth]{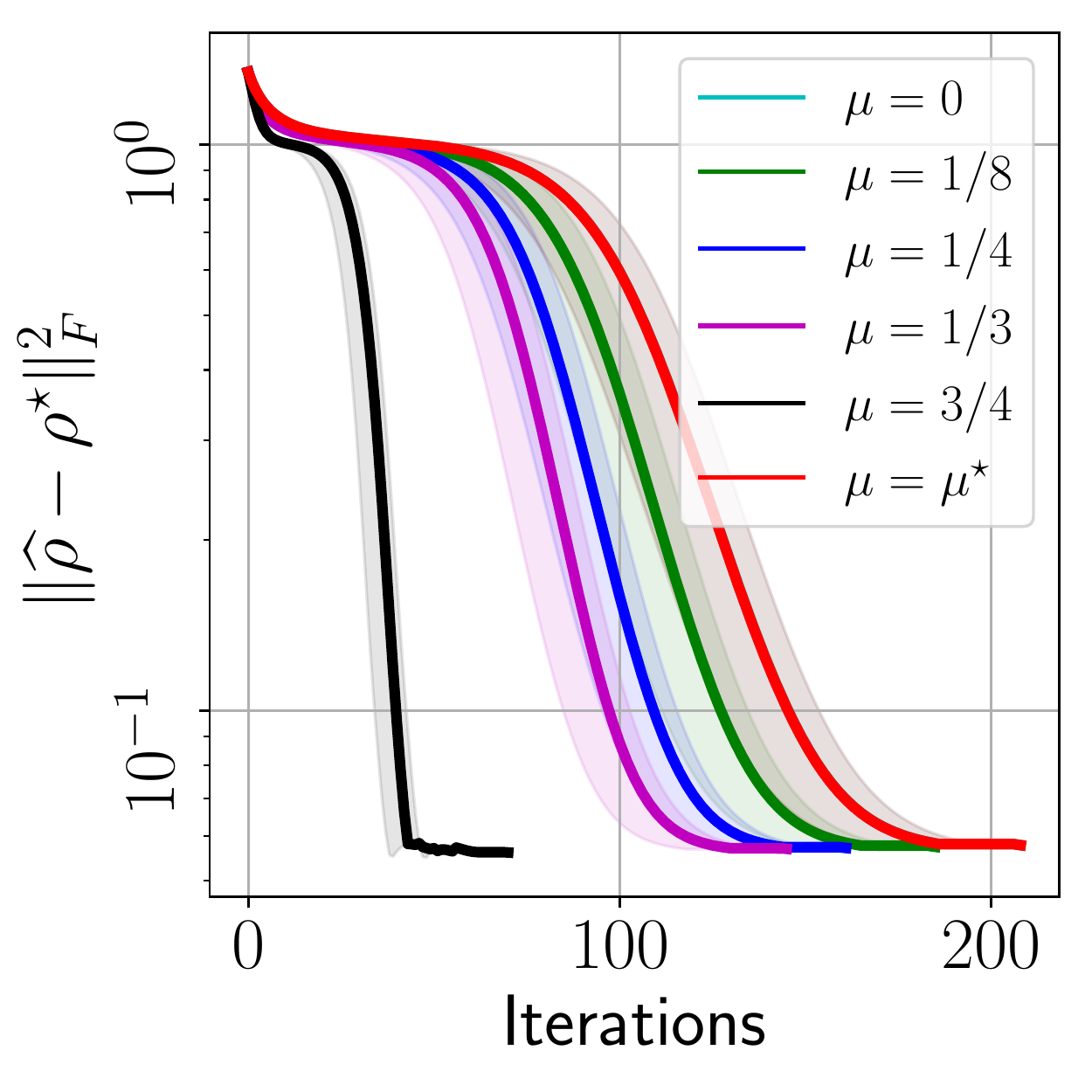} \includegraphics[width=0.32\textwidth]{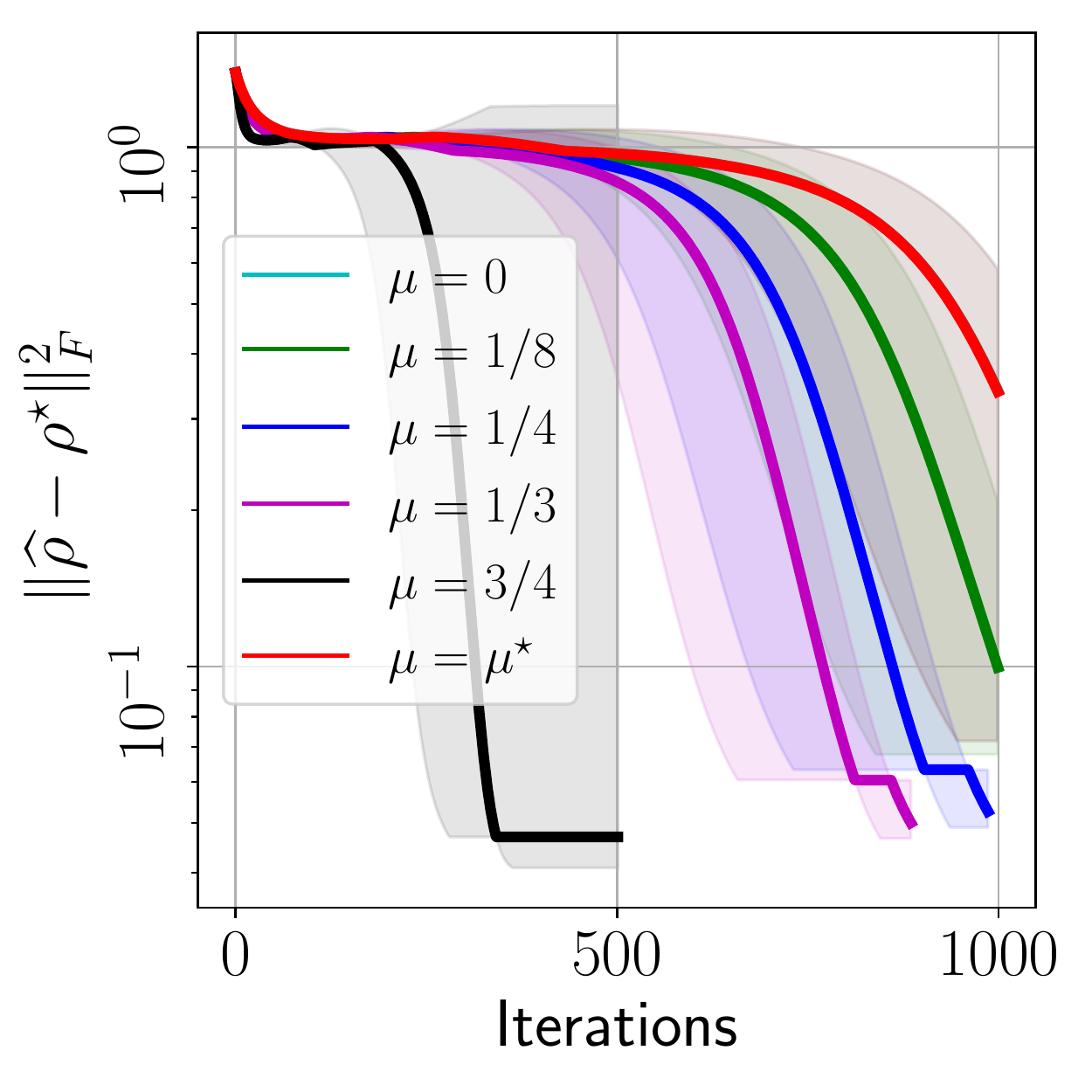} \includegraphics[width=0.32\textwidth]{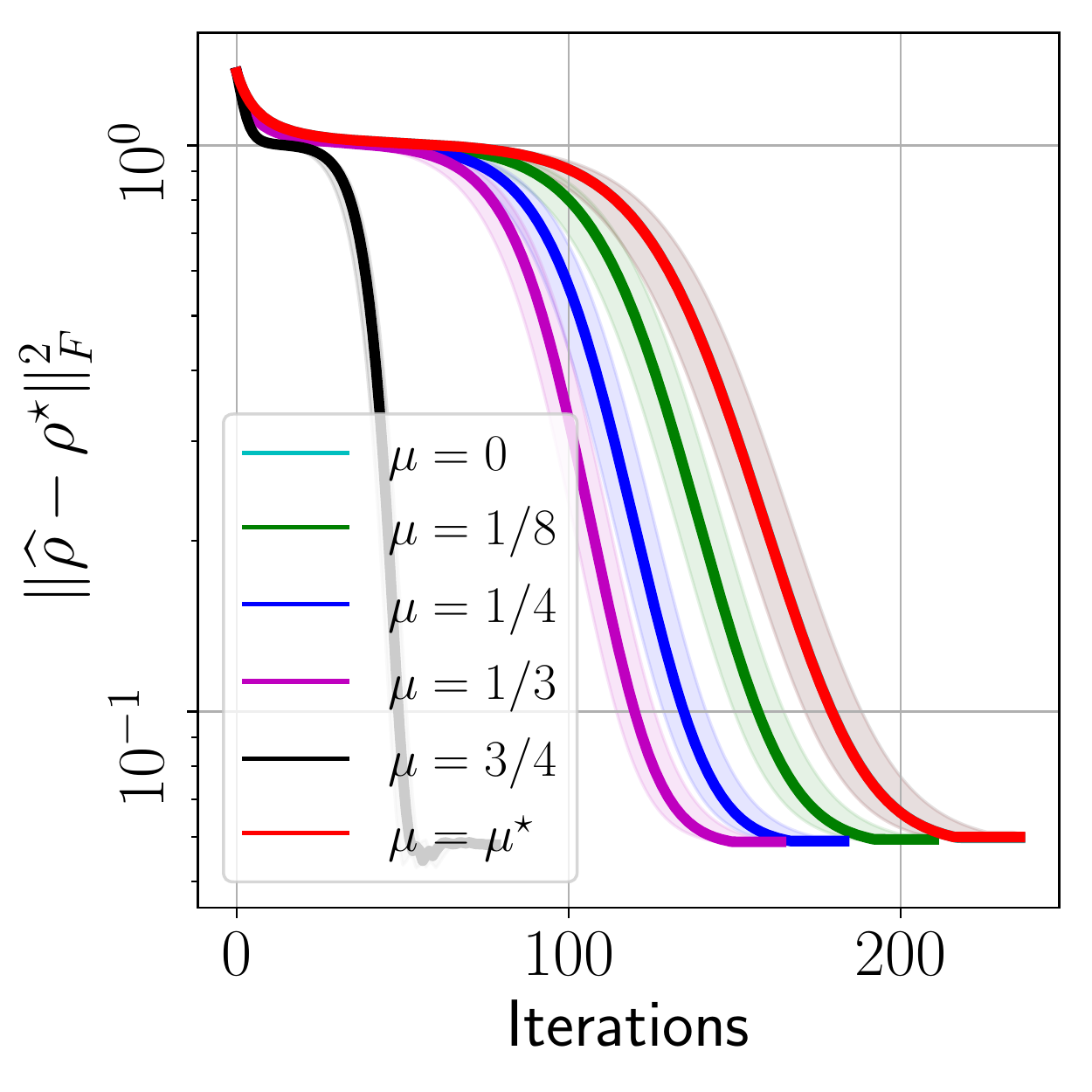}
    \caption{
      Target error list plots $\| \widehat{\rho} - \rho^\star \|_F^2$ versus method iteration using synthetic IBM's quantum simulator data. \textbf{Top-left}: $\texttt{GHZ}_{-}(6)$ with 2048 shots; \textbf{Top-middle}: $\texttt{GHZ}_{-}(6)$ with 8192 shots;  
      \textbf{Top-right}: $\texttt{GHZ}_{-}(8)$ with 2048 shots; \textbf{Bottom-left}: $\texttt{GHZ}_{-}(8)$ with 4096 shots; 
      \textbf{Bottom-middle}: $\texttt{Hadamard}(6)$ with 8192 shots;  
      \textbf{Bottom-right}: $\texttt{Hadamard}(8)$ with 4096 shots.
      All cases have $\texttt{measpc} = 20\%$.
      Shaded area denotes standard deviation around the mean over repeated runs in all cases.
      %\amir{I would suggest changing this figure too with accordance to my comments on Fig. 2.}
      }
    \label{fig:target_error_list_ibmq-simulator_measpc_20}
  \end{center}
\end{figure*}

Finally, in Figure \ref{fig:fidelity_list_measpc_20_all}, we depict the fidelity of $\widehat{\rho}$ achieved using \texttt{MiFGD}, defined as $\text{Tr}\big(\rho^\star \widehat{\rho}\big)$, versus various $\mu$ values and for different circuits $(\rho^\star)$. 
Shaded area denotes standard deviation around the mean over repeated runs in all cases.
The plots show the significant gap in performance when using real quantum data versus using synthetic simulated data within a controlled environment.

\begin{figure*}[!h]
  \begin{center}
    \includegraphics[width=1\textwidth]{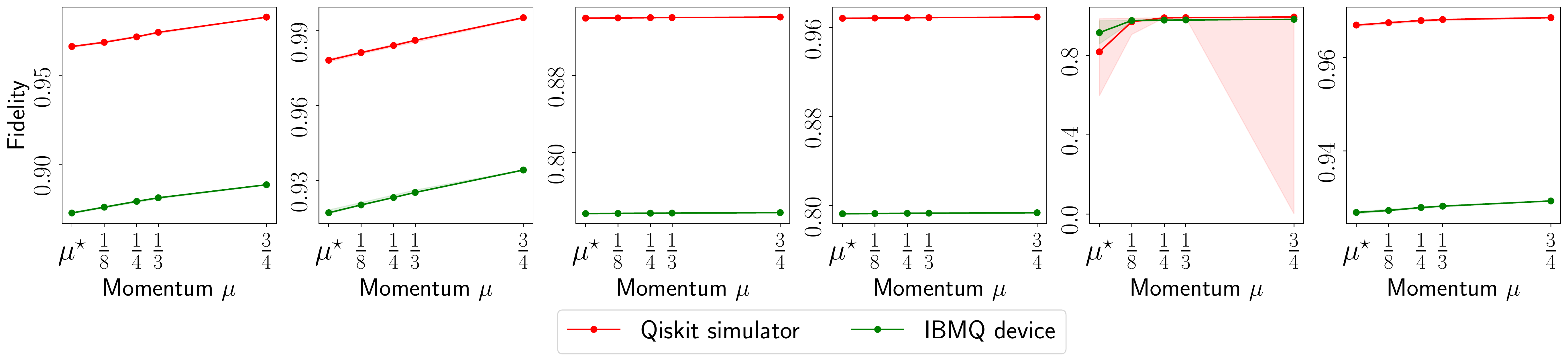}  
    \caption{
      Fidelity list plots where we depict the fidelity of $\widehat{\rho}$ to $\rho^\star$. From left to right: $i)$ $\texttt{GHZ}_{-}(6)$ with 2048 shots; $ii)$ $\texttt{GHZ}_{-}(6)$ with 8192 shots;  
      $iii)$ $\texttt{GHZ}_{-}(8)$ with 2048 shots; $iv)$ $\texttt{GHZ}_{-}(8)$ with 4096 shots; 
      $v)$ $\texttt{Hadamard}(6)$ with 8192 shots;  
      $vi)$ $\texttt{Hadamard}(8)$ with 4096 shots.
      All cases have $\texttt{measpc} = 20\%$.
      Shaded area denotes standard deviation around the mean over repeated runs in all cases. %\amir{This figure is way too small. It is hard to read anything. Since we only have few points $\mu=1/8,1/4,1/3$ and $3/4$, I suggest presenting the data in a table. In this way the results are much more clearly presented. If anyone insist on presenting a plot, then I would suggest to split it into two rows (similar to Fig 2, for example) and also make changes according to my comments on Fig. 2}
      }
    \label{fig:fidelity_list_measpc_20_all}
  \end{center}
\end{figure*}

\subsection{Performance comparison with full tomography methods in Qiskit} 
%\amir{I think it makes more sense to move this subsection before the NN subsection, as it is more "apples to apples"}
We compare \texttt{MiFGD} with publicly available implementations for QST reconstruction. 
Two common techniques for QST, included in the \texttt{qiskit-ignis} distribution \cite{qiskit}, are:
$i)$ the \texttt{CVXPY} fitter method, that uses the \texttt{CVXPY} convex optimization package \cite{diamond2016cvxpy, agrawal2018rewriting}; and $ii)$ the \texttt{lstsq} method, that uses least-squares fitting \cite{smolin2012efficient}.
Both methods solve \emph{the full tomography problem}\footnote{In Ref.~\cite{kalev2015quantum} it was sown that the minimization program~\eqref{eq:obj1} yields a robust estimation of low-rank states in the compressed sensing. Thus, one can use \texttt{CVXPY} fitter method to solve \eqref{eq:obj1} with  $m\ll d^2$ Pauli expectation value to obtain a robust reconstruction of $\rho^\star$.} %\amir{While this is certainly true for the \texttt{lstsq} method, I am not sure this is true for the CVX method. In my paper I show that any CVX can be used with CS} 
according to the following expression:
\begin{equation}\label{eq:obj1}
\begin{aligned}
& \min_{\rho \in \mathbb{C}^{d \times d}}
& & f(\rho) := \tfrac{1}{2} \|\mathcal{A}(\rho) - y\|_2^2 \\
& \text{subject to}
& & \rho \succeq 0, \texttt{Tr}(\rho) = 1.
\end{aligned}
\end{equation}
We note that \texttt{MiFGD} is not restricted to ``tall'' $U$ scenarios to encode PSD and rank constraints: even without rank constraints, one could still exploit the matrix decomposition $\rho = UU^\dagger$ to avoid the PSD projection, $\rho \succeq 0$, where $U \in \mathbb{C}^{d \times d}$.
For the \texttt{lstsq} fitter method, the putative estimate $\widehat{\rho}$ is rescaled using the method proposed in \cite{smolin2012efficient}. 
For \texttt{CVXPY}, the convex constraint makes the optimization problem a semidefinite programming (SDP) instance. 
By default, \texttt{CVXPY} calls the \texttt{SCS} solver that can handle all problems (including SDPs) \cite{ocpb:16, scs}. 
Further comparison results with matrix factorization techniques from the machine learning community is provided in the Appendix for $n = 12$. %\amir{since CVX doesn't need full tomography measurements, I think it would be more natural to compare CVX with CS measurements to \texttt{MiFGD}. We should also compare \texttt{MiFGD} to \texttt{FGD} ($\mu=0$). I don't think it makes sense to compare to \texttt{lstsq} which we know won't be used for large systems.}

The settings we consider for full tomography are the following: 
%\amir{I would remove all this part and re-run with CS measurements. It is clear that CVX would not work on systems of more than about 7 qubits, and no experiment (real or simulation) can actually perform full tomography  on more than about 7 qubits. What we need to show in this section, to my opinion, is that we are winning over state-of-the-art CS reconstruction methods (!) such as CVX on CS or e.g. the method used by Riofrio.} 
$\texttt{GHZ}(n)$, $\texttt{Hadamard}(n)$ and $\texttt{Random}(n)$ quantum states (for $n=3, \dots, 8$).
%$\texttt{GHZ}(3$-$8)$, $\texttt{Hadamard}(3$-$8)$, and $\texttt{Random}(3$-$8)$.
We focus on fidelity of reconstruction and computation timings performance between $\texttt{CVXPY}$, $\texttt{lstsq}$ and $\texttt{MiFGD}$.
We use 100\% of the measurements. 
We experimented with states simulated in QASM and measured taking $2,048$ shots. 
For \texttt{MiFGD}, we set $\eta = 0.001$, $\mu = \tfrac{3}{4}$, and stopping criterion/tolerance $\texttt{reltol} = 10^{-5}$.
All experiments are run on a Macbook Pro with 2.3 GHz Quad-Core Intel Core i7CPU and 32GB RAM.

The results are shown in Figure \ref{fig:fidelity_time_qiskit}; higher-dimensional cases are provided in Table \ref{table:qiskit_fidelity_time_2048}. 
Some notable remarks: $i)$ For small-scale scenarios ($n=3, 4$), \texttt{CVXPY} and \texttt{lstsq} attain almost perfect fidelity, while being comparable or faster than \texttt{MiFGD}. $ii)$ The difference in performance becomes apparent from $n = 6$ and on: while \texttt{MiFGD} attains 98\% fidelity in $<5 $ seconds, \texttt{CVXPY} and \texttt{lstsq} require up to hundreds of seconds to find a good solution. 
%$iii)$ Finally, while \texttt{MiFGD} gets to high-fidelity solutions in seconds for $n = 7, 8$, \texttt{CVXPY} and \texttt{lstsq} did not converge after 48 hours of execution.
%\textcolor{red}{Lyle:
$iii)$ Finally, while \texttt{MiFGD} gets to high-fidelity solutions in seconds for $n = 7, 8$, \texttt{CVXPY} and \texttt{lstsq} methods could not finish tomography as their memory usage exceeded the system's available memory. 
% \footnote{
% \textcolor{blue}{
% \textbf{
% % The reason why \texttt{lstsq} (and \texttt{CVXPY}) runs out of memory for $n=7$ is not entirely due to the computational inefficiency of the methods, but more due to the implementation choice of \texttt{StateTomographyFitter} class provided by \texttt{Qiskit}. 
% The memory shortage of \texttt{lstsq} (and \texttt{CVXPY}) for $n \geq 7$ is not entirely due to the computational inefficiency of the methods, but due to the implementation choice of \texttt{StateTomographyFitter} class provided by \texttt{Qiskit}. This class relies on the function \texttt{\_fitter\_data()}\href{https://github.com/Qiskit/qiskit-ignis/blob/4c45654256ce8fecb60cb1d9d5ff481d6efd3428/qiskit/ignis/verification/tomography/fitters/base_fitter.py#L289}{\texttt{[link]}}, which loops over $3^n$ items, where $n$ is the number of qubits. Within the \texttt{for} loop, a complex matrix of size $2^n \times 4^n$ is created and appended to a list. Ignoring all other objects, at the end of the \texttt{for} loop, this list object has $(3 \cdot 2 \cdot 4)^n$ complex entries; for $n=7$, this results in more than 68 GB, assuming each complex entry uses 128 bits.}}
% }

It is noteworthy that the reported fidelities for \texttt{MiFGD} are the fidelities at the last iteration, before the stopping criterion is activated, or the maximum number of iterations is exceeded. 
%\textcolor{red}{Lyle: Qiskit methods also do not directly maximize the fidelity criterion}
%Since \texttt{MiFGD} does not directly maximize the fidelity criterion (i.e., it minimizes the Euclidean $\ell_2$-norm objective), t
However, the reported fidelity is not necessarily the best one during the whole execution: for all cases, we observe that \texttt{MiFGD} finds intermediate solutions with fidelity $>99\%$. 
Though, it is not realistic to assume that the iteration with the best fidelity is known a priori, and this is the reason we report only the final iteration fidelity.

\begin{table}[!htp]
\centering
\begin{tabular}{llrr}
\hline
   Circuit   & Method   &   Fidelity &   Time (secs) \\
\hline
          $\texttt{GHZ}(7)$ & \texttt{MiFGD}               &   0.969397 &                10.6709   \\
          $\texttt{Hadamard}(7)$ & \texttt{MiFGD}               &   0.969397 &                10.4926   \\
          $\texttt{Random}(7)$ & \texttt{MiFGD}               &   0.968553 &                 9.59607  \\
          All above   & \texttt{lstsq}, \texttt{CVXPY}             &   \multicolumn{2}{c}{Memory limit exceeded} \\    
          \hline
          $\texttt{GHZ}(8)$       & \texttt{MiFGD}               &   0.940389 &                35.0666   \\
          $\texttt{Hadamard}(8)$  & \texttt{MiFGD}               &   0.940390  &                37.5331   \\
          $\texttt{Random}(8)$    & \texttt{MiFGD}               &   0.942815 &                36.3251   \\
          All above   & \texttt{lstsq}, \texttt{CVXPY}             &   \multicolumn{2}{c}{Memory limit exceeded} \\    
\hline
\end{tabular}
\caption{
Fidelity of reconstruction and computation timings using 100\% of the complete measurements. Rows correspond to combinations of number of qubits (7 $\sim$ 8), synthetic circuit, and tomographic method (\texttt{MiFGD}, Qiskit's \texttt{lstsq} and \texttt{CVXPY} fitters. 2048 shots per measurement circuit. For \texttt{MiFGD}, $\eta=0.001$, $\mu=\tfrac{3}{4}$, $\texttt{reltol}=10^{-5}$. 
% ``-'' indicates the procedure did not finish running in 48 hours. 
All experiments are run on a 13'' Macbook Pro with 2.3 GHz Quad-Core Intel Core i7 CPU and $32$ GB RAM.}
\label{table:qiskit_fidelity_time_2048}
\end{table}

%\subsection{Synthetic experiments: 10 qubits}
% \added[id=GK]{[Tables in this section have been automatically generated by parsing earlier (April 2018), unattended runs captured in notebooks also under:
%   \protect\url{https://github.com/gidiko/quantum-tomography/tree/real/runs/fgd_10/notebooks}.
%   Code for both parsing the notebooks and generating the tables is available in
%   \protect\url{https://github.com/gidiko/quantum-tomography/blob/real/runs/fgd_10/parse_notebooks.ipynb}.
%   The json file with the data extracted from the notebooks is available in
%   \protect\url{https://github.com/gidiko/quantum-tomography/blob/real/runs/fgd_10/nb-parsed_10_qiskit-simulator_8192.json}]}

% \added[id=GK]{[Describe the semantics of column labels based on the content from the notebooks and comment on load times. Also confirm whether the exceution server was the same as for the newer experiments in the next section]}

\begin{figure*}[!h]
  \begin{center}
    \includegraphics[width=0.32\textwidth]{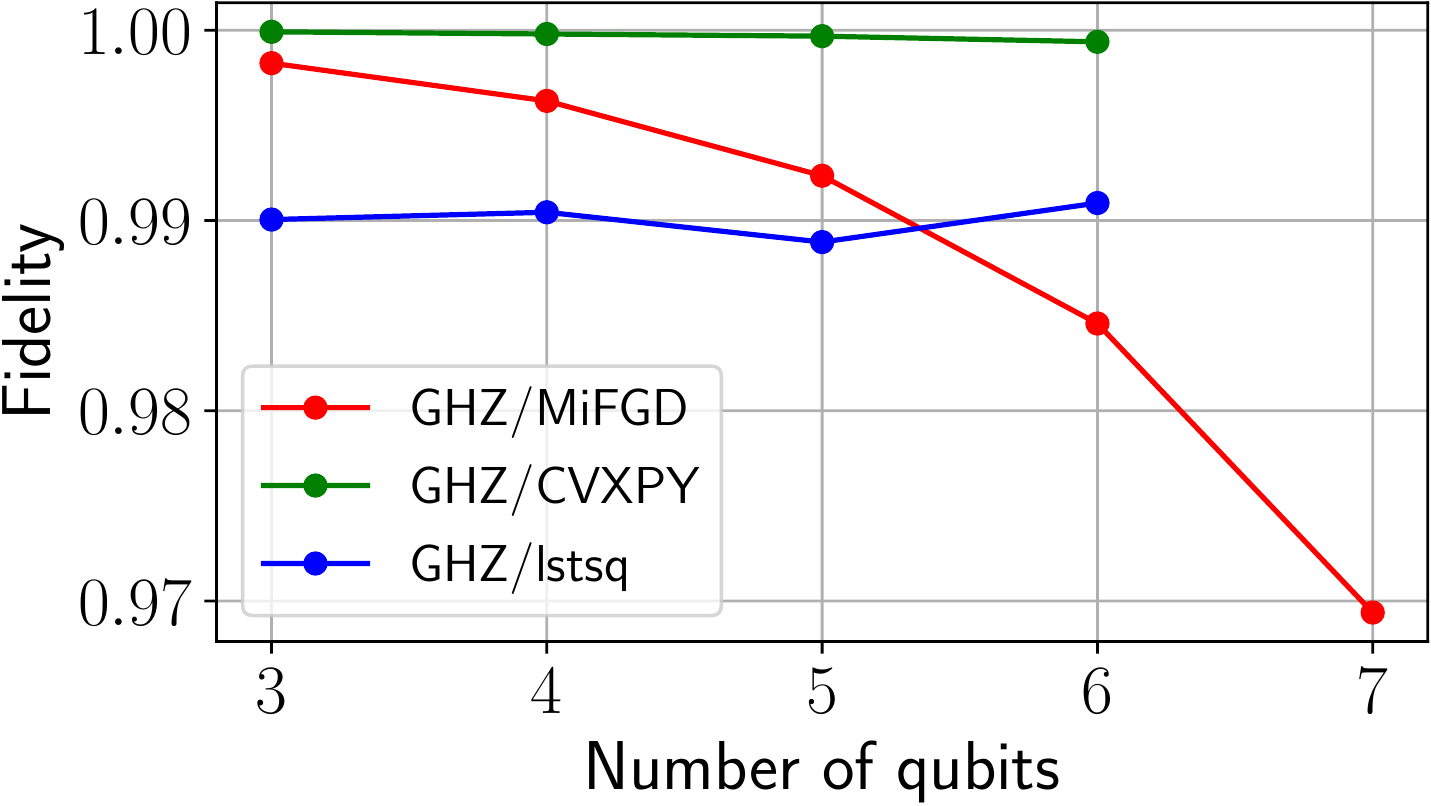}
    \includegraphics[width=0.32\textwidth]{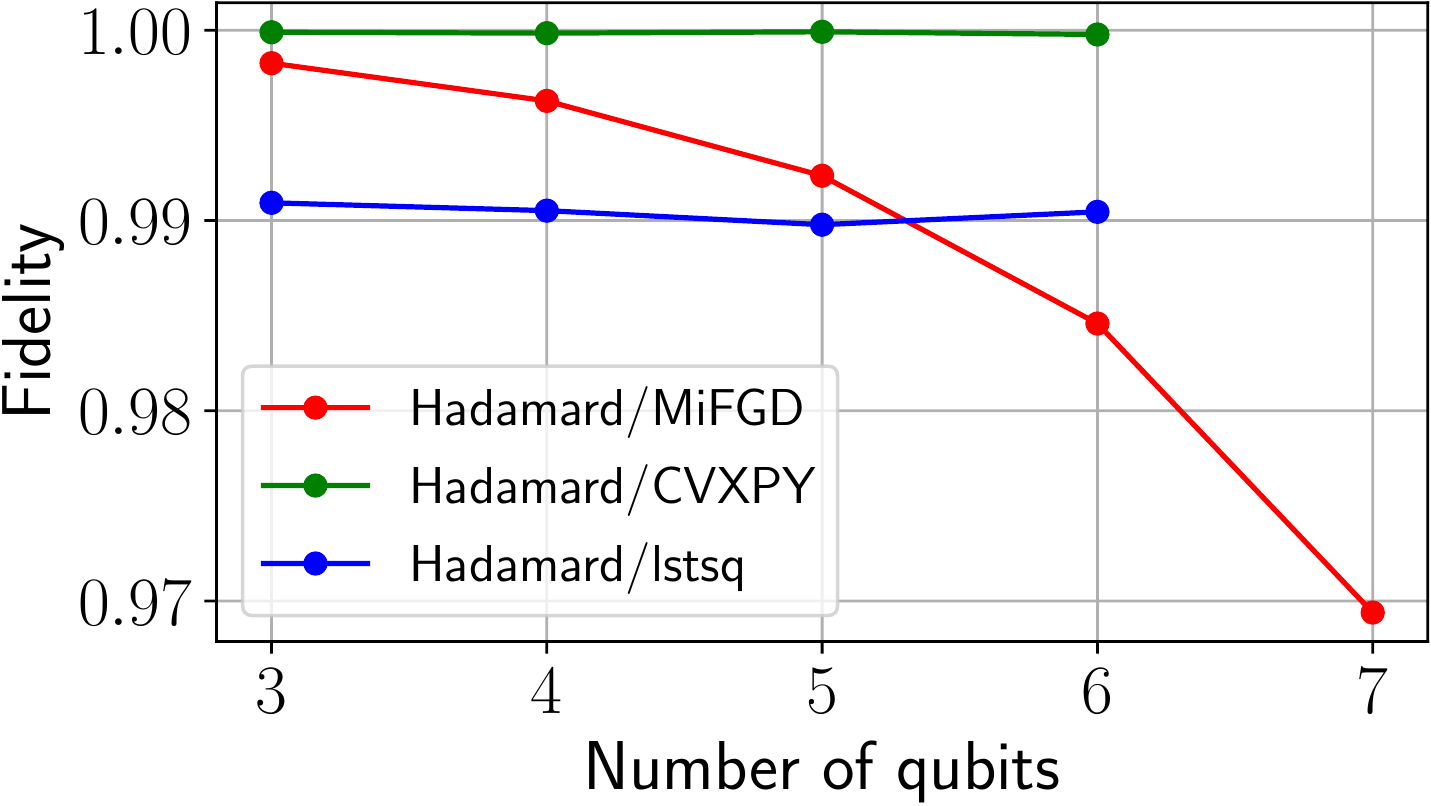}
    \includegraphics[width=0.32\textwidth]{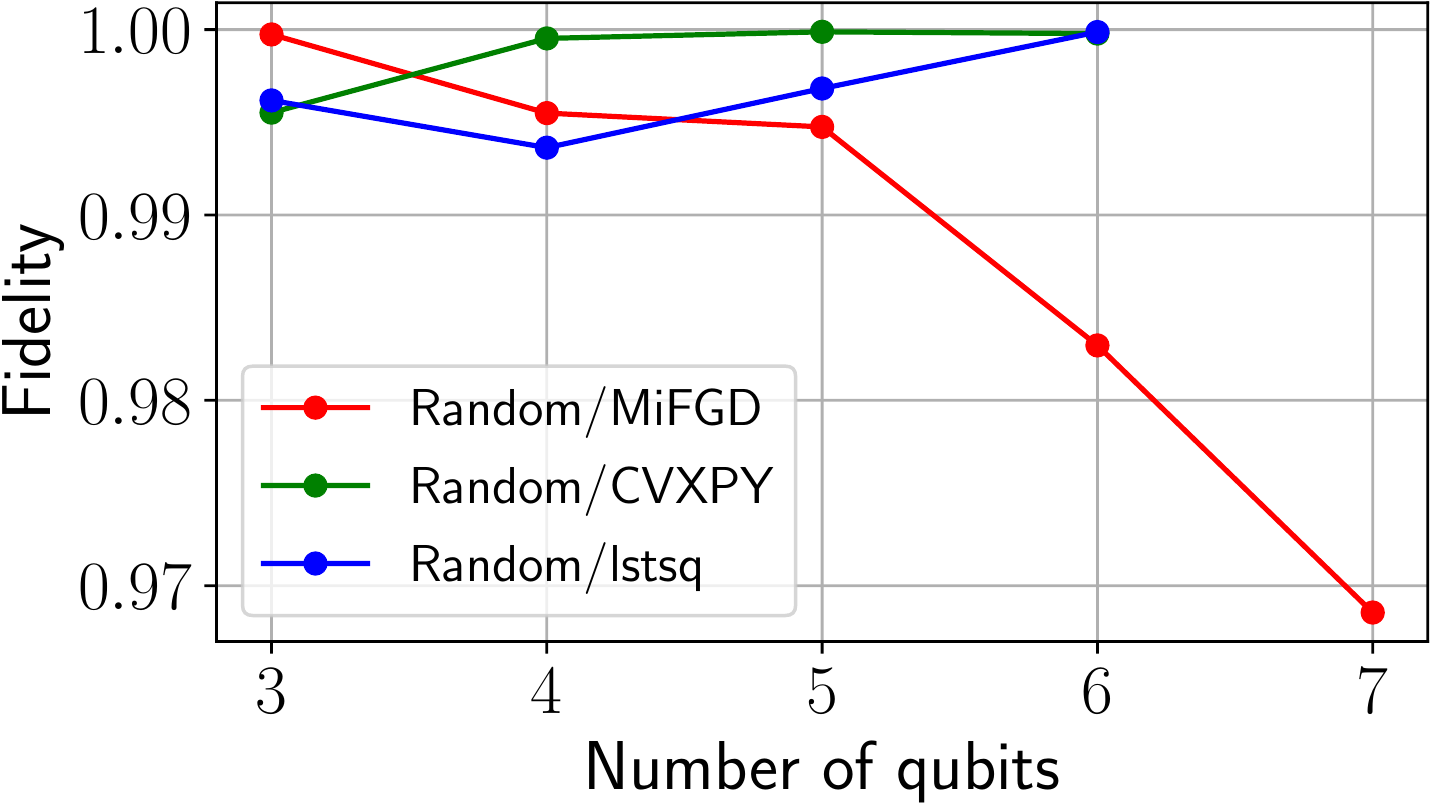}
    \includegraphics[width=0.32\textwidth]{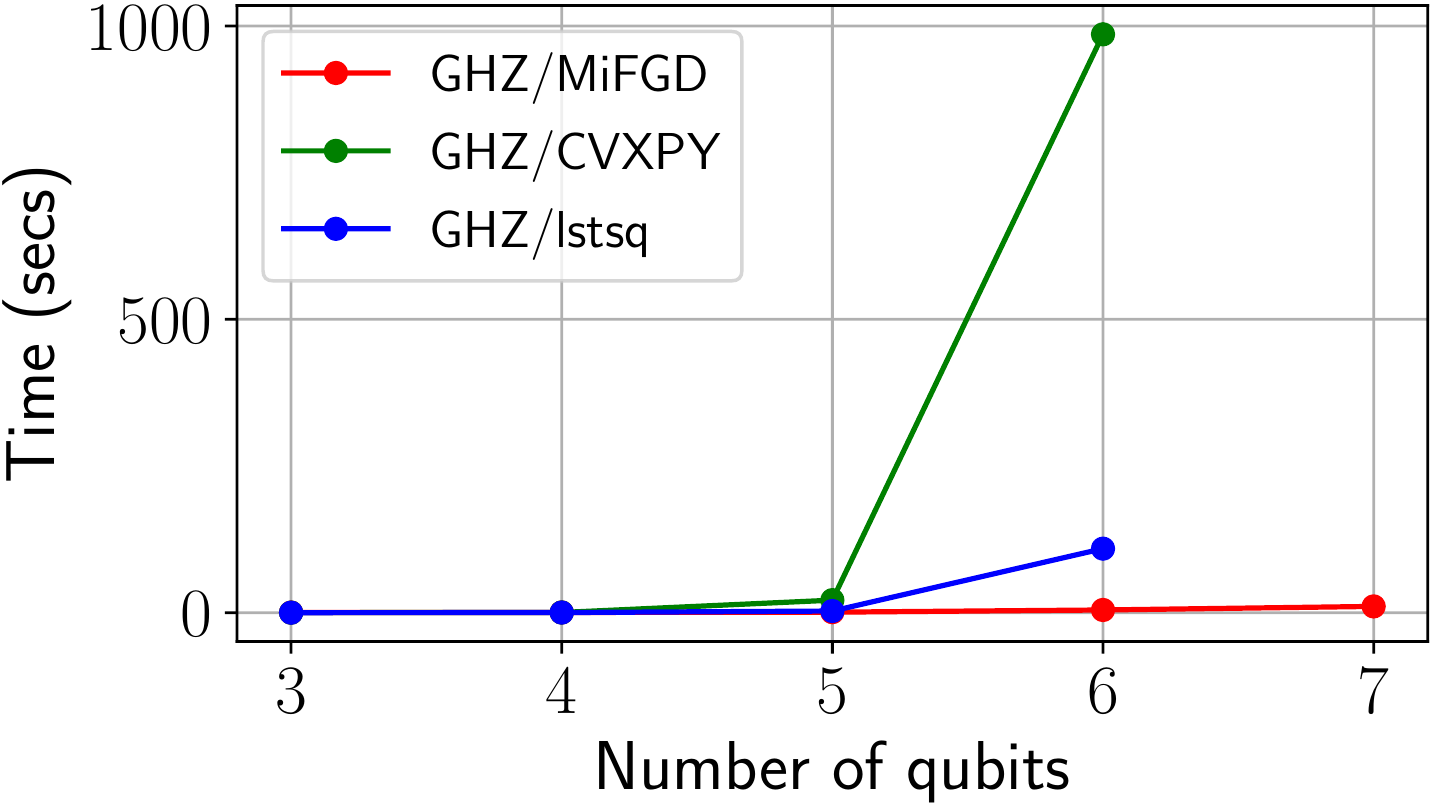}
    \includegraphics[width=0.32\textwidth]{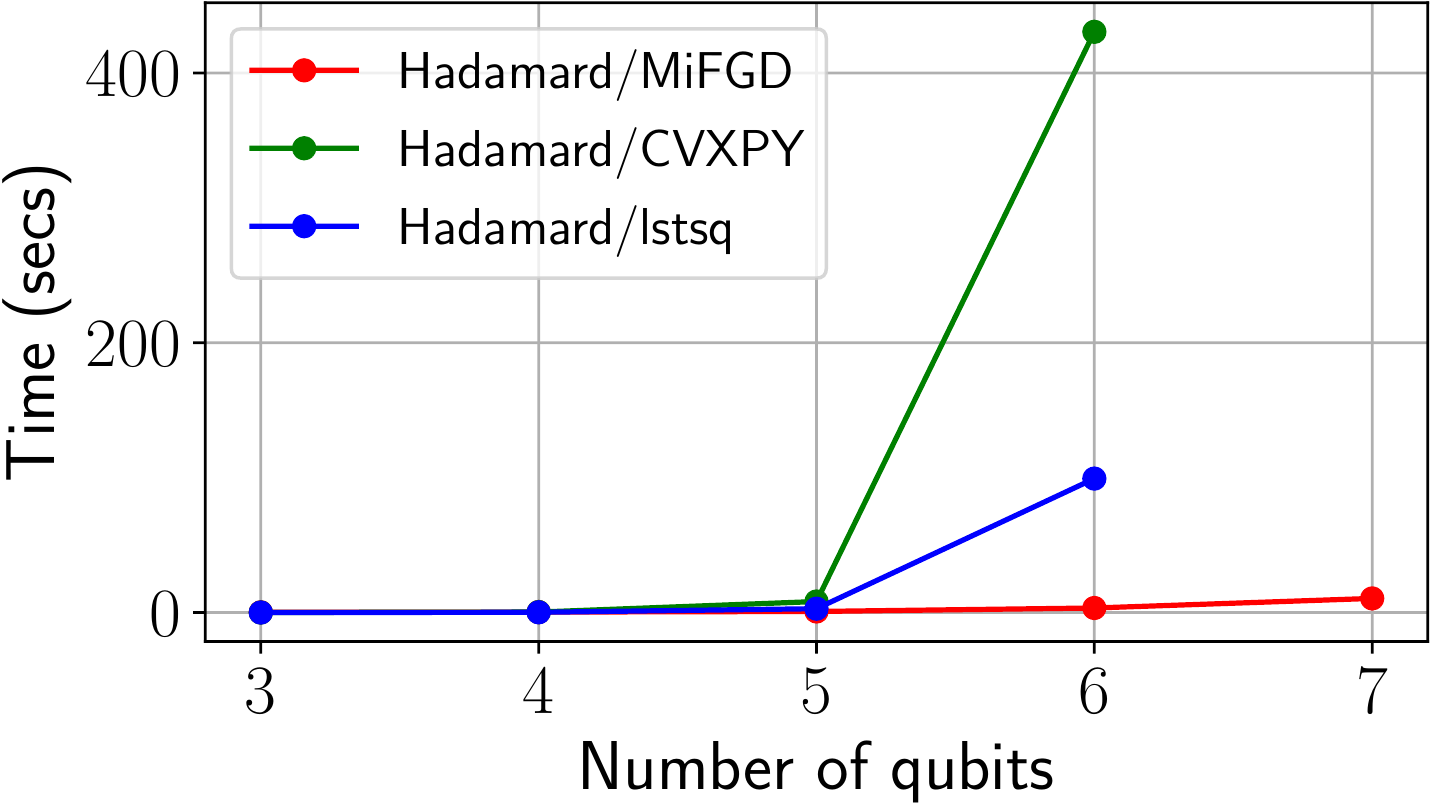}
    \includegraphics[width=0.32\textwidth]{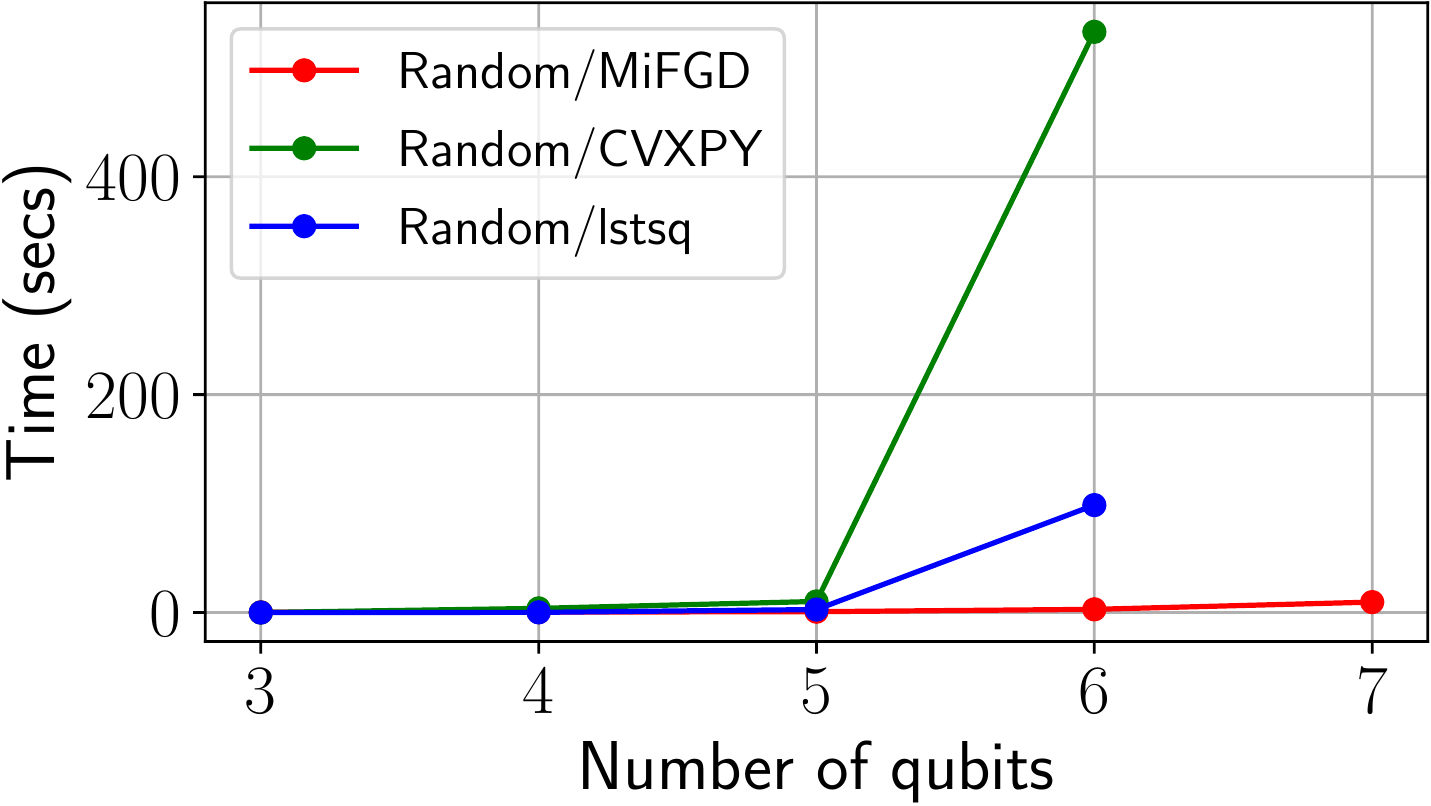}
    \caption{
      Fidelity versus time plots using synthetic IBM's quantum simulator data. \textbf{Left panel}: $\texttt{GHZ}_{-}(n)$ for $n = 3, 4$; \textbf{Middle panel}: $\texttt{Hadamard}_{-}(n)$ for $n = 3, 4$; 
      \textbf{Right panel}: $\texttt{Random}_{-}(n)$ for $n = 3, 4$.}
    \label{fig:fidelity_time_qiskit}
  \end{center}
\end{figure*}

\subsection{Performance comparison of \texttt{MiFGD} with neural-network quantum state tomography}
% \added[id=GK]{[For an example of relevant codes for this part please refer under:
%   \protect\url{https://github.com/gidiko/quantum-tomography/tree/real/runs/qucumber_random_21},
%   particularly: \protect\url{https://github.com/gidiko/quantum-tomography/blob/real/runs/qucumber_random_2/run_qucumber.ipynb} and the code in \protect\url{https://github.com/gidiko/quantum-tomography/blob/real/runs/qucumber_random_2/run_NN.py}]}

% Recently, \cite{Torlai2018Neural, beach2019qucumber, torlai2019machine, gao2017efficient} show how machine learning and neural networks can be used to perform QST, driven by experimental data. 
% The neural network architecture used is based on restricted Boltzmann machines (RBMs) \cite{sutskever2009recurrent}, which feature a visible and a hidden layer of stochastic binary neurons, fully connected with weighted edges. 
% Test cases considered include the reconstruction of the W state, magnetic observables of local Hamiltonians, and the unitary dynamics induced by Hamiltonian evolution. 

We compare the performance of \texttt{MiFGD} with neural network approaches.
Per \cite{Torlai2018Neural, beach2019qucumber, torlai2019machine, gao2017efficient}, we model a quantum state with a two-layer Restricted Boltzmann Machine (RBM). 
RBMs are stochastic neural networks, where each layer contains a number of binary stochastic variables: the size of the visible layer corresponds to the number of input qubits, while the size of the hidden layer is a hyperparameter controlling the representation error. 
We experiment with three types of RBMs for reconstructing either the positive-real wave function, the complex wave function, or the density matrix of the quantum state.
In the first two cases the state is assumed pure while in the last, general mixed quantum states can be represented. 
We leverage the implementation in QuCumber \cite{beach2019qucumber},
\texttt{PositiveRealWaveFunction} (\texttt{PRWF}), \texttt{ComplexWaveFunction} (\texttt{CWF}), and \texttt{DensityMatrix} (\texttt{DM}), respectively.

We reconstruct $\texttt{GHZ}(n)$, $\texttt{Hadamard}(n)$ and $\texttt{Random}(n)$ quantum states (for $n=3, \dots, 8$), by training \texttt{PRWF}, \texttt{CWF}, and \texttt{DM} neural networks\footnote{We utilize GPU (NVidia GeForce GTX 1080 TI,11GB RAM) for faster training of the neural networks.} with measurements collected by the QASM Simulator. 

%\amir{We need to explain why we chose CWF for the ghz and for the hadamard though they are real, and also why we use DM where we are dealing with pure states. BTW, since in our algorithm we assume pure state, is it fair to compare it to DM NN to begin with? } 

% \textcolor{magenta}{@Lyle: Let's discuss possible testing scenarios: $i)$ We want to measure the time NN requires to get to some fidelity for $n = 2, 3, 4$, using only CPUs and only GPUs. $ii)$ We need to check what happens when $n$ is larger and we use GPUs. $iii)$ We need to report the final accuracy NN training gets compared to us, irrespective of time (probably NNs have better final fidelity and we have to report that). $iv)$ I'm worried about the settings we run against NN: in the text, we mention $\mu = 0$ (not sure why), we need to do a bit of hyper parameter tuning; also I'm not sure why we use only 128 \texttt{shots}. The final result can be reported as: $i)$ a set of figures like figs 6-8 where we show how fidelity evolves over time; $ii)$ a table that shows the computational time needed to get some milestone fidelities (e.g., fidelities 0.6, 0.8, 0.9, 0.925, 0.95, 0.99).}

For our setting, we consider  \texttt{measpc} = 50\% and \texttt{shots} = 2048. 
The set of measurements is presented to the RBM implementation, along with the target positive-real wave function (for \texttt{PRWF}), complex wavefunction (for \texttt{CWF}) or the target density matrix (for \texttt{DM}) in a suitable format for training. 
We train \texttt{Hadamard} and \texttt{Random} states with 20 epochs, and \texttt{GHZ} state with 100 epochs.\footnote{We experimented higher number of epochs (up to 500) for all cases, but after the reported number of epochs, Qucumber methods did not improve, if not worsened.}
%\amir{I even more confused as why using CWF and DM if the NN is trained on the target. Shouldn't we just use RWF  for the ghz and the hadamard and CWF for the random (if it is indeed complex). Is it a fair comparison?} 
We set the number of hidden variables (and also of additional auxilliary variables for \texttt{DM}) to be equal to the number of input variables $n$ and we use $100$ data points for both the positive and the negative phase of the gradient (as per the recommendation for the defaults). 
% %%%%%%%%%%%%%%%%%%
% \textcolor{red}{Lyle: $k=10$, learning rate = 10. ``using the totality of available measurements''--> \texttt{MiFGD} also uses 50\%. }
% %%%%%%%%%%%%%%%%%%
We choose $k=10$ contrastive divergence steps and fixed the learning rate to $10$ (per hyperparameter tuning). 
Lastly, we limit the fitting time of Qucumber methods (excluding data preparation time) to be three hours.
%\amir{we use here lots of jargon, and assume prior knowledge about the subject. We didn't make that assumption when we discussed quantum tomography above. I suggest to elaborate a bit more about the RBM model a couple of paragraphs, so the jargon we use here is a bit more "defined".} 
To compare to the RBM results, we run \texttt{MiFGD} with $\eta= 0.001$, $\mu=\tfrac{3}{4}$, $\texttt{reltol}=10^{-5}$ and using \texttt{measpc} = 50\%, keeping previously chosen values for all other hyperparameters.

We report the fidelity of the reconstruction as a function of elapsed training time for $n = 3, 4$ in Figure \ref{fig:fidelity_time_neural} for \texttt{PRWF}, \texttt{CWF}, and \texttt{DM}. 
%\textcolor{red}{Note that Time (secs) on $x$-axis is plotted with log-scale.}
% We observe that for the $\texttt{GHZ}(n)$ and $\texttt{Random}(n)$, reaching reasonable fidelities is slower for both \texttt{CWF} and \texttt{DM}. 
We observe that for all cases, Qucumber methods are orders of magnitude slower than \texttt{MiFGD}.
% Moreover, in many cases, Qucumber methods do not reach reasonable fidelity, most notably on  $\texttt{GHZ}(n)$ state. 
E.g., for $n=8$, for all three states, \texttt{CWF} and \texttt{DM} did not finish a single epoch in 3 hours, while \texttt{MiFG} achieves high fidelity in less than 30 seconds. 
For the $\texttt{Hadamard}(n)$ and $\texttt{Random}(n)$, reaching reasonable fidelities is significantly slower for both \texttt{CWF} and \texttt{DM}, while \texttt{PRWF} hardly improves its performance throughout the training.
For the \texttt{GHZ} case, \texttt{CWF} and \texttt{DM} also shows \emph{non-monotonic} behaviors: even after a few \emph{thousands of seconds}, fidelities have not ``stabilized'',
while \texttt{PRWF} stabilizes in very low fidelities.
% For the \texttt{GHZ} case, the curves include also \emph{non-monotonic} parts: even after a few \emph{thousands of seconds}, fidelities have not ``stabilized''.  
In comparison \texttt{MiFGD} is several orders of magnitude faster than both \texttt{CWF} and \texttt{DM} and fidelity smoothly increases to comparable or higher values. 
Further, in Table \ref{table:qucumber_fidelity_time_2048}, we report \emph{final} fidelities (within the 3 hour time window), and reported times.

\begin{figure*}[!h]
  \begin{center}
    \includegraphics[width=0.32\textwidth]{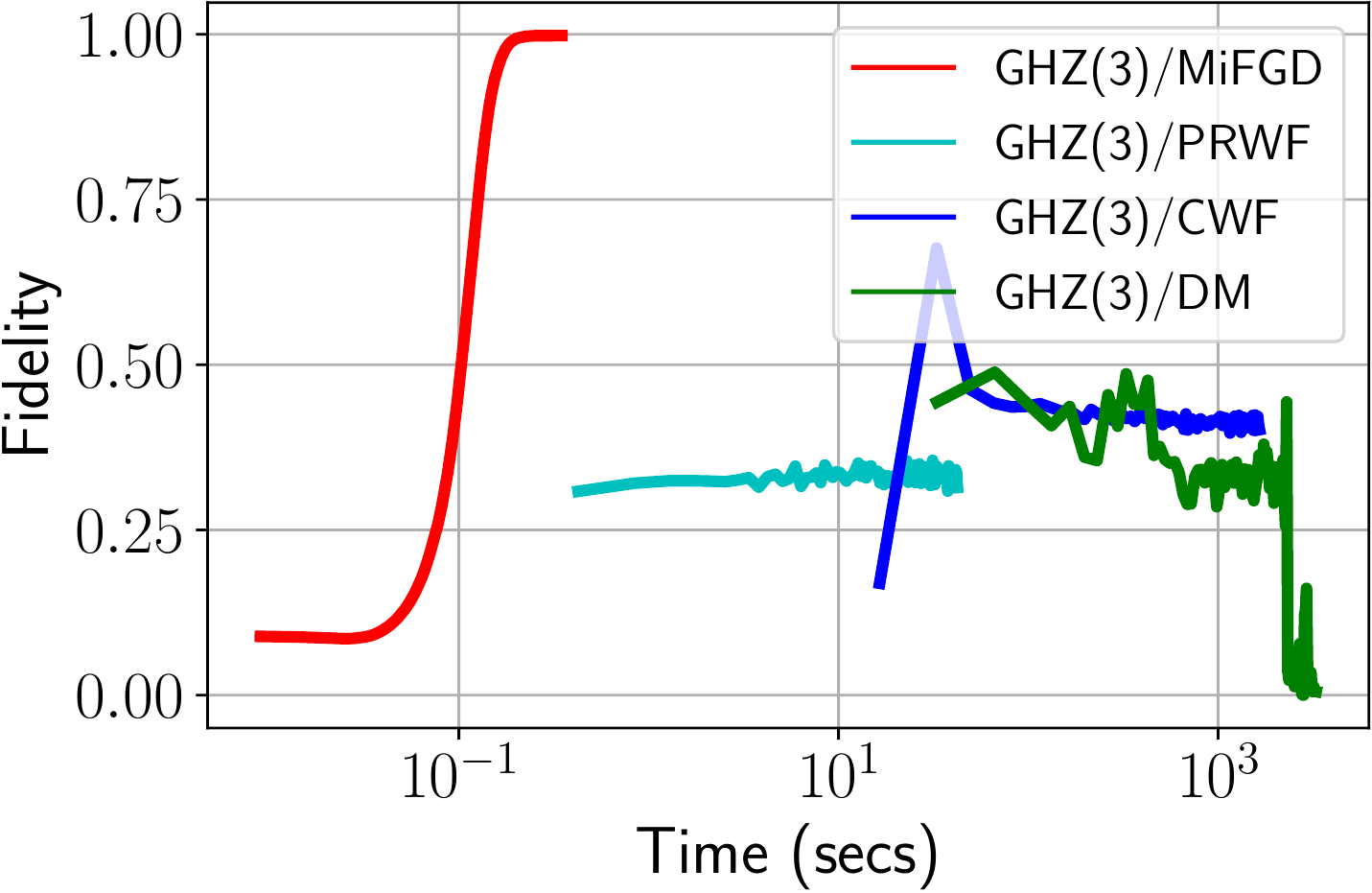}
    \includegraphics[width=0.32\textwidth]{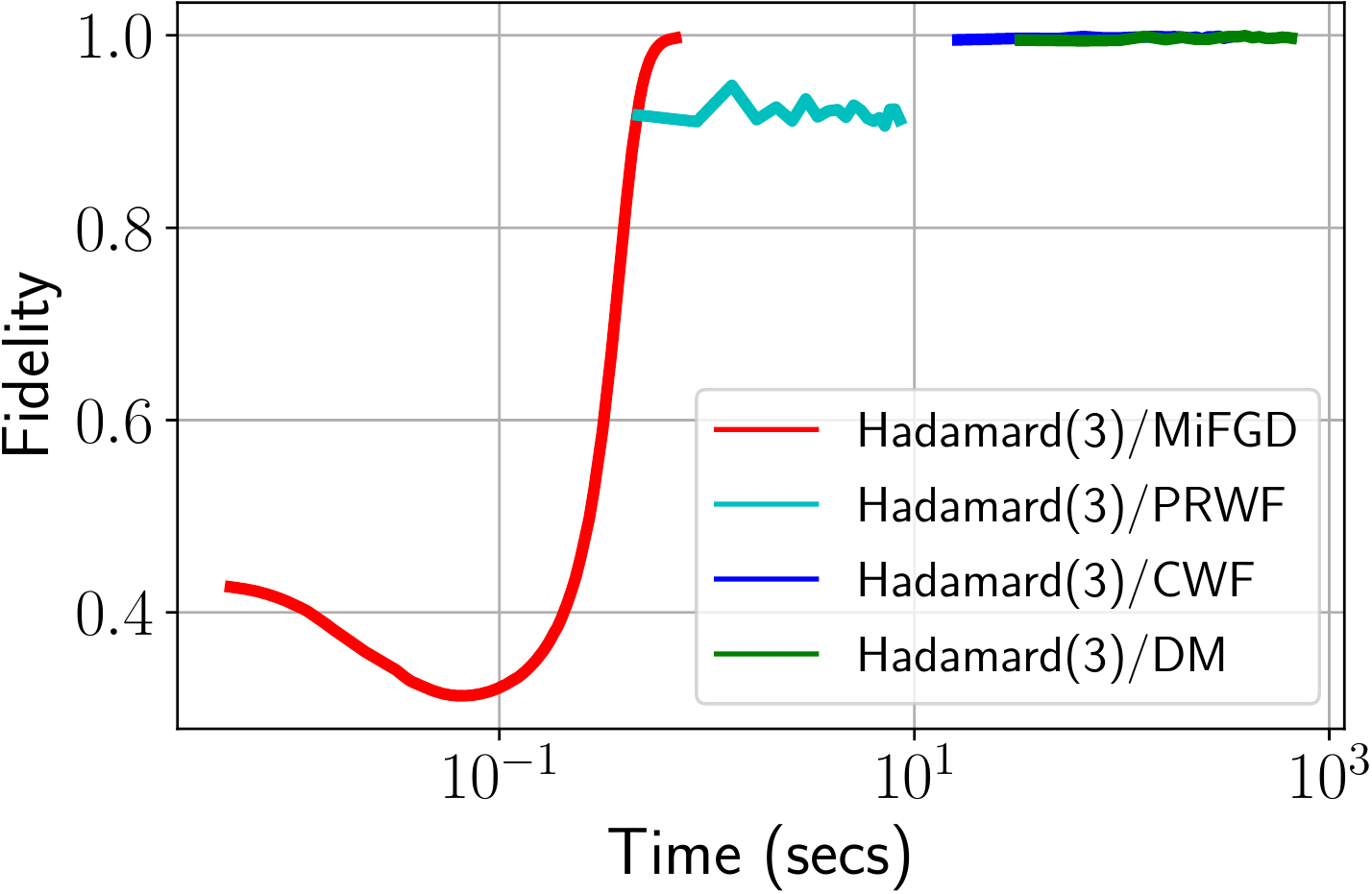}
    \includegraphics[width=0.32\textwidth]{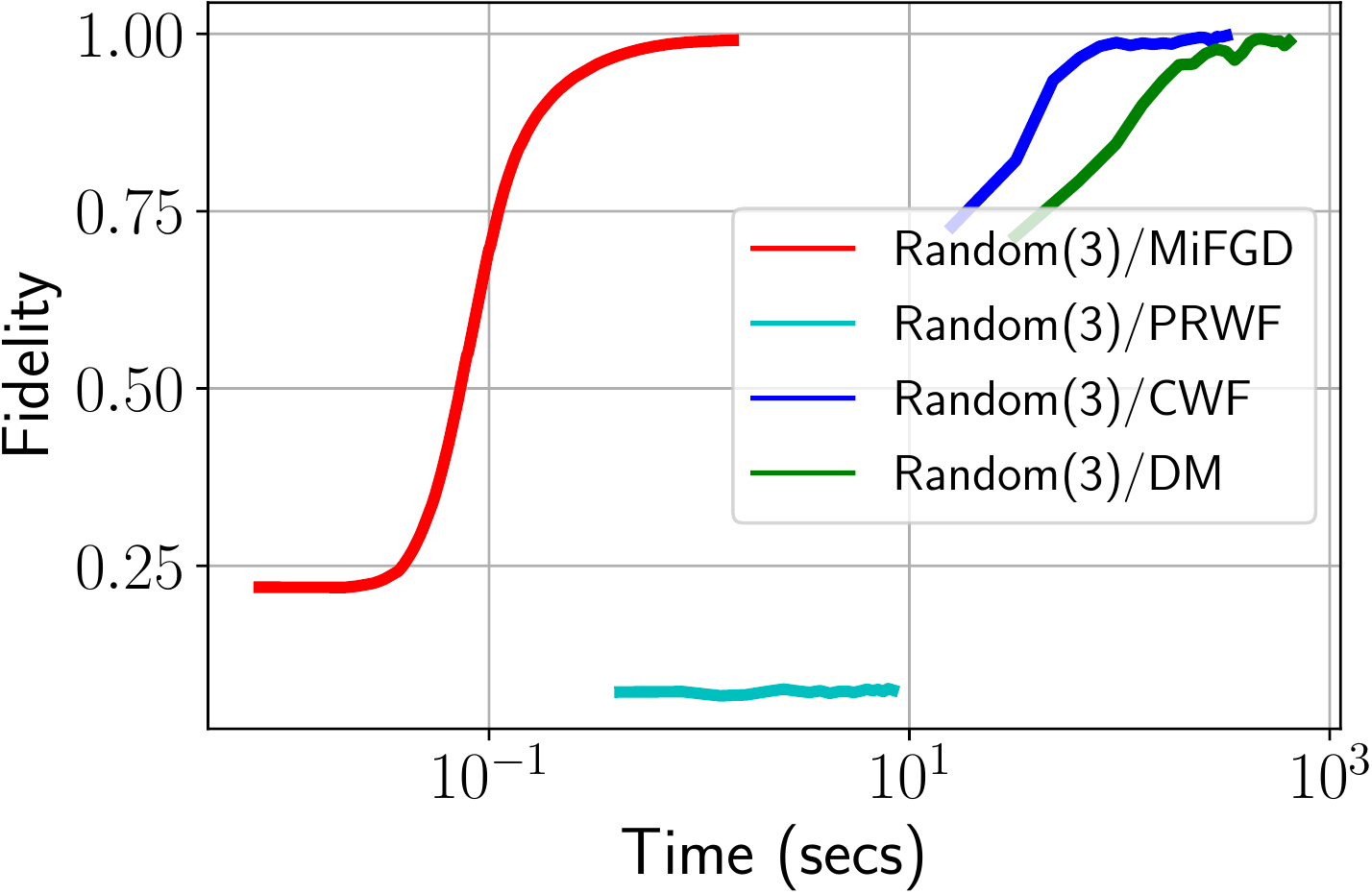}
    \includegraphics[width=0.32\textwidth]{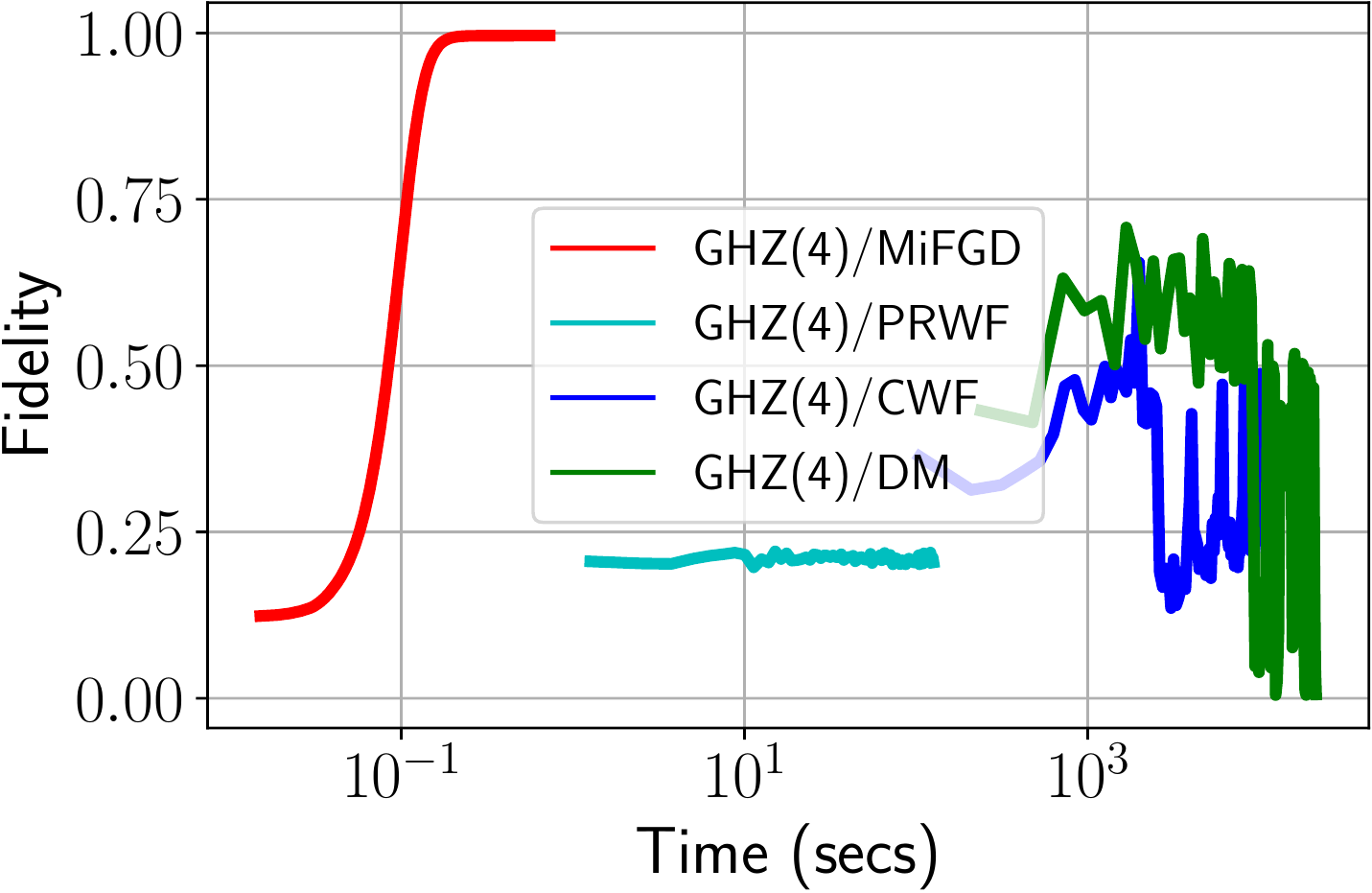}
    \includegraphics[width=0.32\textwidth]{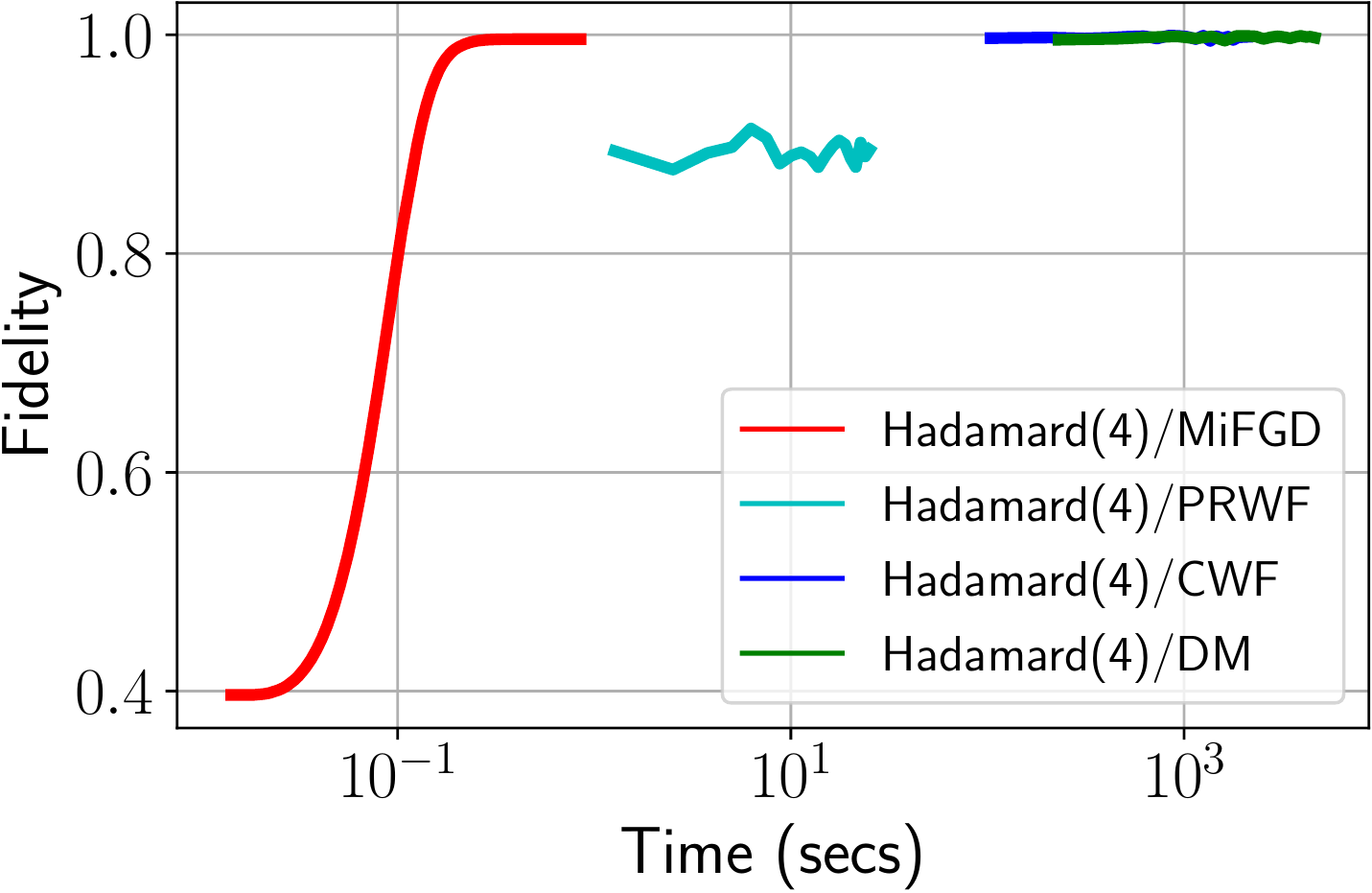}
    \includegraphics[width=0.32\textwidth]{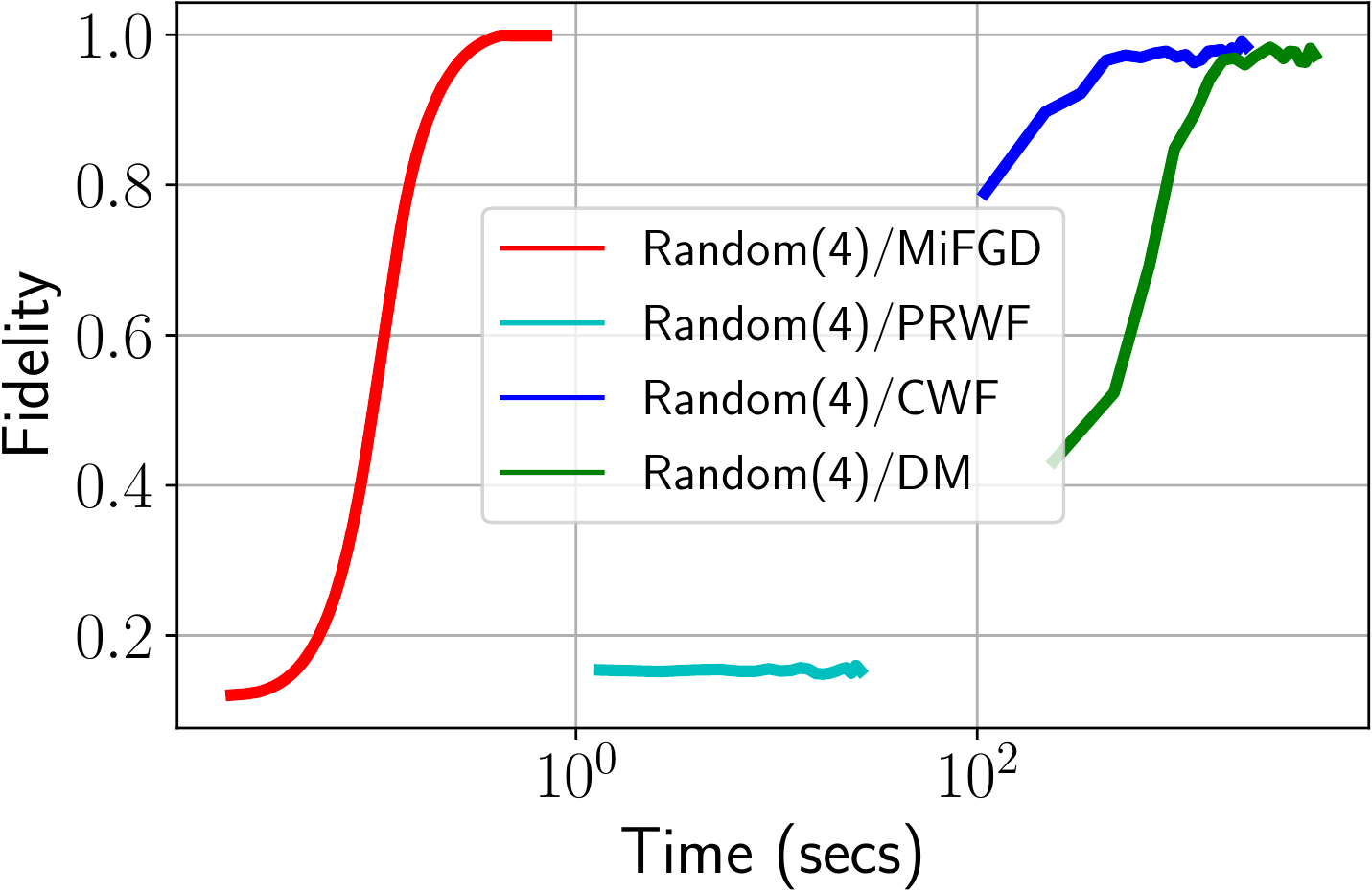}
    \caption{
      %\textcolor{red}{Lyle: NEWPLOT}
      Fidelity versus time plots on \texttt{MiFGD}, \texttt{PRWF}, \texttt{CWF}, and \texttt{DM}, using synthetic IBM's quantum simulator data. \textbf{Left panel}: $\texttt{GHZ}(n)$ for $n = 3, 4$; \textbf{Middle panel}: $\texttt{Hadamard}(n)$ for $n = 3, 4$; 
      \textbf{Right panel}: $\texttt{Random}(n)$ for $n = 3, 4$. %\amir{I suggest cutting the four figures on the right where the convergence is T=1000 and T=10000 for top two and bottom two, resp. Also the y axis should be from 0.75 to 1. in this way we zoom in and see how the 3 methods compare}
      }
    \label{fig:fidelity_time_neural}
  \end{center}
\end{figure*}

\begin{table}[!htp]
\centering
\begin{tabular}{ccccccc}
    \hline
    Circuit &  & \multicolumn{5}{c}{Method}  \\ 
    \hline
     &  & \texttt{MiFGD} & \texttt{FGD} & \texttt{PRWF} & \texttt{CWF}  & \texttt{DM} \\
    \hline
    \multirow{2}{*}{$\texttt{GHZ}(3)$} & Fidelity  & 0.997922 & 0.997857 & 0.314167 & 0.401737  & 0.005389 \\
                                       & Time (secs)  & 0.348652 & 1.061421 & 42.27607 & 1649.224  & 3279.118 \\ 
    \hline
    \multirow{2}{*}{$\texttt{Hadamard}(3)$} & Fidelity  & 0.997229 & 0.994191 & 0.912268 & 0.997914  & 0.997222 \\
                                       & Time (secs)  & 0.706872 & 2.399405 & 8.492405 & 325.7040  & 656.6696 \\ 
    \hline
    \multirow{2}{*}{$\texttt{Random}(3)$} & Fidelity  & 0.991063 & 0.988746 & 0.074774 & 0.997493  & 0.989754 \\
                                       & Time (secs)  & 1.447057 & 3.431218 & 8.345135 & 322.4730  & 640.8185 \\ 
    %%%%%%%%%%%%%%%%%%%%%%%%%%%%%%%%%%%%%%%%%%%%%
    \hline
    \multirow{2}{*}{$\texttt{GHZ}(4)$} & Fidelity  & 0.996029 & 0.996041 & 0.204313 & 0.276491 & 0.138459 \\
                                       & Time (secs) & 0.733128 & 2.081035 & 126.2749 & 10756.87  & $>$ 3h \\ 
    \hline
    \multirow{2}{*}{$\texttt{Hadamard}(4)$} & Fidelity  & 0.996078 & 0.996083 & 0.894883 & 0.998071  & 0.997389 \\
                                       & Time (secs)  & 0.852895 & 2.368223 & 25.15520 & 2087.540  & 4613.964 \\ 
    \hline
    \multirow{2}{*}{$\texttt{Random}(4)$} & Fidelity  & 0.998850 & 0.998876 & 0.152971 & 0.984164  & 0.972877 \\
                                       & Time (secs)  & 0.713302 & 2.380326 & 26.18863 & 2185.091  & 4802.495 \\ 
    %%%%%%%%%%%%%%%%%%%%%%%%%%%%%%%%%%%%%%%%%%%%%
    \hline
    \multirow{2}{*}{$\texttt{GHZ}(5)$} & Fidelity  & 0.992105 & 0.992106 & 0.132725 & 0.274665  & 0.005138 \\
                                       & Time (secs)  & 0.946350 & 3.287358 & 395.3379 & $>$ 3h & $>$ 3h \\ 
    \hline
    \multirow{2}{*}{$\texttt{Hadamard}(5)$} & Fidelity  & 0.992102 & 0.992100 & 0.869603 & 0.998246  & 0.996516 \\
                                       & Time (secs)  & 1.183290 & 3.895312 & 79.39444 & 9319.140  & $>$ 3h \\ 
    \hline
    \multirow{2}{*}{$\texttt{Random}(5)$} & Fidelity  & 0.995126 & 0.995109 & 0.015913 & 0.623273  & 0.086777 \\
                                       & Time (secs)  & 0.988173 & 3.407487 & 79.22450 & 9275.836  & $>$ 3h \\ 
    %%%%%%%%%%%%%%%%%%%%%%%%%%%%%%%%%%%%%%%%%%%%%
    \hline
    \multirow{2}{*}{$\texttt{GHZ}(6)$} & Fidelity  & 0.984352 & 0.984340 & 0.089355 & 0.437323  & 0.310067 \\
                                       & Time (secs)  & 3.829866 & 13.306954 & 1167.985 & $>$ 3h  & $>$ 3h \\ 
    \hline
    \multirow{2}{*}{$\texttt{Hadamard}(6)$} & Fidelity  & 0.984384 & 0.984377 & 0.842515 & 0.990849 & 0.998077 \\
                                       & Time (secs)  & 2.500354 & 8.661999 & 246.0011 & $>$ 3h  & $>$ 3h \\ 
    \hline
    \multirow{2}{*}{$\texttt{Random}(6)$} & Fidelity  & 0.989543 & 0.989536 & 0.143145 & 0.784873 & 0.302534 \\
                                       & Time (secs)  & 1.991154 & 7.604232 & 237.7037 & $>$ 3h  & $>$ 3h \\ 
    %%%%%%%%%%%%%%%%%%%%%%%%%%%%%%%%%%%%%%%%%%%%%
    \hline
    \multirow{2}{*}{$\texttt{GHZ}(7)$} & Fidelity  & 0.969174 & 0.969168 & 0.058387 & 0.080648  & N/A \\
                                       & Time (secs)  & 6.174129 & 15.895504 & 3633.082 & $>$ 3h  & $>$ 3h \\ 
    \hline
    \multirow{2}{*}{$\texttt{Hadamard}(7)$} & Fidelity  & 0.969156 & 0.969156 & 0.818174 & 0.996586  & N/A \\
                                       & Time (secs)  & 6.324469 & 16.283301 & 713.9404 & $>$ 3h  & $>$ 3h \\ 
    \hline
    \multirow{2}{*}{$\texttt{Random}(7)$} & Fidelity  & 0.967640 & 0.967619 & 0.141745 & 0.06568  & N/A \\
                                       & Time (secs)  & 6.802577 & 16.594162 & 746.2630 & $>$ 3h  & $>$ 3h \\ 
    %%%%%%%%%%%%%%%%%%%%%%%%%%%%%%%%%%%%%%%%%%%%%
    \hline
    \multirow{2}{*}{$\texttt{GHZ}(8)$} & Fidelity  & 0.940601 & 0.940600 & 0.0400391 & N/A  & N/A \\
                                       & Time (secs)  & 21.16011 & 36.892739 & $>$ 3h & $>$ 3h  & $>$ 3h \\ 
    \hline
    \multirow{2}{*}{$\texttt{Hadamard}(8)$} & Fidelity  & 0.940638 & 0.940638 & 0.794892 & N/A  & N/A \\
                                       & Time (secs)  & 22.30246 & 41.472961 & 2344.796 & $>$ 3h  & $>$ 3h \\ 
    \hline
    \multirow{2}{*}{$\texttt{Random}(8)$} & Fidelity  & 0.939418 & 0.939416 & 0.050521 & N/A  & N/A \\
                                       & Time (secs)  & 22.81059 & 41.193810 & 2196.259 & $>$ 3h  & $>$ 3h \\ 
    \hline
\end{tabular}
\caption{%\textcolor{red}{Lyle: NEW TABLE}
Fidelity of reconstruction and computation timings using $\texttt{measpc}=50\%$ and $\texttt{shots}=2048$. Rows correspond to combinations of number of qubits (3 $\sim$ 8), final fidelity within the 3h time limit, and computation time. 
% Rows correspond to combinations of number of qubits (3 $\sim$ 8), synthetic circuit, and tomographic method (\texttt{MiFGD}, Qucumber's \texttt{CWF} and \texttt{DM}).  
For \texttt{MiFGD}, $\eta=0.001$, $\mu=\tfrac{3}{4}$, $\texttt{tol}=10^{-5}$. 
For \texttt{FGD}, $\eta=0.001$, $\texttt{tol}=10^{-5}$. 
% ``N/A'' indicates the procedure could not be completed due to memory limit. 
``N/A'' indicates that the method could not complete a single epoch in 3 hour training time limit, and thus could not provide any fidelity result.  
% \texttt{MiFGD} experiments are run on a 13'' Macbook Pro with 2.3 GHz Quad-Core Intel Core i7 CPU and $32$GB RAM. Qucumber experiments are run on a NVidia GeForce GTX 1080 TI, $11$GB RAM. 
All experiments are run on a NVidia GeForce GTX 1080 TI, $11$GB RAM. 
%\amir{It would be nice to add to this table the fidelities we get at the end of the reconstruction in Fig.5. From the figure it is not clear what fidelities we get. }
}
\label{table:qucumber_fidelity_time_2048}
\end{table}

\subsection{The effect of parallelization}

We study the effect of parallelization in running \texttt{MiFGD}.
We parallelize the iteration step across a number of processes, that can be either distributed and network connected, or sharing memory in a multicore environment. 
Our approach is based on Message Passing Interface (MPI) specification \cite{10.1145/169627.169855}, which is the lingua franca for interprocess communication in high performance parallel and supercomputing applications. 
A MPI implementation provides facilities for launching processes organized in a virtual topology and highly tuned primitives for point-to-point and collective communication between them.

We assign to each process a subset of the measurement labels consumed by the parallel computation. 
At each step, a process first computes the local gradient-based corrections due only to its assigned Pauli monomials and corresponding measurements. 
These local gradient-based corrections will then $(i)$ need to be communicated, so that they can be added, and $(ii)$ finally, their sum will be shared across all processes to produce a global update for $U$ for next step. 
We accomplish this structure in MPI using \texttt{MPI\_Allreduce} collective communication primitive with \texttt{MPI\_SUM} as its reduction operator: the underlying implementation will ensure minimum communication complexity for the operation (e.g. $\log p$ steps for $p$ processes organized in a communication ring) and thus maximum performance.\footnote{This communication pattern can alternatively be realized in two stages, as naturally suggested in its structure: $(i)$ first invoke MPI's \texttt{MPI\_Reduce} primitive, with \texttt{MPI\_SUM} as its reduction operator, which results in the element-wise accumulation of local corrections (vector sum) at a single, designated \emph{root} process, and $(ii)$ finally, send a ``copy'' of this sum from \emph{root} process to each process participating in the parallel computation (broadcasting); \texttt{MPI\_Bcast} primitive can be utilized for this latter stage. However, \texttt{MPI\_Allreduce} is typically faster, since its actual implementation is not constrained by the requirement to have the sum available at a specific, \emph{root} process, at an intermediate time point - as the two-stage approach implies.}
We leverage \texttt{mpi4py} \cite{dalcin2011parallel} bindings to issue MPI calls in our parallel Python code.

% 2 x E5-2690 v4 CPUs @ 2.60GHz (28/56 physical/virtual cores), 512 GB RAM.
We conducted our parallel experiments on a server equipped with 4 x E7-4850 v2 CPUs @ 2.30GHz (48/96 physical/virtual cores), 256 GB RAM, using shared memory multiprocessing over multiple cores. 
We experimented with states simulated in QASM and measured taking $8,192$ \texttt{shots};  parallel \texttt{MiFGD} runs with default parameters and using all measurements (\texttt{measpc}=100\%). 
Reported times are wall-clock computation time.
These exclude initialization time for all processes to load Pauli monomials and measurements: we here target parallelizing computation proper in \texttt{MiFGD}.
%\added[id=GK]{[Confirm server's configuration]}

In our first round of experiments, we investigate the scalability of our approach. We vary the number $p$ of parallel processes ($p=1, 2, 4, 8, 16, 32, 48, 64, 80, 96$), collect timings for reconstructing $\texttt{GHZ}(4)$, $\texttt{Random}(6)$ and $\texttt{GHZ}_{-}(8)$ states and report speedups $T_p/T_1$ we gain from \texttt{MiFGD} in Figure \ref{fig:parallel}(Left panel). 
We observe that the benefits of parallelization are pronounced for bigger problems (here: $n=8$ qubits) and maximum scalability results when we use all physical cores ($48$ in our platform).

Further, we move to larger problems ($n=10$ qubits: reporting on reconstructing $\texttt{Hadamard}(10)$ state) and focus on the effect parallelization to achieving a given level of fidelity in reconstruction. 
In Figure \ref{fig:parallel}(Middle panel) we illustrate the fidelity as a function of the time spent in the iteration loop of \texttt{MiFGD} for ($p=8, 16, 32, 48, 64$): we observe the smooth path to convergence in all $p$ counts which again minimizes compute time for $p=48$. Note that in this case we use \texttt{measpc} = 10\% and $\mu=\frac{1}{4}$.

Finally, in Figure \ref{fig:parallel}(Right panel), we fix the number of processes to $p=48$, in order to minimize compute time and increase the percentage of used measurements to $20\%$ of the total available for $\texttt{Hadamard}(10)$. 
We vary the acceleration parameter, $\mu=0$ (no acceleration) to $\mu=\frac{1}{4}$ and confirm that we indeed get faster convergence times in the latter case while the fidelity value remains the same (i.e. coinciding upper plateau value in the plots). 
We can also compare with the previous fidelity versus time plot, where the same $\mu$ but half the measurements are consumed: more measurements translate to faster convergence times (plateau is reached roughly $25\%$ faster; compare the green line with the yellow line in the previous plot).

\begin{figure*}[ht]
\centering
\includegraphics[width=0.32\textwidth]{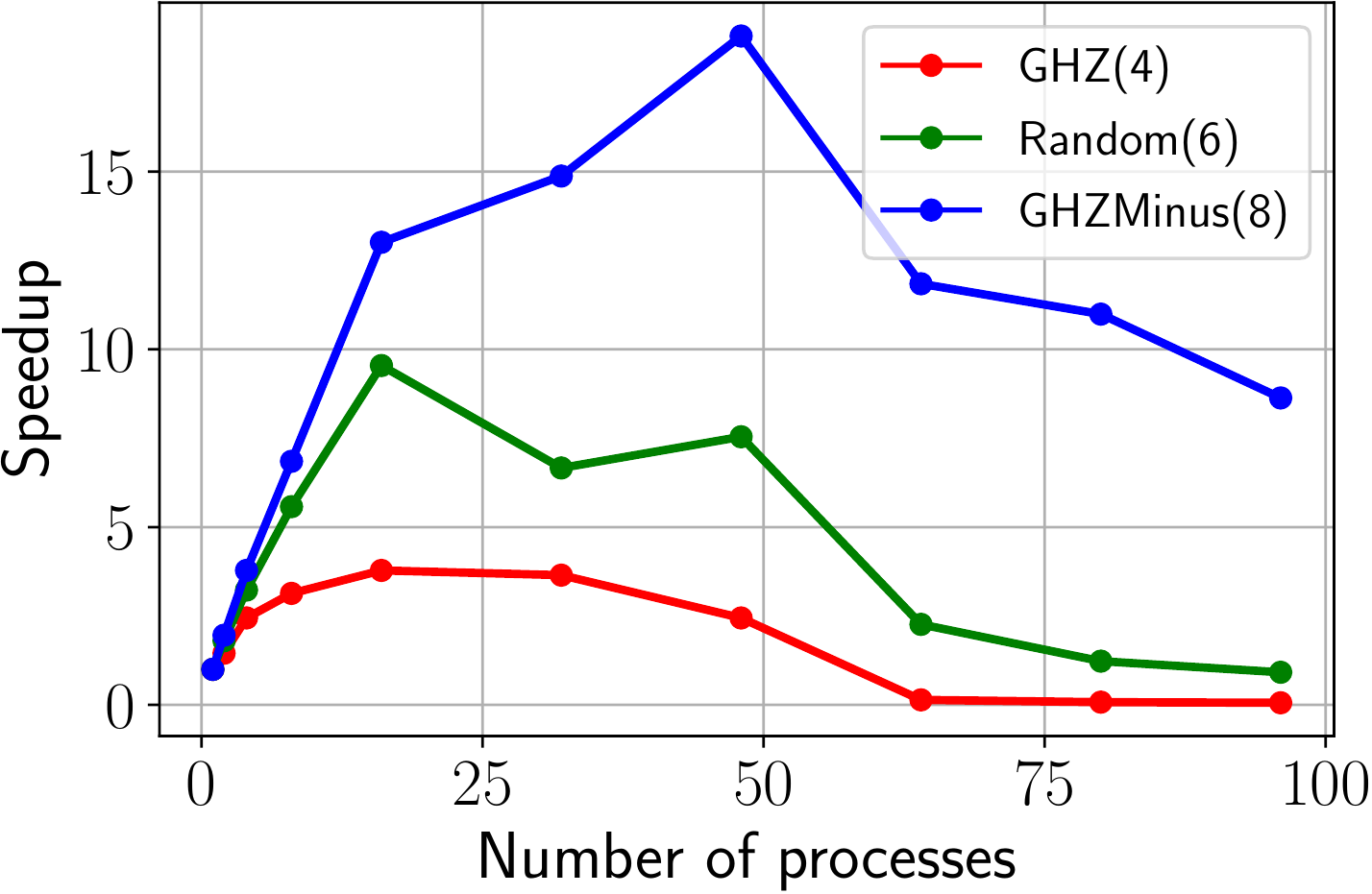} \includegraphics[width=0.31\textwidth]{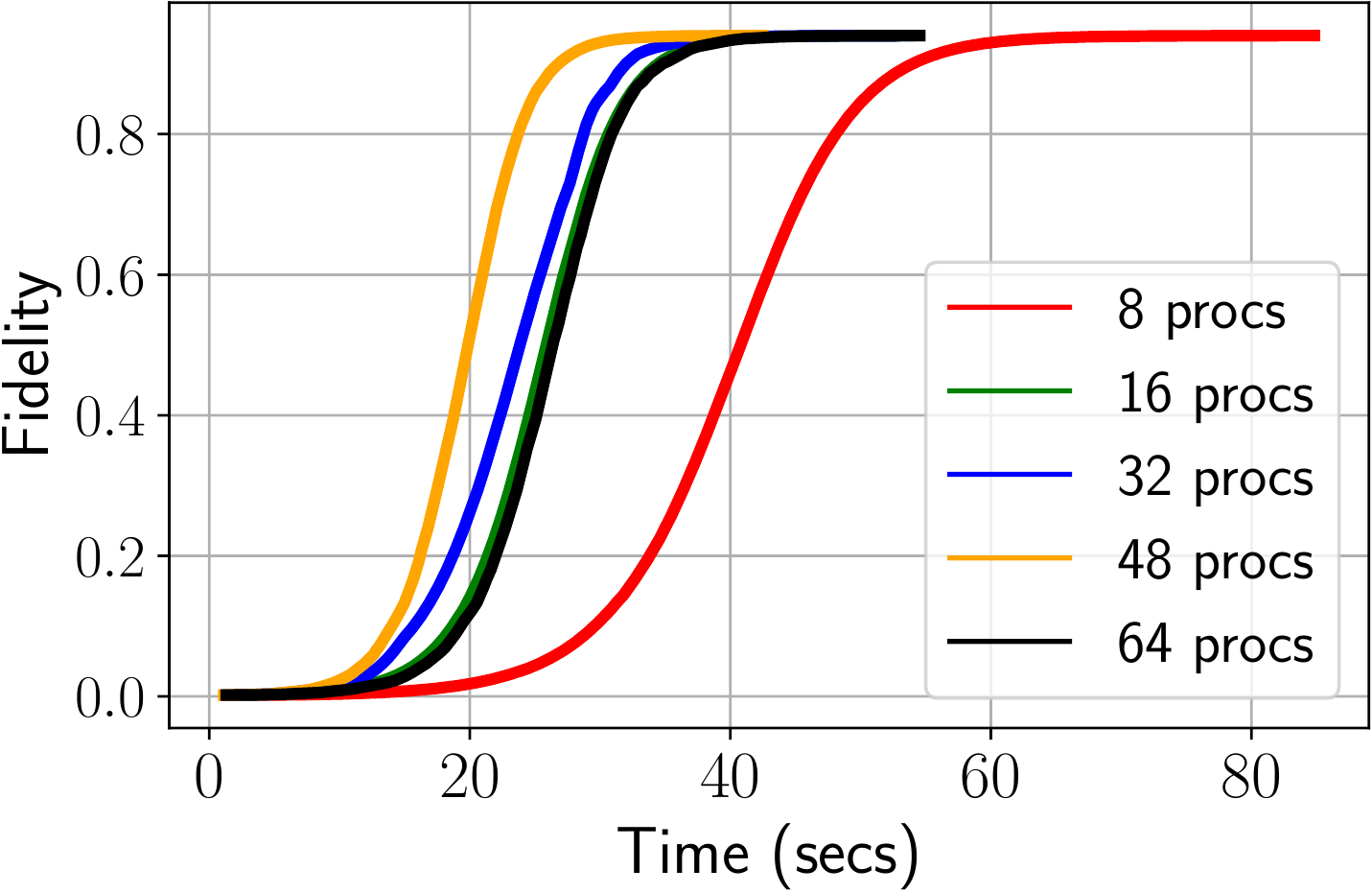}  
\includegraphics[width=0.31\textwidth]{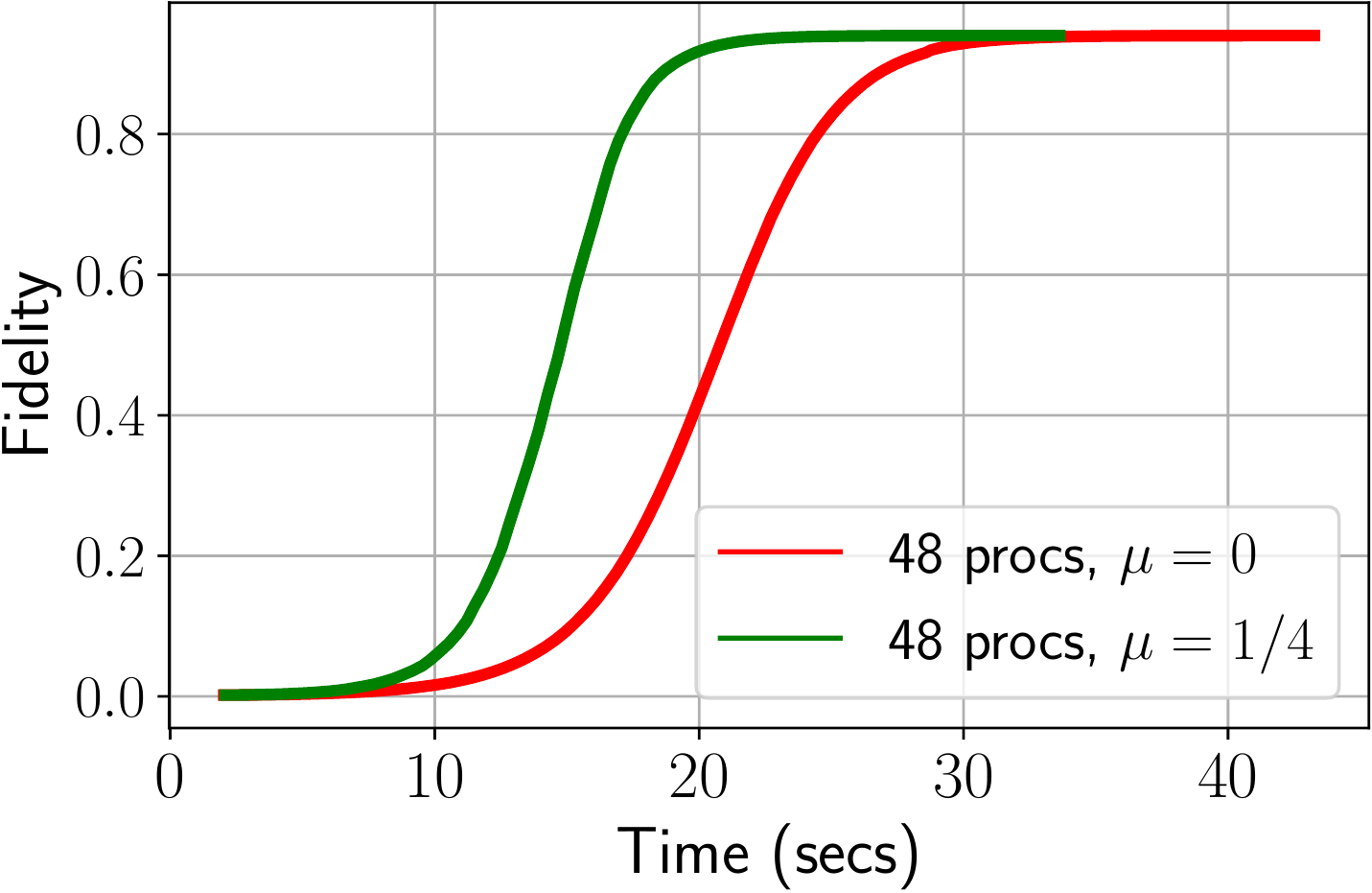}
 \caption{\textbf{Left panel}: Scalability of our approach as we vary the number $p$ of parallel processes. \textbf{Middle panel:} Fidelity function versus time consumed for different number of processes $p$. \textbf{Right panel:} The effect of momentum for a fixed scenario with $\texttt{Hadamard}(10)$ state, $p=48$, and varying momentum from $\mu = 0$ to $\mu = \tfrac{1}{4}$.}
 \label{fig:parallel}
\end{figure*}

%  \begin{figure*}[!h]
%   \begin{center}
%     \includegraphics[width=0.90\textwidth]{figures/parallel_speedup_procs.pdf}
%     \caption{
%   }
%   \label{fig:parallel_speedup_procs}
% \end{center}
% \end{figure*}

%  \begin{figure*}[!h]
%   \begin{center}
%     \includegraphics[width=0.90\textwidth]{figures/parallel_fidelity_time_10_hadamard_10pc_25mupc.pdf}
%     \caption{
%   }
%   \label{fig:parallel_fidelity_time_10_hadamard_10pc_25mupc}
% \end{center}
% \end{figure*}

%  \begin{figure*}[!h]
%   \begin{center}
%     \includegraphics[width=0.90\textwidth]{figures/parallel_fidelity_time_10_hadamard_20pc_48procs.pdf}
%     \caption{
%   }
%   \label{fig:parallel_fidelity_time_10_hadamard_20pc_48procs}
% \end{center}
% \end{figure*}
 
\section{Discussion}
We have introduced the \texttt{MiFGD} algorithm for the factorized form of the low-rank QST problems. 
We proved that, under certain assumptions on the problem parameters, \texttt{MiFGD} converges linearly to a neighborhood of the optimal solution, whose size depends on the momentum parameter $\mu$, while using acceleration motions in a non-convex setting.
We demonstrate empirically, using both simulated and real data, that \texttt{MiFGD} outperforms non-accelerated methods on both the original problem domain and the factorized space, contributing to recent efforts on testing QST algorithms in real quantum data \cite{riofrio2017experimental}.
These results expand on existing work in the literature illustrating the promise of factorized methods for certain low-rank matrix problems. 
Finally, we provide a publicly available implementation of our approach, compatible to the open-source software Qiskit \cite{qiskit}, where we further exploit parallel computations in \texttt{MiFGD} by extending its implementation to enable efficient, parallel execution over shared and distributed memory systems.

Despite our theory does not apply to the Pauli basis measurement  directly (i.e., \emph{using randomly selected Pauli bases $\Pi_\alpha$, does not lead to the $\ell_2$-norm RIP}), %\textcolor{magenta}{Tasos: I'm not sure we have these results here. Let me check with George whether we have any results for this. If not, we leave it as an open question.} we show that 
using the data from random Pauli basis measurements directly could provide excellent tomographic reconstruction with \texttt{MiFGD}.
%This setting corresponds to the case of \emph{rank-1 matrix sensing} \cite{cai2015rop, kueng2017low}.
Preliminary results suggest that only $O(r \cdot \log d)$ random Pauli bases should be taken for a reconstruction, with the same level of accuracy as with $O(r \cdot d \cdot \log d)$ expectation values of random Pauli matrices.
We leave the analysis of our algorithm in this case for future work, along with detailed experiments.
%This can be understood, in very general terms, from the fact that a Pauli basis measurement has $d$ outcomes.

\subsection{Related Work}\label{sec:related}
{\it Matrix sensing.} The problem of low-rank matrix reconstruction from few samples was first studied within the paradigm of convex optimization, using the nuclear norm minimization \cite{recht2010guaranteed, lee2009guaranteed, liu2009interior}. 
The use of non-convex approaches for low-rank matrix recovery---by imposing rank-constraints---has been proposed in \cite{jain2010guaranteed,lee2010admira,kyrillidis2014matrix}. 
In all these works, the convex and non-convex algorithms involve a full, or at least a truncated, singular value decomposition (SVD) per algorithm iteration. 
Since SVD can be prohibitive, these methods are limited to relatively small system sizes. 

Momentum acceleration methods are used regularly in the convex setting, as well as in machine learning practical scenarios  \cite{kingma2014adam, tieleman2012lecture, kingma2014adam, beck2009fast, o2015adaptive, bubeck2015geometric, goh2017why}. 
While momentum acceleration was previously studied in non-convex programming setups, it mostly involve non-convex constraints with a convex objective function \cite{kyrillidis2011recipes, kyrillidis2014matrix, khanna2017iht, xu2018accelerated}; and generic non-convex settings but only considering with the question of whether momentum acceleration leads to fast convergence to a saddle point or to a local minimum, rather than to a global optimum \cite{ghadimi2013stochastic, lee2016gradient, carmon2016accelerated, agarwal2016finding}.

The factorized version for semi-definite programming was popularized in \cite{burer2003nonlinear}. 
Effectively the factorization of a the set of PSD matrices to a product of rectangular matrices results in a non-convex setting. 
This approach have been heavily studied recently, due to  computational and space complexity 
advantages \cite{jain2013provable, chen2015fast, zhao2015nonconvex, zheng2015convergent, tu2016low, park2016provable, park2016non, sun2016guaranteed, bhojanapalli2016dropping, bhojanapalli2016global, park2016finding, ge2017no, hsieh2017non, kyrillidis2017provable}.
None of the works above consider the inclusion and analysis of momentum. 
Moreover, the \emph{Procrustes Flow} approach \cite{zheng2015convergent, tu2016low} uses certain initializations techniques, and thus relies on multiple SVD computations. 
Our approach on the other hand uses a single, unique, top-$r$ SVD computation. 
Comparison results beyond QST are provided in the appendix.

\medskip
\noindent
{\it Compressed sensing QST using non-convex optimization.} 
There are only few works that study non-convex optimization in the context of compressed sensing QST. 
The authors of \cite{shang2017superfast} propose a hybrid algorithm that $i)$ starts with a conjugate-gradient (CG) algorithm in the factored space, in order to get initial rapid descent, and $ii)$ switch over to accelerated first-order methods in the original $\rho$ space, provided one can determine the switch-over point cheaply. 
Using the multinomial maximum likelihood  objective, in the initial CG phase, the Hessian of the objective is computed per iteration (\emph{i.e.}, a $4^n{\times}4^n$ matrix), along with its eigenvalue decomposition. 
Such an operation is costly, even for moderate values of qubit number $n$, and heuristics are proposed for its completion. 
From a theoretical perspective, \cite{shang2017superfast} provide no convergence or convergence rate guarantees.

From a different perspective, \cite{o2016efficient} relies on spectrum estimation techniques \cite{hayashi2002quantum, christandl2006spectra} and the Empirical Young Diagram algorithm \cite{alicki1988symmetry, keyl2005estimating} to prove that $O(r d / \varepsilon)$ copies suffice to obtain an estimate $\hat{\rho}$ that satisfies $\|\hat{\rho} - \rho^\star\|_F^2 \leq \varepsilon$; however, to the best of our knowledge, there is no concrete implementation of this technique to compare with respect to scalability.

Ref. \cite{toth2010permutationally} proposes an efficient quantum tomography protocol by determining the permutationally invariant part of the quantum state. 
The authors determine the minimal number of local measurement  settings, which scales quadratically with the number of qubits. 
The paper determines which quantities have to be measured in order to get the smallest uncertainty possible.
See \cite{moroder2012permutationally} for a more recent work on permutationally invariant tomography.
The method has been tested in a six-qubit experiment in \cite{schwemmer2014experimental}.

Ref. \cite{riofrio2017experimental} presented an experimental implementation of compressed sensing QST of a $n = 7$ qubit system, where only $127$ Pauli basis measurements are available. 
To achieve recovery in practice, the authors proposed a computationally efficient estimator, based on gradient descent method  in the factorized space. 
The authors of~\cite{riofrio2017experimental}  focus on the experimental efficiency of the method, and provide no specific results on the optimization efficiency, neither convergence guarantees of the algorithm. Further, there is no available implementation publicly available.

Similar to~\cite{riofrio2017experimental}, Ref.~\cite{kyrillidis2017provable} also proposes a non-convex projected gradient decent algorithm that works on the factorized space in the QST setting. 
The authors prove a rigorous convergence analysis and show that,  under proper initialization and step-size, the algorithm is guaranteed to converge to the global minimum of the problem, thus  ensuring a provable tomography procedure. 
\emph{Our results extend these results by including acceleration techniques in the factorized space}. 
The key contribution of our work is proving convergence of the proposed algorithm in a {\it linear} rate to the global minimum of the problem, under common assumptions. 
Proving our results required developing a whole set of new techniques, which are not based on a mere extension of existing results.

\medskip
\noindent {\it Compressed sensing QST using convex optimization.}
The original formulation of compressed sensing QST~\cite{gross2010quantum} is based on  convex optimization methods, solving the trace-norm minimization, to obtain an estimation of the low-rank state.
It was later shown \cite{kalev2015quantum} that essentially any convex optimization program can be used to robustly estimate the state. 
In general, there are two drawbacks in using convex optimization optimization in QST. 
Firstly, as the dimension of density matrices grow exponentially in the number of qubits, the search space in convex optimization grows exponentially in the number of qubits. 
Secondly, the optimization requires projection onto the PSD cone at every iteration, which becomes exponentially hard in the number of qubits. 
\emph{We avoid these two drawbacks by working in the factorized space}. 
Using this factorization results in a search space that is substantially smaller than its convex counterpart, and moreover, in a single use of top-$r$ SVD during the entire execution algorithm.
Bypassing  these drawbacks, together with accelerating motions, allows us to estimate quantum states of larger qubit systems than state-of-the-art algorithms.

\medskip
\noindent {\it Full QST using non-convex optimization.} 
The use of non-convex algorithms in QST was studied in the context of full tomography as well. 
By ``full tomography'' we refer to the situation where an informationally complete measurement is performed, so that the input data to the algorithm is of size $4^n$. 
The exponential scaling of the data size restrict the applicability of full tomography to relatively small system sizes. 
In this setting non-convex algorithms which work in the factored space were studied \cite{banaszek1999maximum, paris2001maximum, rehacek2007diluted, gonccalves2012local, siah2013informationally}.
Except of the work \cite{gonccalves2012local}, we are not aware of  theoretical results on the convergence of the proposed algorithm due to the presence of spurious local minima. 
The authors of~\cite{gonccalves2012local}  characterize the local vs. the global behavior of the objective function under the factorization $\rho = UU^\dagger$ and discuss how existing methods fail due to improper stopping criteria or due to the lack of algorithmic convergence results. 
Their work highlights the lack of rigorous convergence results of non-convex algorithms used in full quantum state tomography.
There is no available implementation publicly available for these methods as well.

\medskip
\noindent {\it Full QST using convex optimization.}
Despite the non-scalability of full QST, and the limitation of convex optimization, a lot of research was devoted to this topic. 
Here, we mention only a few notable results that extend  the applicability of full QST using specific techniques in convex optimization. 
Ref~\cite{smolin2012efficient} shows that for given measurement schemes the solution for the maximum likelihood is given by a linear inversion scheme, followed by a projection onto the set of density matrices. 
%The computational complexity of the linear inversion is  $O(2^{4n})$, while the projection onto the set of density matrices is  $O(2^{3n})$.  
More recently, the authors of~\cite{hou2016full} used a combination of the techniques of~\cite{smolin2012efficient} with the sparsity of the Pauli matrices and the use of GPUs to perform a full QST of 14 qubits. 
While pushing the limit of full QST using convex optimization, obtaining full tomographic {\it experimental} data for more than a dozen qubits is significantly time-intensive. 
Moreover, this approach is highly centralized, in comparison to our approach that can be distributed. Using the sparsity pattern property of the Pauli matrices and GPUs is an excellent candidate approach to further enhance the performance of non-convex compressed sensing QST.

%quantum state tomography.}{\it Methods for learning quantum states.}~{\color{blue} Methods for learning quantum states, e.g. by using machine learning or "shadow tomography"-related, have a much more favorable scaling than tomography. Maybe it's better to leave out that section, and comment on it only if a referee asks? }\\

\medskip
\noindent {\it QST using neural networks.}
Deep neural networks are ubiquitous, with many applications to science and industry.  
%In fact, there are few scientific disciplines or industrial concerns that have \emph{not} been touched by deep learning.
%Neural networks are used in physics \cite{cheng2018recursive, arjona2018arxiv},
%astronomy \cite{gabbard2018matching, khalifa2018deep},
%healthcare \cite{amyar2018radiomics}, 
%biology \cite{camacho2018next},
%software engineering \cite{murali2017bayesian},
%computer security \cite{rhode2018early}, among others.
Recently, \cite{Torlai2018Neural, beach2019qucumber, torlai2019machine, gao2017efficient} show how machine learning and neural networks can be used to perform QST, driven by experimental data. 
The neural network architecture used is based on restricted Boltzmann machines (RBMs) \cite{sutskever2009recurrent}, which feature a visible and a hidden layer of stochastic binary neurons, fully connected with weighted edges. 
Test cases considered include reconstruction of W state, magnetic observables of local Hamiltonians, the unitary dynamics induced by Hamiltonian evolution. 
Comparison results are provided in the Main Results section. 
Alternative approaches include conditional generative adversarial networks (CGANs) \cite{ahmed2020quantum, ahmed2020classification}: in this case, two dueling neural networks, a generator and a discriminator, learn to generate and identify multi-modal models from data. 

\medskip
\noindent {\it QST for Matrix Product States (MPS).}
This is the case of highly structured quantum states where the state is well-approximated by a MPS of low bond dimension \cite{Cramer2010Efficient, lanyon2017efficient}.
The idea behind this approach is, in order to overcome exponential bottlenecks in the general QST case, we require highly structured subsets of states, similar to the assumptions made in compressed sensing QST. 
MPS QST is considered an alternative approach to reduce the computational and storage complexity of QST.
% \textcolor{red}{Lyle: this is the only place where we use ``ground truth''. I think we can replace with ``target state'' or ``unknown state''}

\medskip
\noindent {\it Direct fidelity estimation.} Rather than focusing on entrywise estimation of density matrices, the direct fidelity estimation procedure focuses on checking how close is the state of the system to a target state, where closeness is quantified by the fidelity metric. Classic techniques require up to $2^n / \epsilon^4$ number of samples, where $\epsilon$ denotes the accuracy of the fidelity term, when considering a general quantum state \cite{paini2019approximate,huang2020predicting}, but can be reduced to almost dimensionality-free $1/\epsilon^2$ number of samples for specific cases, such as stabilizer states \cite{flammia2011direct, da2011practical, kalev2019validating}. Shadow tomography is considered as an alternative and generalization of this technique \cite{aaronson2018shadow, aaronson2019gentle}; however, as noted in \cite{huang2020predicting}, the procedure in \cite{aaronson2018shadow, aaronson2019gentle} requires exponentially long quantum circuits that act collectively on all the copies of the unknown state stored in a quantum memory, and thus has not been implemented fully on real quantum machines.
A recent neural network-based implementation of such indirect QST learning methods is provided here \cite{smith2020efficient}. 

The work in \cite{paini2019approximate,huang2020predicting}, goes beyond simple fidelity estimation, and utilizes random single qubit rotations to learn a minimal sketch of the unknown quantum state by which one that can predict arbitrary linear function of the state. Such methods constitute a favorable alternative to QST as they do not require number of samples that scale polynomially with the dimension; however, this, in turn, implies that these methods cannot be used in general to estimate the density matrix itself.

\section{Methods}
\noindent \textbf{\texttt{MiFGD} algorithm.}
\begin{algorithm}
\begin{algorithmic}
   \STATE {\bfseries Input:} $\mathcal{A}$, $y$,  $r$, $\mu$, and  $\#$ iterations $J$.
   \STATE Set $U_0$ randomly or as in ~\eqref{lem:init}.
\STATE Set  $Z_0=U_0$.
 \STATE Set $\eta$ as in~\eqref{eq:step}.
\FOR{$i=0$ to $J-1$}
\STATE  $U_{i+1} = Z_i  - \eta  \mathcal{A}^\dagger \left( \mathcal{A}(ZZ^\dagger) - y \right) \cdot Z_i$ \\
\STATE   $Z_{i+1} = U_{i+1} + \mu \left(U_{i+1} - U_i\right)$
\ENDFOR
 \STATE {\bfseries Output:} $\rho=U_{J} U_{J}^\dagger$
\end{algorithmic}
\caption{Momentum-Inspired Factored Gradient Descent} \label{alg:algo1}
\vspace{-0.1cm}
\end{algorithm}
Algorithm~\ref{alg:algo1} contains the details of the Momentum-Inspired Factored Gradient Descent. %The algorithm requires as input the target rank\footnote{In this work, we assume we know the target rank \emph{a priori}. For the cases where we undershoot the rank, the theory from \cite{park2016finding} can be used to extend our theory.} $r$, the number of iterations $J$, and the momentum parameter $\mu$. 
%For our theory to hold, we make the following selections.
%$i)$ The $\mu$ selection is conservative as we show next, but more aggressive $\mu$ choices lead to different requirements in our theory.
The initial point $U_0$ is either randomly selected \cite{bhojanapalli2016global, park2016non}, or set according to Lemma~4 in~\cite{kyrillidis2017provable} to:
\begin{align}\label{lem:init}
\rho_0 = U_0 U_0^\dagger = \Pi_{\mathcal{C}}\big(\tfrac{-1}{1 + \delta_{2r}} \cdot \nabla f(0) \big)=\tfrac{1}{1 + \delta_{2r}}\Pi_{\mathcal{C}}\bigg(\sum_{i=1}^m y_i A_i\bigg)
\end{align}
where $\Pi_{\mathcal{C}}(\cdot)$ is the projection onto the set of PSD matrices, $\delta_{2r}\in(0,1)$ is the RIP constant and $\nabla f(0)$ denoted the gradient of $f$ evaluated at the all-zero matrix. 
Since computing the RIP constants is NP-hard, in practice we compute  $U_0$ through $\rho_0 = \tfrac{-1}{ \widehat{L}} \Pi_{\mathcal{C}}\big(\sum_{i=1}^m y_i A_i\big)$, where $\widehat{L} \in (1,\sfrac{11}{10}]$, see Theorem~\ref{thm:00} below. 
Compared to randomly selecting $U_0$, Eq.~ \ref{lem:init} involves a gradient descent computation and a top-$r$ eigenvalue calculation. %\amir{why GD and why top-r?}.
%While randomly selecting $U_0$ guarantees convergence \cite{bhojanapalli2016global,park2016non}, Eq.~ \ref{lem:init} provides the initial conditions for our theory also leads to convergence rate guarantees. 
As for the step size in algorithm~\ref{alg:algo1}, following Lemma~\ref{lem:equiveta} below,  it is set to
\begin{align}
\eta = \tfrac{1}{4 \left( (1 + \delta_{2r})\|Z_0Z_0^\dagger\|_2 +\|\mathcal{A}^\dagger \left(\mathcal{A}(Z_0Z_0^\dagger) - y \right)\|_2 \right)}, \label{eq:step}
\end{align}
where $Z_0 = U_0$. Here as well, in practice we replace $\delta_{2r}$ by $ \widehat{L}$. The step size $\eta$ remains constant at every iteration step of the algorithm, and requires only two top-eigenvalue computations to calculate the spectral norms $\|Z_0Z_0^\dagger\|_2$ and $\|\mathcal{A}^\dagger \big( \mathcal{A}(Z_0 Z_0^\dagger - y\big)\|_2$.  These computations can be efficiently implemented by any off-the-shelf eigenvalue solver, such as the Power Method or the Lanczos method.\\
\noindent
\medskip
We now present the formal convergence theorem, where under certain conditions, \texttt{MiFGD} achieves an accelerated linear rate.

\begin{theorem}[Accelerated convergence rate]{\label{thm:00}}
Assume that $\mathcal{A}$ satisfies the RIP with constant $\delta_{2r} \leq \sfrac{1}{10}$.
Let $U_0$ and $U_{-1}$ be such that \begin{small}$\min_{R \in \mathcal{O}} \|U_0 -U^\star R\|_F, ~\min_{R \in \mathcal{O}} \|U_{-1} -U^\star R\|_F \leq \tfrac{\sqrt{\sigma_r(\rho^\star)}}{10^3 \sqrt{\kappa\tau(\rho^\star)}}$\end{small},
where \begin{small}$\kappa := \tfrac{1 + \delta_{2r}}{1 - \delta_{2r}}$\end{small}, \begin{small}$\tau(\rho) := \tfrac{\sigma_1(\rho)}{\sigma_r(\rho)}$\end{small} for rank-$r$ $\rho$, and \begin{small}$\sigma_i(\rho)$\end{small} is the $i$th singular value of $\rho$. 
Set step size $\eta$ such that 
\begin{align*}
    \left[ 1 - \left( \tfrac{ \sqrt{1+\delta_{2r}} - \sqrt{1-\delta_{2r}} }{(\sqrt{2}+1) \sqrt{1+\delta_{2r}}} \right)^4 \right] \cdot \tfrac{10}{4 \sigma_r (\rho^\star) (1-\delta_{2r})} \leq \eta \leq \tfrac{10}{4 \sigma_r (\rho^\star) (1-\delta_{2r})},
\end{align*}
and the momentum parameter \begin{small}$\mu = \frac{\varepsilon}{2\cdot10^3 r  \tau(\rho^\star) \sqrt{\kappa}},$\end{small} for user-defined $\varepsilon \in (0, 1]$.
% For $y = \mathcal{A}(\rho^\star)$, where \texttt{rank}$(\rho^\star)=r$, if $\tau(\rho^\star)\leq{50}$, then \texttt{MiFGD} returns a solution such that
For $y = \mathcal{A}(\rho^\star)$ where \texttt{rank}$(\rho^\star)=r,$ \texttt{MiFGD} returns a solution such that
\begin{small}
\begin{align}
\min_{R \in \mathcal{O}} \|U_{J+1} - U^\star R\|_F &\leq
        \left(1-\sqrt{ \tfrac{1-\delta_{2r}}{1+\delta_{2r}} }\right)^{J+1} \left( \min_{R \in \mathcal{O}} \|U_0 - U^\star R\|_F^2 + \min_{R \in \mathcal{O}} \|U_{-1} - U^\star R\|_F^2 \right)^{1/2} \nonumber \\
        &\hspace{1.4cm} +
\xi \cdot |\mu| \cdot \sigma_1(\rho^\star)^{1/2} \cdot r \cdot \left( 1 - \left(1-\sqrt{ \tfrac{1-\delta_{2r}}{1+\delta_{2r}} }\right)^{J+1} \right) \left(1-\sqrt{ \tfrac{1-\delta_{2r}}{1+\delta_{2r}} }\right)^{-1}  \nonumber \\
&\lessapprox \left(1-\sqrt{ \tfrac{1-\delta_{2r}}{1+\delta_{2r}} }\right)^{J+1} \left( \min_{R \in \mathcal{O}} \|U_0 - U^\star R\|_F^2 + \min_{R \in \mathcal{O}} \|U_{-1} - U^\star R\|_F^2 \right)^{1/2} + O(\mu), \label{eq:recursion}
\end{align}
\end{small}
where
$\xi = \sqrt{1 - \tfrac{4\eta \sigma_r(\rho^\star)(1 - \delta_{2r})}{10}}$.
That is, the algorithm has an accelerated linear convergence rate in iterate distances up to a constant proportional to the  momentum parameter $\mu$. 
% (second term on the right hand side of Eq.~\eqref{eq:recursion}). 
% Further, $U_i$ satisfies \begin{small}$\min_{R \in \mathcal{O}} \|U_i - U^\star R\|_F \leq \tfrac{\sqrt{\sigma_r(\rho^\star)}}{10^3 \sqrt{\kappa\tau(\rho^\star)}}$\end{small}, for each $i \geq 1$, i.e., the algorithm preserves the per-iteration conditions to continue decreasing the distance to the optimum.
\end{theorem}

The interpretation of the theorem is that the right hand side of Eq.~\eqref{eq:recursion} depends on the initial distances $\min_{R \in \mathcal{O}} \|U_0 - U^\star R\|_F$ and $\min_{R \in \mathcal{O}} \|U_{-1} - U^\star R\|_F$ as in convex optimization, where $\left(1-\sqrt{ \tfrac{1-\delta_{2r}}{1+\delta_{2r}} }\right)$ appear as a contraction constant. In contrast, the contraction factor of vanilla FGD \cite{kyrillidis2017provable} is of the form $\left( 1 - \frac{1-\delta_{2r}}{1+\delta_{2r}}\right)$, ignoring some constants.

As we assume the sensing map $\mathcal{A}$ satisfies RIP, the condition number of $f$ depends on the RIP constants $\delta$ such that $\frac{L}{\mu} \propto \frac{1+\delta}{1-\delta},$ since the eigenvalues of the Hessian of $f$,  $\mathcal{A}^\dagger A(\cdot),$ lie between $1-\delta$ and $1+\delta,$ when restricted to low-rank matrices.  
Hence, \texttt{MiFGD} has better dependency on the (inverse) condition number of $f$ than FGD. 
Such improvement of the dependency on the condition number is called ``acceleration'' in convex optimization \cite{nesterov2013introductory, carmon2016accelerated}.

% The interpretation of the theorem is that the right hand side of Eq.~\eqref{eq:recursion} depends on the initial distance $\min_{R \in \mathcal{O}} \|U_0 - U^\star R\|_2$, as in convex optimization, where the two parameters $c, \alpha$ appear as contraction constants. 
% Crucially, since $c<1$ the product $\alpha c^{J}$ goes exponentially fast to zero.
Thus, assuming that the initial points $U_0$ and $U_{-1}$ are close enough to the optimum, as stated in the theorem, \texttt{MiFGD} decreases its distance to $U^\star$ with an accelerated linear rate, up to an ``error'' level that depends on the momentum parameter $\mu$ and it is bounded by \begin{small}$\frac{1}{2\cdot10^3 r  \tau(\rho^\star) \sqrt{\kappa}}$.\end{small}

Theorem~\ref{thm:00} requires a strong assumption on the momentum parameter $\mu$, which depends on quantities that might not be known \emph{a priori} for general problems. 
However, we note that for the special case of QST, we know these quantities exactly: 
$r$ is the rank of density matrix---thus, for pure states this value is equal to $r=1$; $\tau(\rho^\star)$ is the (rank-restricted) condition number of the density matrix $\rho$---for pure states, $\tau(\rho^\star) = \tfrac{\sigma_1(\rho)}{\sigma_r(\rho)} = \tfrac{\sigma_1(\rho)}{\sigma_1(\rho)} = 1$; and, finally, $\kappa$ is the condition number of the sensing map, where, given the constraint $\delta_{2r} \leq \sfrac{1}{10}$, leads to the following bound: $\kappa \leq \sfrac{11}{9}$.
The above lead to a momentum value $\mu \approx \sfrac{\varepsilon}{2211}$.
However, as we show in the numerical experiments, the theory is conservative; much larger values of $\mu$ lead to stable, improved performance.
Finally, the bound on the condition number in Theorem \ref{thm:00} is not strict, and comes out of the analysis we follow; we point the reader to similar assumptions made where $\tau(\rho^\star)$ is assumed constant $O(1)$ \cite{li2019nonconvex}. 

%$\alpha$ depends on the spectral gap of the contraction matrix $A$, as well as on $\xi$, where the latter depends on the condition number of the objective and the condition number of $\rho^\star$.

%Finally, the theory generates an additional term $O(\mu)$. Similar results exist in the literature (see, for example, the constant step size convergence of convex SGD, where one achieves linear convergence up to an error level that depends on the step size).  

The detailed proof is provided in the supplementary material.
To the best of our knowledge, this is the first proof for momentum-inspired factorization technique, under common assumptions: both regarding the problem setting, and the assumptions made for its completion. 
The proof differs from state of the art proofs for non-accelerated factored gradient descent: due to the inclusion of the memory term, three different terms --$U_{i+1}, U_i, U_{i-1}$-- need to be handled simultaneously. 
Further, the proof differs from recent proofs on non-convex, but non-factored, gradient descent methods, as in \cite{khanna2017iht}: the distance metric over rotations $\min_{R \in \mathcal{O}} \|Z_i - U^\star R\|_F$, where $Z_i$ includes estimates from two steps in history, is not amenable to simple triangle inequality bounds, and a careful analysis is required.
The analysis requires the design of two-dimensional dynamical systems, where we require to characterize and bound the eigenvalues of a $2\times 2$ contraction matrix.

\section{Acknowledgements}
Anastasios Kyrillidis and Amir Kalev acknowledge funding by the NSF (CCF-1907936).
Anastasios Kyrillidis thanks Mads Mikkelsen for his performance in ``Jagten'' and ``Druk''.

\section{Data availability}
The empirical results were obtained via synthetic and real experiments; the algorithm's implementation is available at \href{https://github.com/gidiko/MiFGD}{https://github.com/gidiko/MiFGD}.
\section{Competing interests}
The authors declare no competing financial or non-financial interests.
\section{Author contribution}
All authors have made substantial contributions to the paper: design of the work, drafting the manuscript, final approval and accountability for all aspects of the work.

\bibliographystyle{plain}
\bibliography{biblio}
%\end{small}
%\end{quote}

\clearpage
\onecolumn

\newpage
\section{Appendix}

\subsection{IBM Quantum system experiments: $\texttt{GHZ}_{-}(6)$ circuit, 2048 \texttt{shots}}

\begin{figure*}[!h]
  \begin{center}
    \includegraphics[width=1\textwidth]{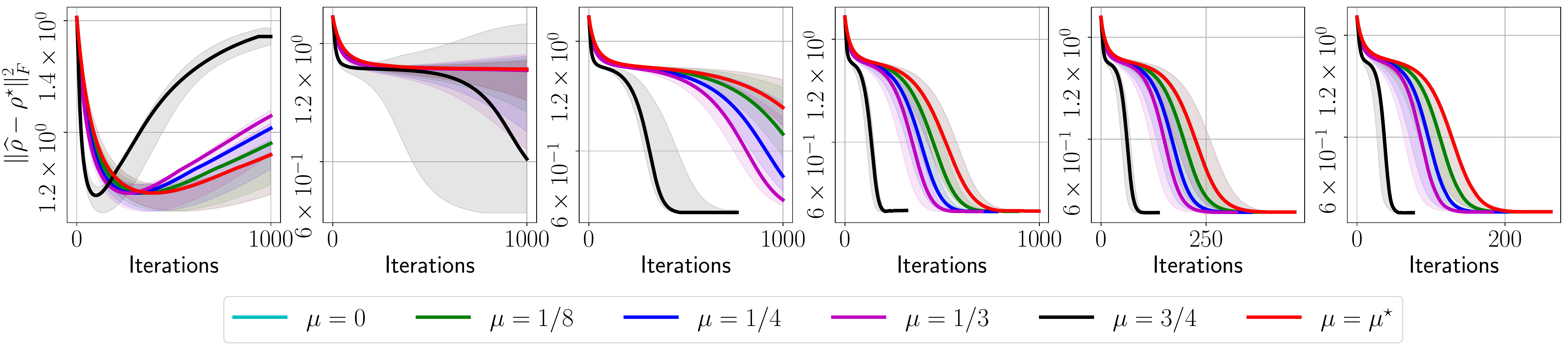}  \vspace{-0.9cm}
    \caption{
      %\emph{Left panel:} $\| \widehat{\rho} - \rho^\star \|_F$ vs. \# of iterations for different $\mu$;
      %\emph{Middle panel:} \# of iterations to reach $\texttt{reltol}$ vs. $\mu$ and for different circuits $(\rho^\star)$;
      %\emph{Right panel:} Fidelity of $\widehat{\rho}$, defined as $\text{Tr}\big(\sqrt{ \sqrt{\rho^\star} \widehat{\rho} \sqrt{\rho^\star}}\big)^2$, vs. $\mu$ %parameters and for different circuits $(\rho^\star)$. Shaded area denotes standard deviation around the mean over repeated runs in all cases.
      Target error list plots for reconstructing $\texttt{GHZ}_{-}(6)$ circuit using real measurements from IBM Quantum system experiments.
      }
    \label{temp}
  \end{center}
\end{figure*}
\vspace{-0.9cm}
\begin{figure*}[!h]
  \begin{center}
    \includegraphics[width=1\textwidth]{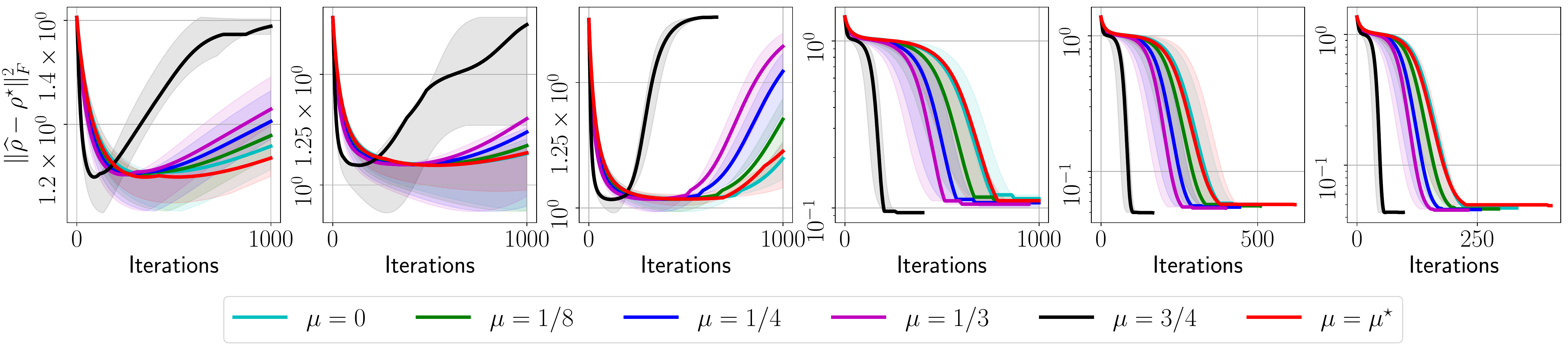}   \vspace{-0.9cm}
    \caption{
      Target error list plots for reconstructing $\texttt{GHZ}_{-}(6)$ circuit using synthetic measurements from IBM's quantum simulator.
      }
    \label{temp}
  \end{center}
\end{figure*}
\vspace{-0.9cm}
\begin{figure*}[!h]
  \begin{center}
    \includegraphics[width=1\textwidth]{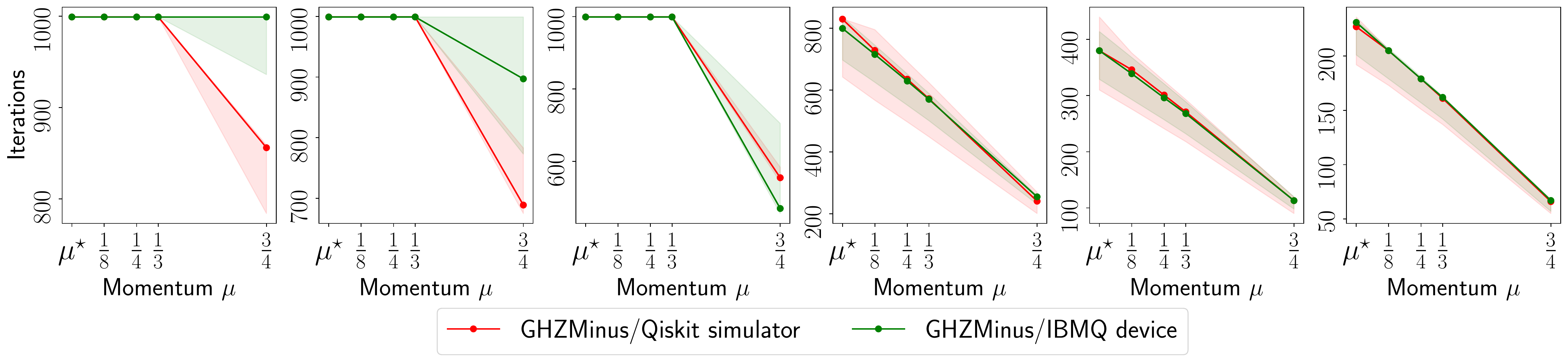}  
    \vspace{-0.9cm}
    \caption{
      Convergence iteration plots for reconstructing $\texttt{GHZ}_{-}(6)$ circuit using using real measurements from IBM Quantum system experiments and synthetic measurements from Qiskit simulation experiments.
      }
    \label{fig:convergence_iteration_6_ghzminus_2048}
  \end{center}
\end{figure*}
\vspace{-0.9cm}
\begin{figure*}[!h]
  \begin{center}
    \includegraphics[width=1\textwidth]{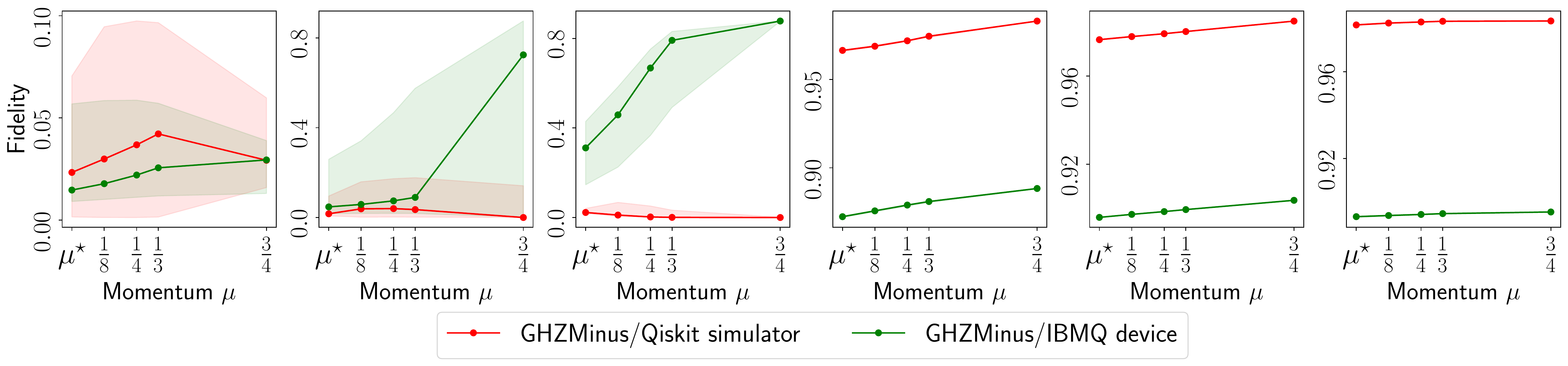}
    \vspace{-0.9cm}
    \caption{
      Fidelity list plots for reconstructing $\texttt{GHZ}_{-}(6)$ circuit using using real measurements from IBM Quantum system experiments and synthetic measurements from Qiskit simulation experiments.
      }
    \label{fig:fidelity_list_6_ghzminus_2048}
  \end{center}
\end{figure*}

\newpage 
\subsection{IBM Quantum system experiments: $\texttt{GHZ}_{-}(6)$ circuit, 8192 \texttt{shots}} 

\begin{figure*}[!h]
  \begin{center}
    \includegraphics[width=1\textwidth]{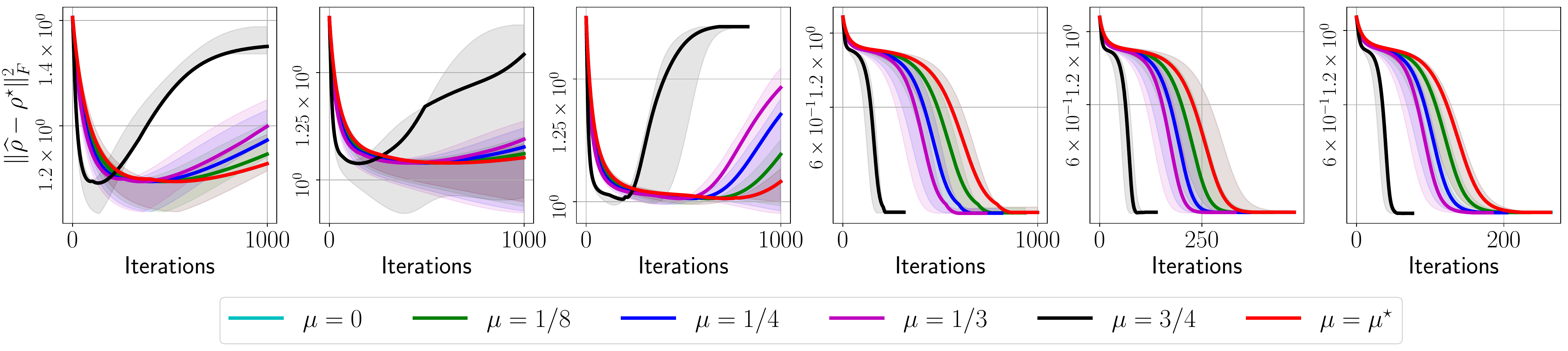}  \vspace{-0.9cm}
    \caption{
      Target error list plots for reconstructing $\texttt{GHZ}_{-}(6)$ circuit using real measurements from IBM Quantum system experiments.
      }
    \label{temp}
  \end{center}
\end{figure*}
\vspace{-0.9cm}
\begin{figure*}[!h]
  \begin{center}
    \includegraphics[width=1\textwidth]{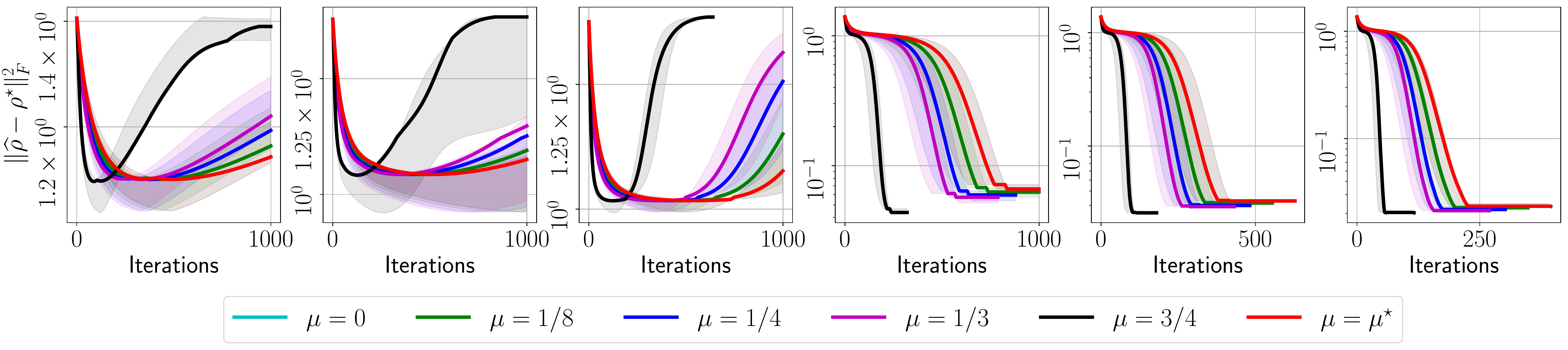}   \vspace{-0.9cm}
    \caption{
      Target error list plots for reconstructing $\texttt{GHZ}_{-}(6)$ circuit using synthetic measurements from IBM's quantum simulator.
      }
    \label{temp}
  \end{center}
\end{figure*}
\vspace{-0.9cm}
\begin{figure*}[!h]
  \begin{center}
    \includegraphics[width=1\textwidth]{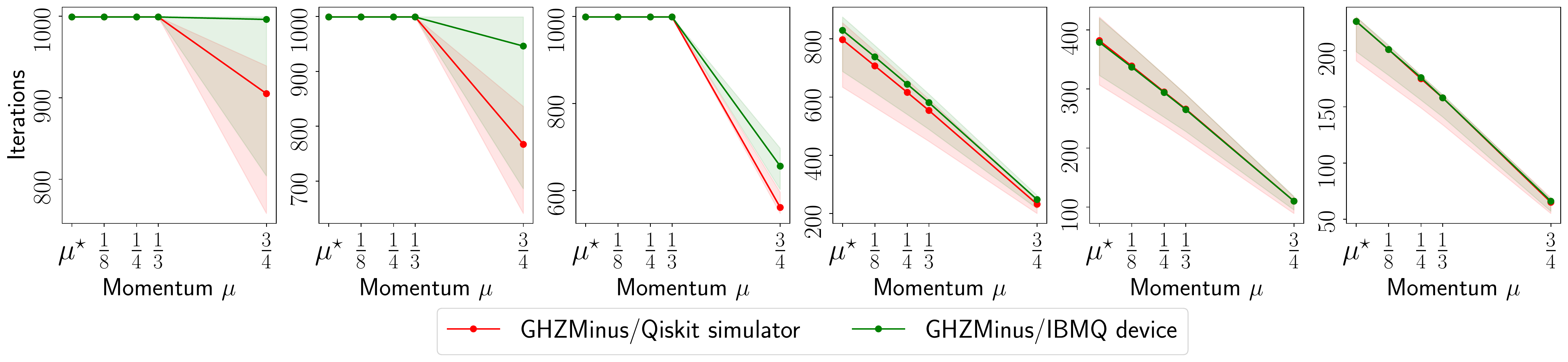}  
    \vspace{-0.9cm}
    \caption{
      Convergence iteration plots for reconstructing $\texttt{GHZ}_{-}(6)$ circuit using using real measurements from IBM Quantum system experiments and synthetic measurements from Qiskit simulation experiments.
      }
    \label{fig:convergence_iteration_6_ghzminus_8192}
  \end{center}
\end{figure*}
\vspace{-0.9cm}
\begin{figure*}[!h]
  \begin{center}
    \includegraphics[width=1\textwidth]{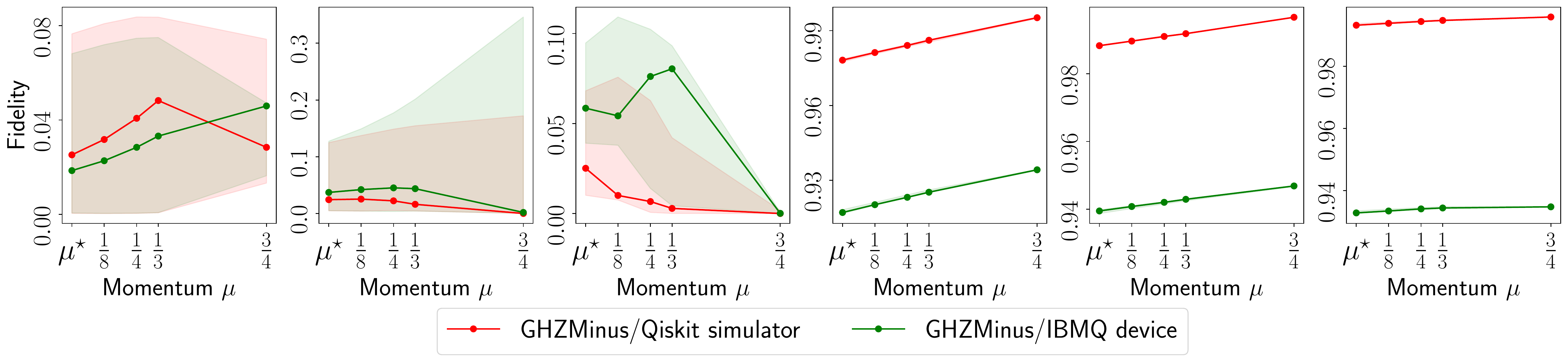}  
    \vspace{-0.9cm}
    \caption{
      Fidelity list plots for reconstructing $\texttt{GHZ}_{-}(6)$ circuit using using real measurements from IBM Quantum system experiments and synthetic measurements from Qiskit simulation experiments.
      }
    \label{fig:fidelity_list_6_ghzminus_8192}
  \end{center}
\end{figure*}

\newpage 
\subsection{IBM Quantum system experiments: $\texttt{GHZ}_{-}(8)$ circuit, 2048 \texttt{shots}} 

\begin{figure*}[!h]
  \begin{center}
    \includegraphics[width=1\textwidth]{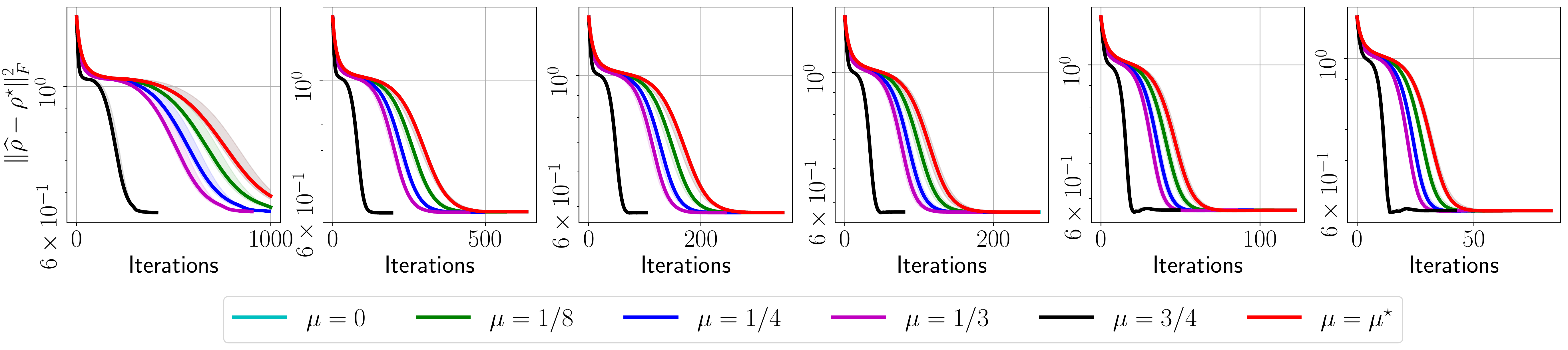}  \vspace{-0.9cm}
    \caption{
      Target error list plots for reconstructing $\texttt{GHZ}_{-}(8)$ circuit using real measurements from IBM Quantum system experiments.
      }
    \label{temp}
  \end{center}
\end{figure*}
\vspace{-0.9cm}
\begin{figure*}[!h]
  \begin{center}
    \includegraphics[width=1\textwidth]{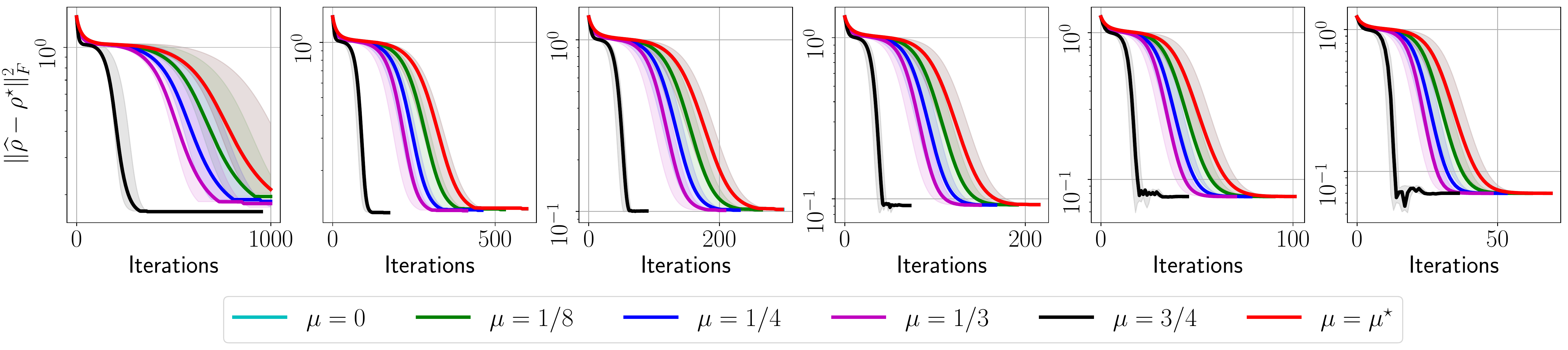}   \vspace{-0.9cm}
    \caption{
      Target error list plots for reconstructing $\texttt{GHZ}_{-}(8)$ circuit using synthetic measurements from IBM's quantum simulator.
      }
    \label{temp}
  \end{center}
\end{figure*}
\vspace{-0.9cm}
\begin{figure*}[!h]
  \begin{center}
    \includegraphics[width=1\textwidth]{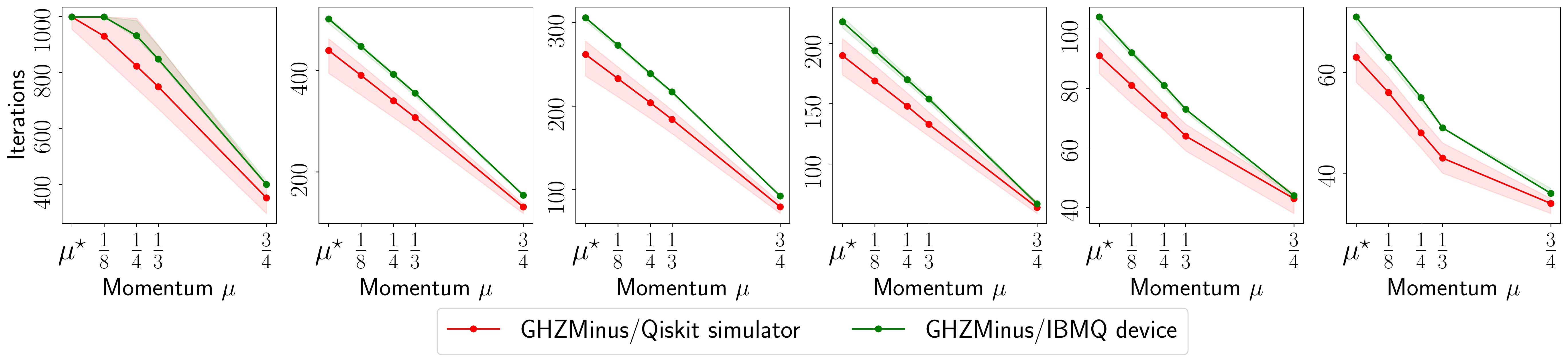}  
    \vspace{-0.9cm}
    \caption{
      Convergence iteration plots for reconstructing $\texttt{GHZ}_{-}(8)$ circuit using using real measurements from IBM Quantum system experiments and synthetic measurements from Qiskit simulation experiments.
      }
    \label{fig:convergence_iteration_8_ghzminus_2048}
  \end{center}
\end{figure*}
\vspace{-0.9cm}
\begin{figure*}[!h]
  \begin{center}
    \includegraphics[width=1\textwidth]{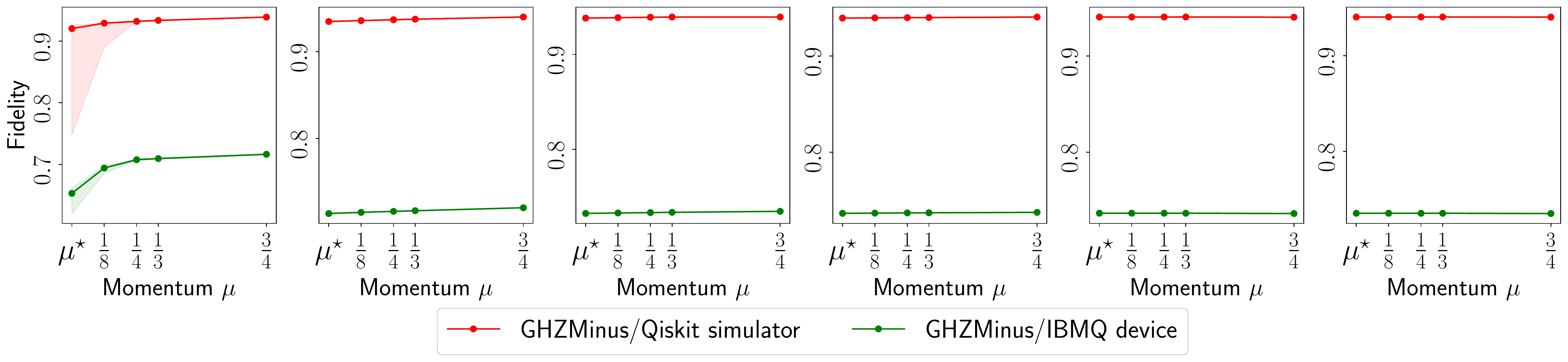}  
    \vspace{-0.9cm}
    \caption{
      Fidelity list plots for reconstructing $\texttt{GHZ}_{-}(8)$ circuit using using real measurements from IBM Quantum system experiments and synthetic measurements from Qiskit simulation experiments.
      }
    \label{fig:fidelity_list_8_ghzminus_2048}
  \end{center}
\end{figure*}

\newpage 
\subsection{IBM Quantum system experiments: $\texttt{GHZ}_{-}(8)$ circuit, 4096 \texttt{shots}} 

\begin{figure*}[!h]
  \begin{center}
    \includegraphics[width=1\textwidth]{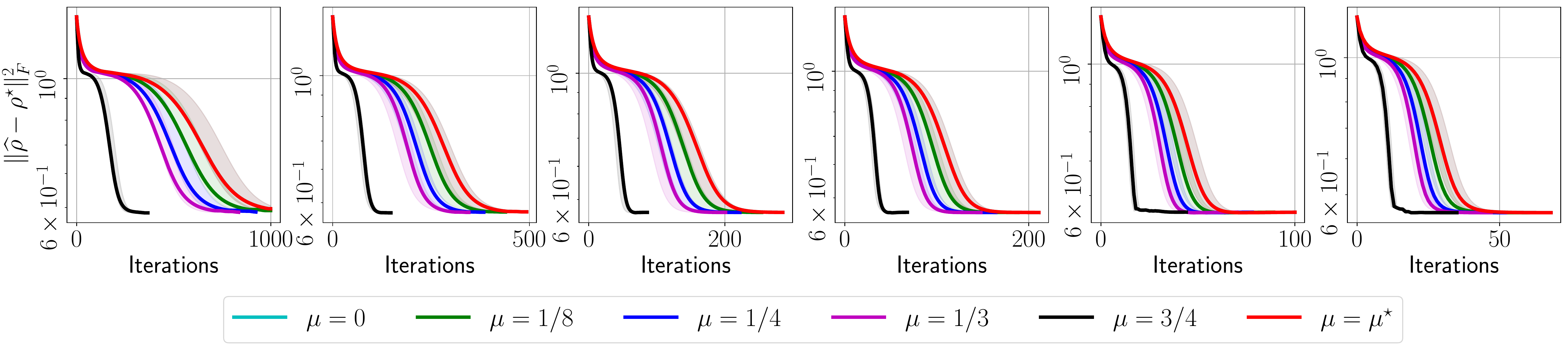}  \vspace{-0.9cm}
    \caption{
      Target error list plots for reconstructing $\texttt{GHZ}_{-}(8)$ circuit using real measurements from IBM Quantum system experiments.
      }
    \label{temp}
  \end{center}
\end{figure*}
\vspace{-0.9cm}
\begin{figure*}[!h]
  \begin{center}
    \includegraphics[width=1\textwidth]{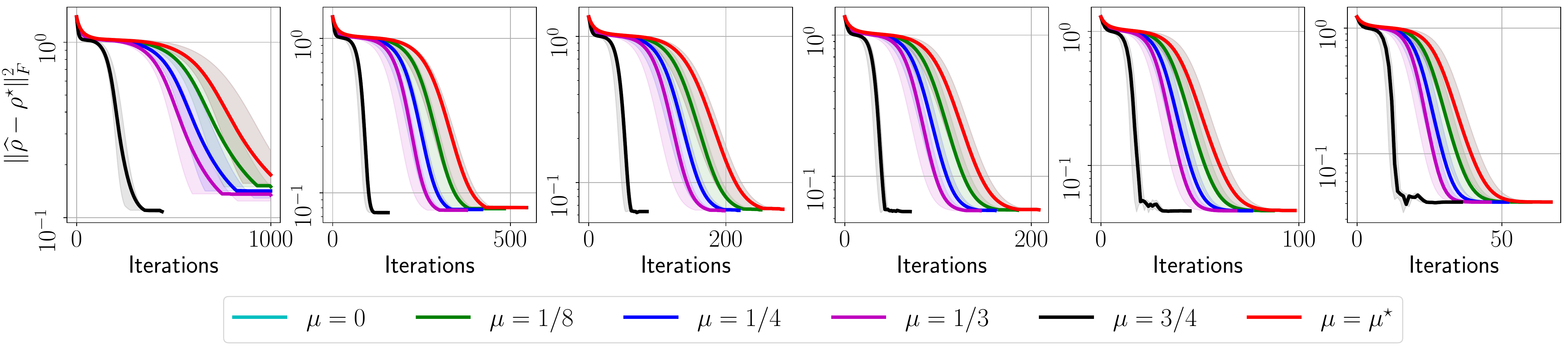}   \vspace{-0.9cm}
    \caption{
      Target error list plots for reconstructing $\texttt{GHZ}_{-}(8)$ circuit using synthetic measurements from IBM's quantum simulator.
      }
    \label{temp}
  \end{center}
\end{figure*}
\vspace{-0.9cm}
\begin{figure*}[!h]
  \begin{center}
    \includegraphics[width=1\textwidth]{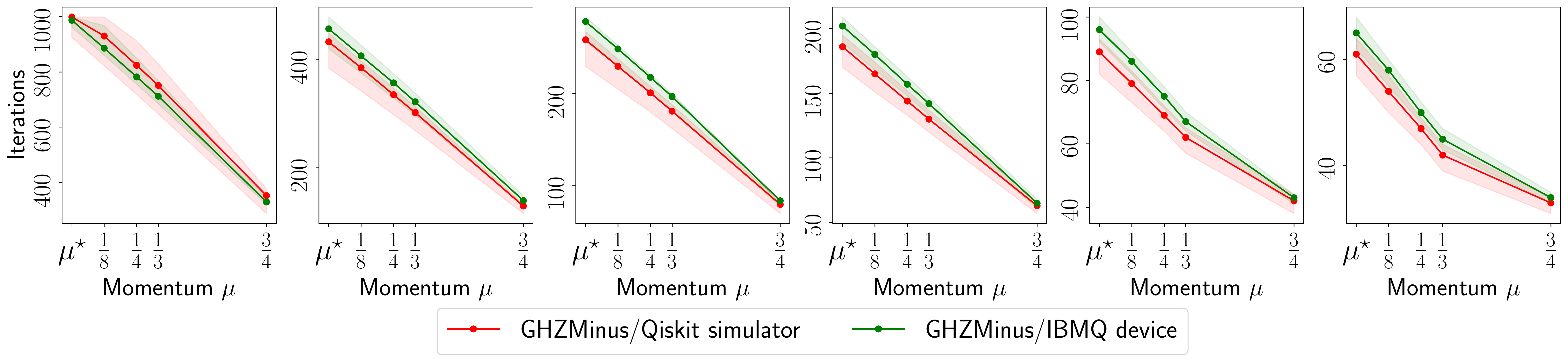}  
    \vspace{-0.9cm}
    \caption{
      Convergence iteration plots for reconstructing $\texttt{GHZ}_{-}(8)$ circuit using using real measurements from IBM Quantum system experiments and synthetic measurements from Qiskit simulation experiments.
      }
    \label{fig:convergence_iteration_8_ghzminus_4096}
  \end{center}
\end{figure*}
\vspace{-0.9cm}
\begin{figure*}[!h]
  \begin{center}
    \includegraphics[width=1\textwidth]{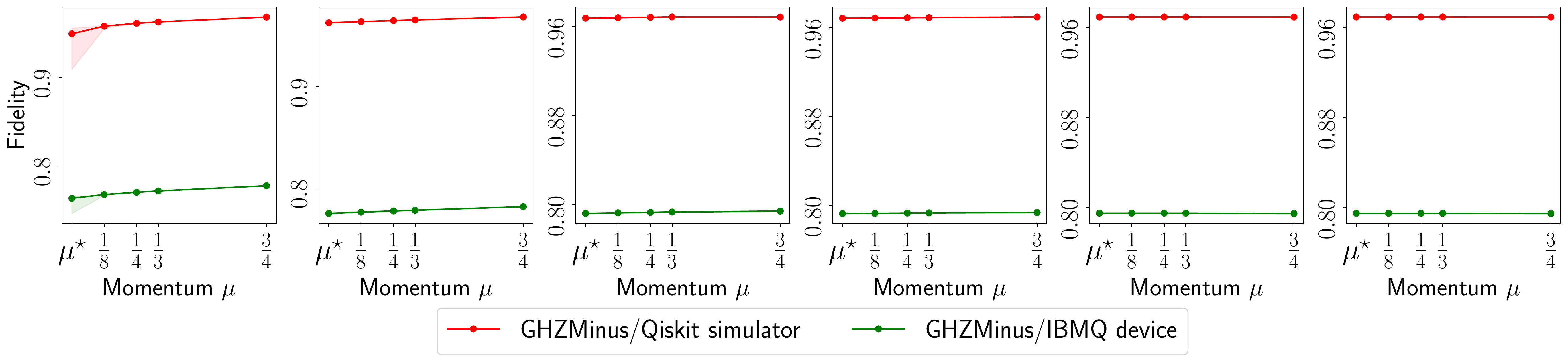}
    \vspace{-0.9cm}
    \caption{
      Fidelity list plots for reconstructing $\texttt{GHZ}_{-}(8)$ circuit using using real measurements from IBM Quantum system experiments and synthetic measurements from Qiskit simulation experiments.
      }
    \label{fig:fidelity_list_8_ghzminus_4096}
  \end{center}
\end{figure*}

\newpage 
\subsection{IBM Quantum system experiments: $\texttt{Hadamard}(6)$ circuit, 8192 \texttt{shots}} 

\begin{figure*}[!h]
  \begin{center}
    \includegraphics[width=1\textwidth]{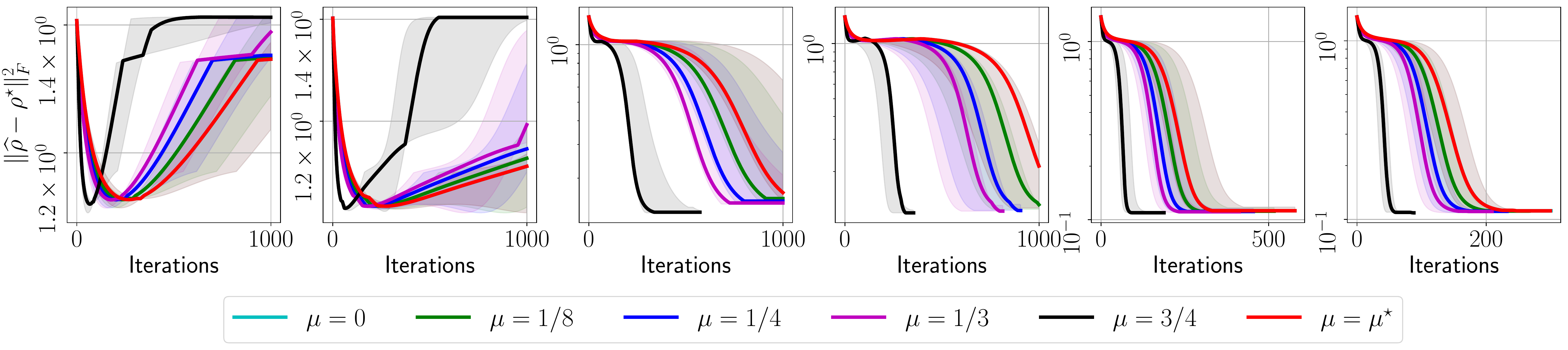}  \vspace{-0.9cm}
    \caption{
      Target error list plots for reconstructing $\texttt{Hadamard}(6)$ circuit using real measurements from IBM Quantum system experiments.
      }
    \label{temp}
  \end{center}
\end{figure*}
\vspace{-0.9cm}
\begin{figure*}[!h]
  \begin{center}
    \includegraphics[width=1\textwidth]{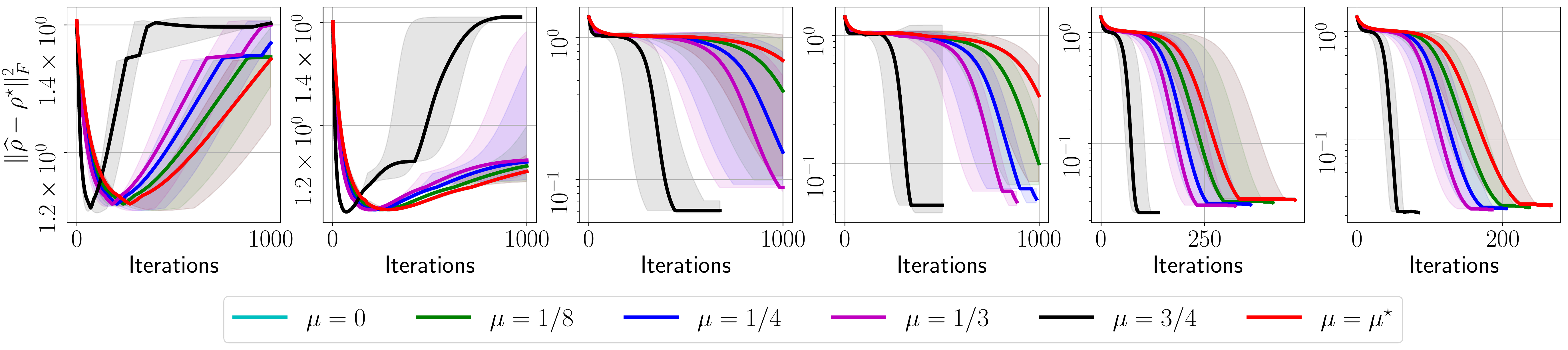}   \vspace{-0.9cm}
    \caption{
      Target error list plots for reconstructing $\texttt{Hadamard}(6)$ circuit using synthetic measurements from IBM's quantum simulator.
      }
    \label{temp}
  \end{center}
\end{figure*}
\vspace{-0.9cm}
\begin{figure*}[!h]
  \begin{center}
    \includegraphics[width=1\textwidth]{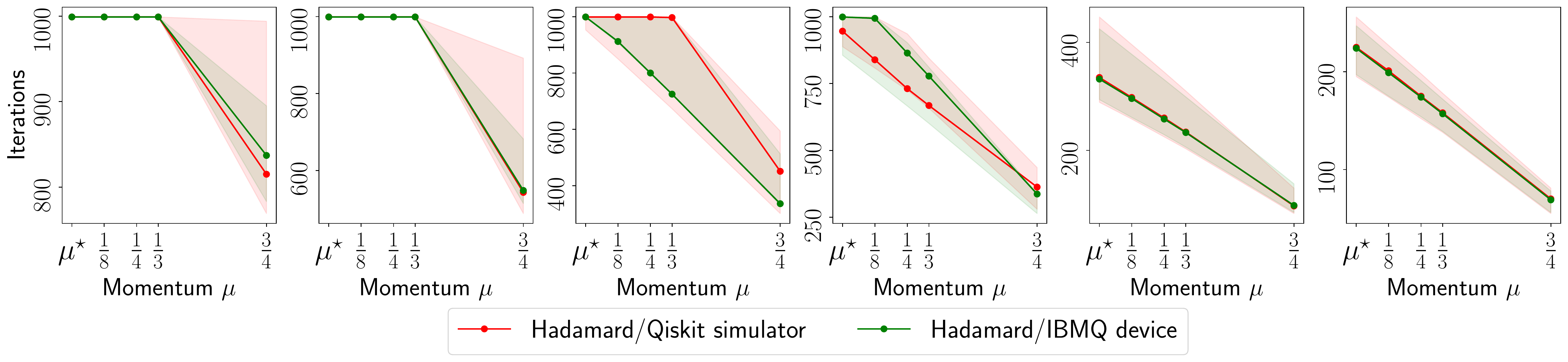}  
    \vspace{-0.9cm}
    \caption{
      Convergence iteration plots for reconstructing $\texttt{Hadamard}(6)$ circuit using using real measurements from IBM Quantum system experiments and synthetic measurements from Qiskit simulation.
      }
    \label{fig:convergence_iteration_6_hadamard_8192}
  \end{center}
\end{figure*}
\vspace{-0.9cm}
\begin{figure*}[!h]
  \begin{center}
    \includegraphics[width=1\textwidth]{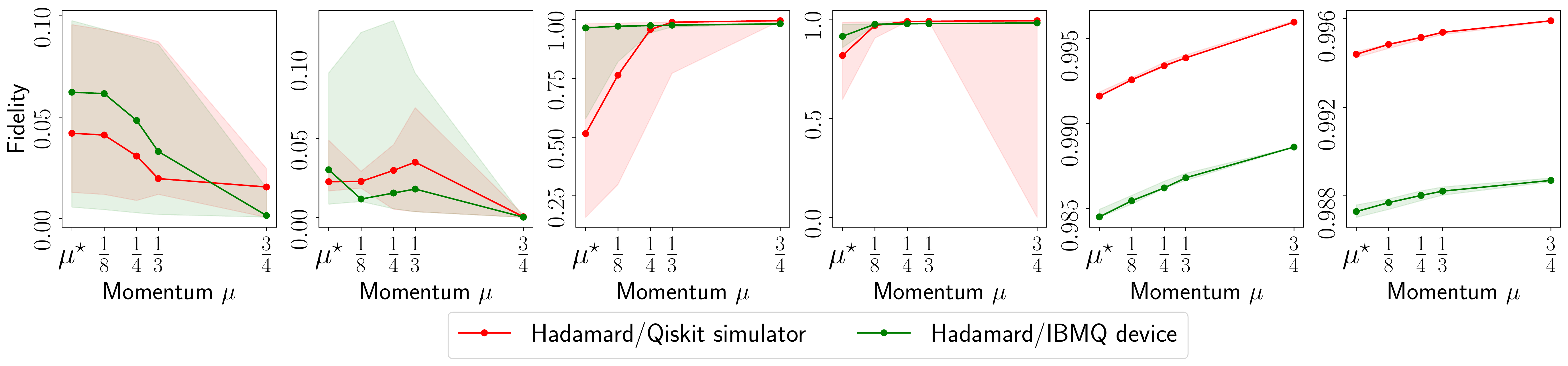}  
    \vspace{-0.9cm}
    \caption{
      Fidelity list plots for reconstructing $\texttt{Hadamard}(6)$ circuit using using real measurements from IBM Quantum system experiments and synthetic measurements from Qiskit simulation experiments.
      }
    \label{fig:fidelity_list_6_hadamard_8192}
  \end{center}
\end{figure*}

\newpage

\subsection{IBM Quantum system experiments: $\texttt{Hadamard}(8)$ circuit, 4096 \texttt{shots}} 

\begin{figure*}[!h]
  \begin{center}
    \includegraphics[width=1\textwidth]{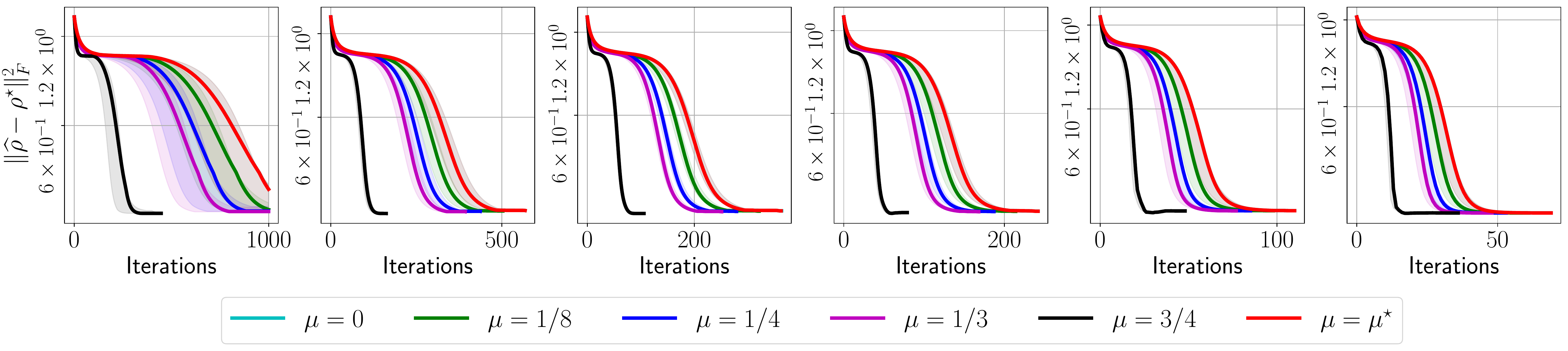}  \vspace{-0.9cm}
    \caption{
      Target error list plots for reconstructing $\texttt{Hadamard}(8)$ circuit using real measurements from IBM Quantum system experiments.
      }
    \label{temp}
  \end{center}
\end{figure*}
\vspace{-0.9cm}
\begin{figure*}[!h]
  \begin{center}
    \includegraphics[width=1\textwidth]{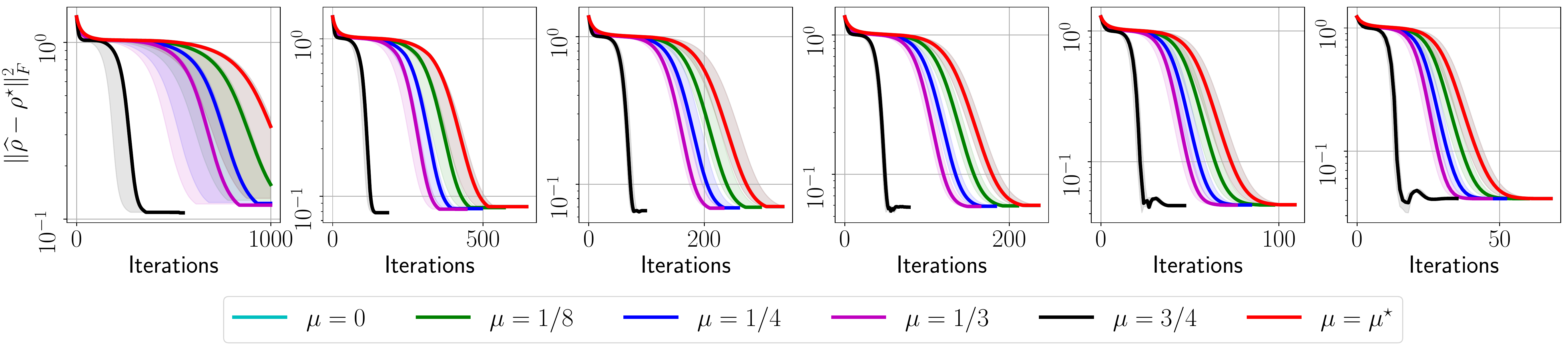}   \vspace{-0.9cm}
    \caption{
      Target error list plots for reconstructing $\texttt{Hadamard}(8)$ circuit using synthetic measurements from IBM's quantum simulator.
      }
    \label{temp}
  \end{center}
\end{figure*}
\vspace{-0.9cm}
\begin{figure*}[!h]
  \begin{center}
    \includegraphics[width=1\textwidth]{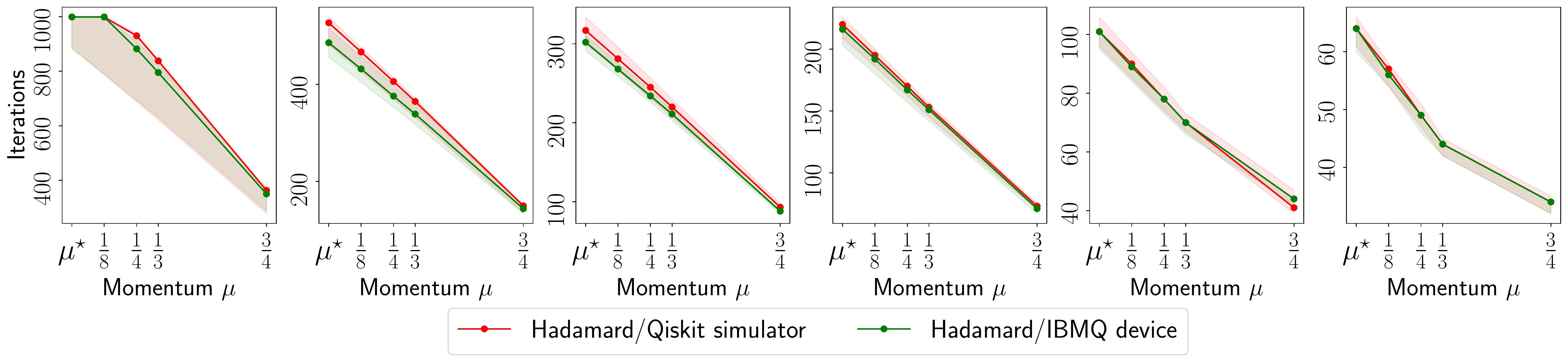}  
    \vspace{-0.9cm}
    \caption{
      Convergence iteration plots for reconstructing $\texttt{Hadamard}(8)$ circuit using using real measurements from IBM Quantum system experiments and synthetic measurements from Qiskit simulation.
      }
    \label{fig:convergence_iteration_8_hadamard_4096}
  \end{center}
\end{figure*}
\vspace{-0.9cm}
\begin{figure*}[!h]
  \begin{center}
    \includegraphics[width=1\textwidth]{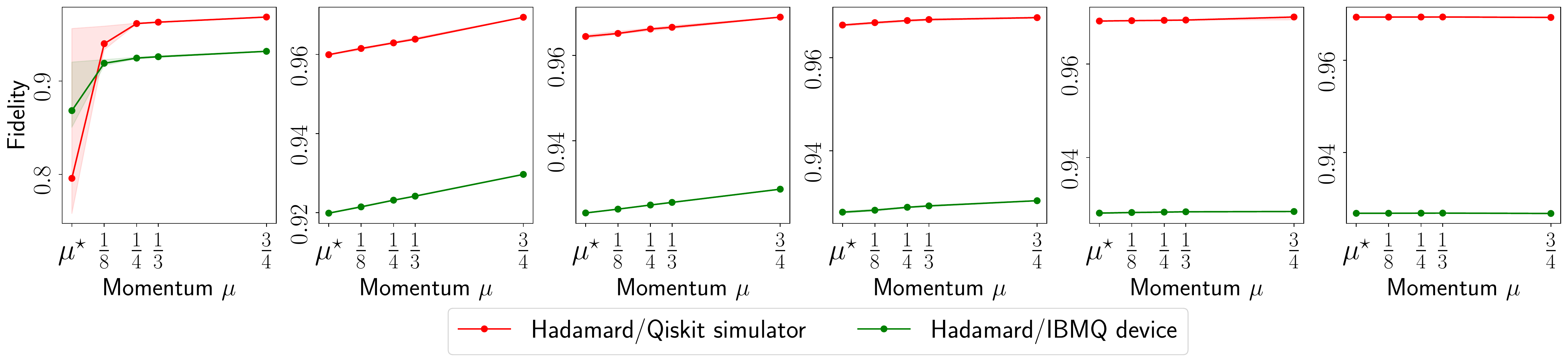}  
    \vspace{-0.9cm}
    \caption{
      Fidelity list plots for reconstructing $\texttt{Hadamard}(8)$ circuit using using real measurements from IBM Quantum system experiments and synthetic measurements from Qiskit simulation experiments.
      }
    \label{fig:fidelity_list_8_hadamard_4096}
  \end{center}
\end{figure*}

\newpage

\subsection{Synthetic experiments for $n = 12$}{\label{sec:synthetic}}
We compare \texttt{MiFGD} with 
$i)$ the \texttt{Matrix ALPS} framework \cite{kyrillidis2014matrix}, \emph{a state of the art projected gradient descent algorithm, and an optimized version of matrix iterative hard thresholding}, operating on the full matrix variable $\rho$, with adaptive step size $\eta$ (we note that this algorithm has outperformed most of the schemes that work on the original space $\rho$; see \cite{kyrillidis2014matrix});
$ii)$ the plain Procrustes Flow/\texttt{FGD} algorithm \cite{tu2016low, bhojanapalli2016dropping, kyrillidis2017provable}, where we use the step size as reported in \cite{bhojanapalli2016dropping}, since the later has reported better performance than vanilla Procrustes Flow.
\emph{We note that the Procrustes Flow/\texttt{FGD} algorithm is similar to our algorithm without acceleration.}
Further, the original Procrustes Flow/\texttt{FGD} algorithm relies on performing many iterations in the original space $\rho$ as an initialization scheme, which is often prohibitive as the problem dimensions grow. 
Both for our algorithm and the plain Procrustes Flow/\texttt{FGD} scheme, we use random initialization.

To properly compare the algorithms in the above list, we pre-select a common set of problem parameters.
We fix the dimension $d = 4096$ (equivalent to $n = 12$ qubits), and the rank of the optimal matrix $\rho^\star \in \mathbb{R}^{d \times d}$ to be $r = 10$ (equivalent to a mixed quantum state reconstruction). 
Similar behavior has been observed for other values of $r$, and are omitted.
We fix the number of observables $m$ to be $m = c \cdot d \cdot r$, where $c \in \{3, 5\}$.
In all algorithms, we fix the maximum number of iterations to 4000, and we use the same stopping criterion: \begin{scriptsize}$\|\rho_{i+1} - \rho_i\|_F / \|\rho_i\|_F \leq \texttt{tol}=10^{-3}$\end{scriptsize}. % , where $\texttt{tol} := 10^{-3}$.
For the implementation of $\texttt{MiFGD}$, we have used the momentum parameter $\mu = \tfrac{2}{3}$, as well as the theoretical $\mu$ value.

\begin{figure*}[ht]
\centering
\includegraphics[width=0.32\textwidth]{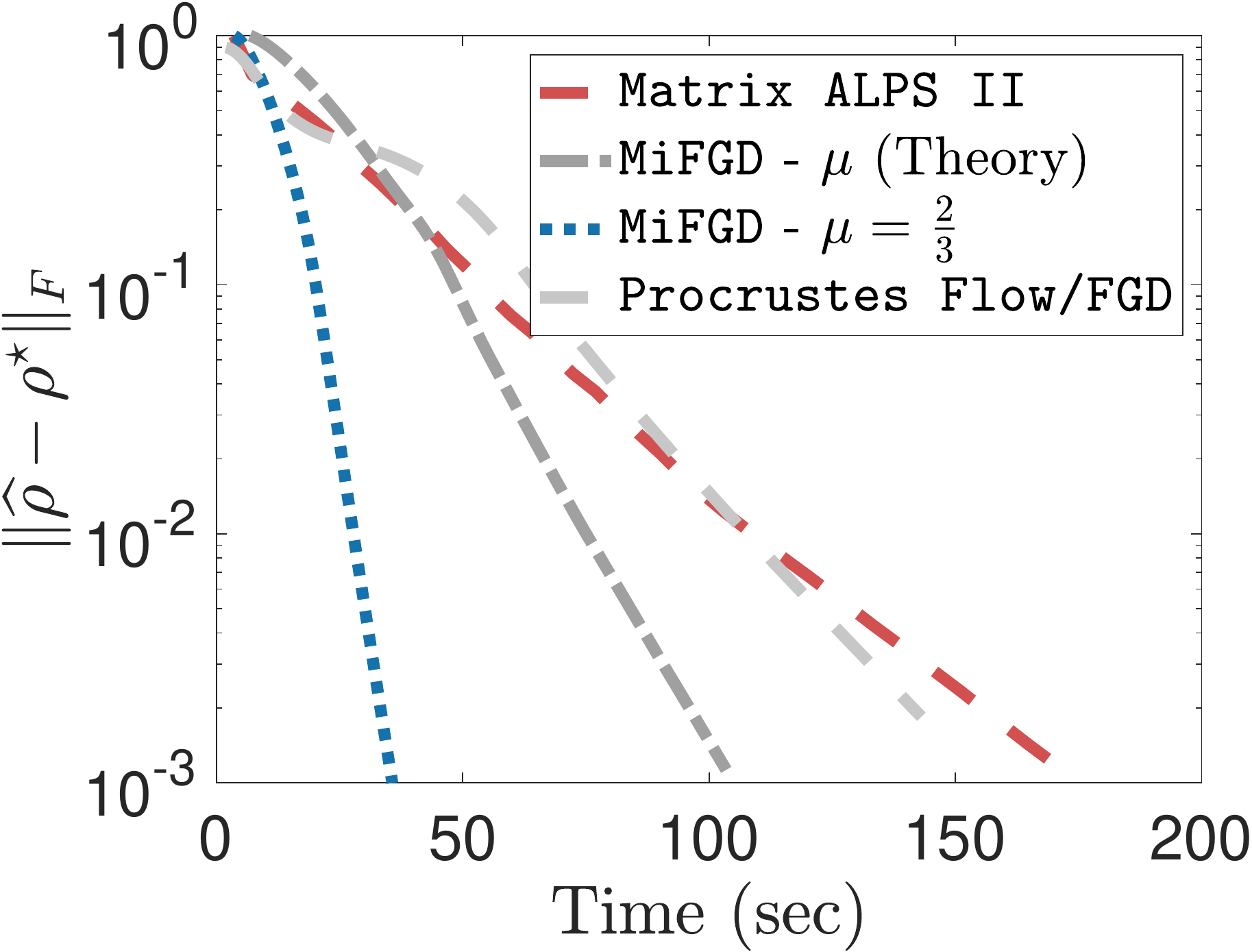} \includegraphics[width=0.31\textwidth]{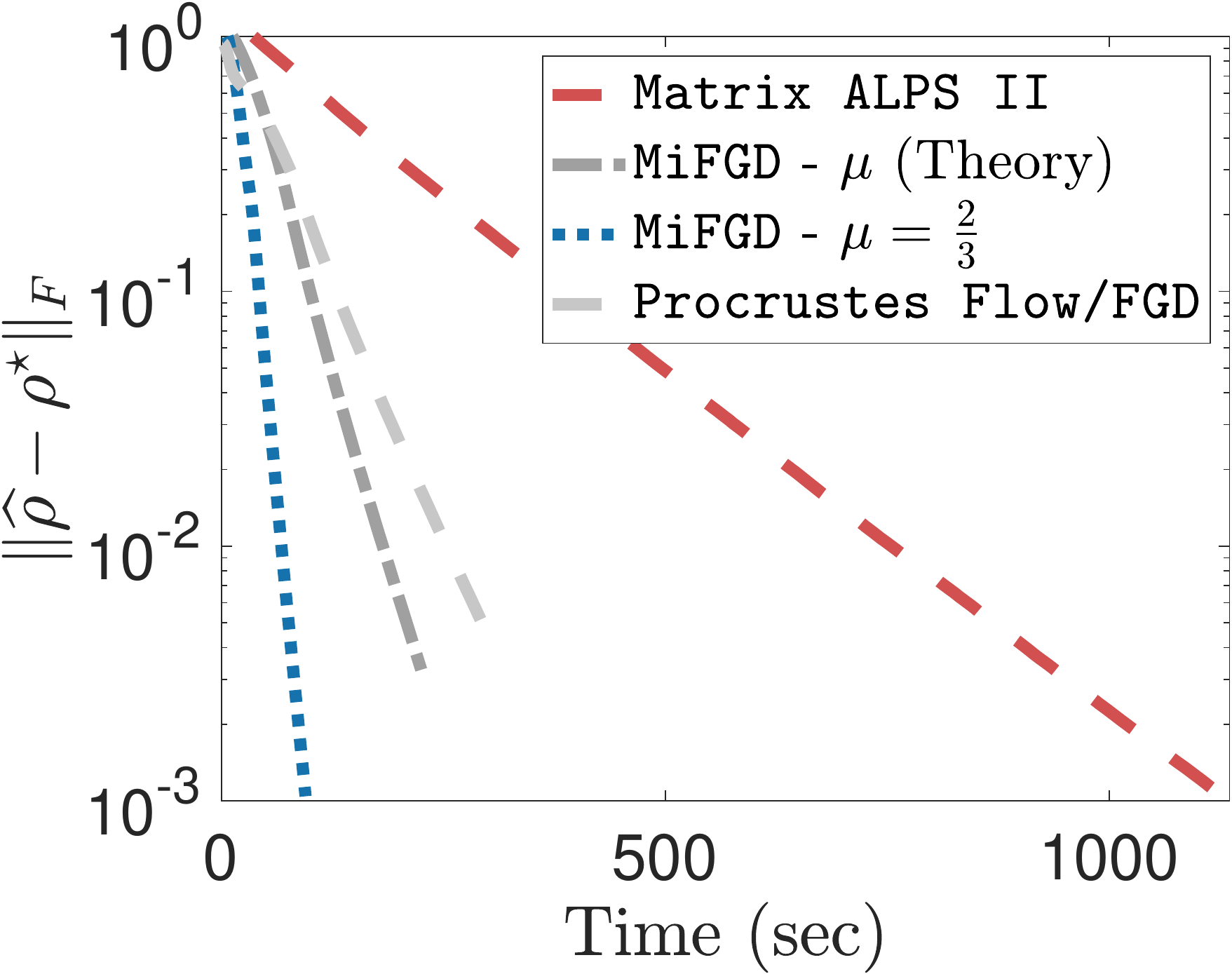}  
\includegraphics[width=0.31\textwidth]{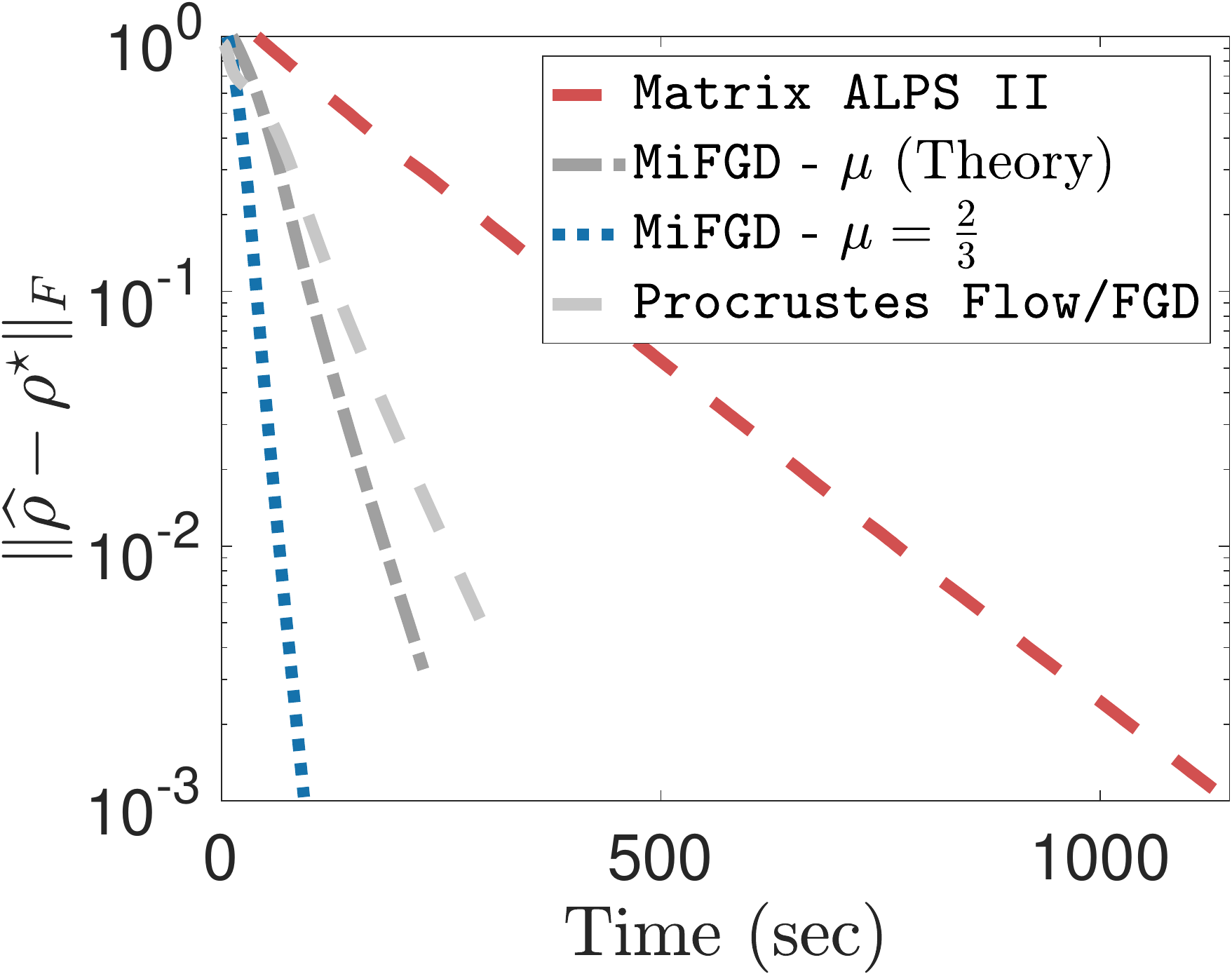} \\ \includegraphics[width=0.32\textwidth]{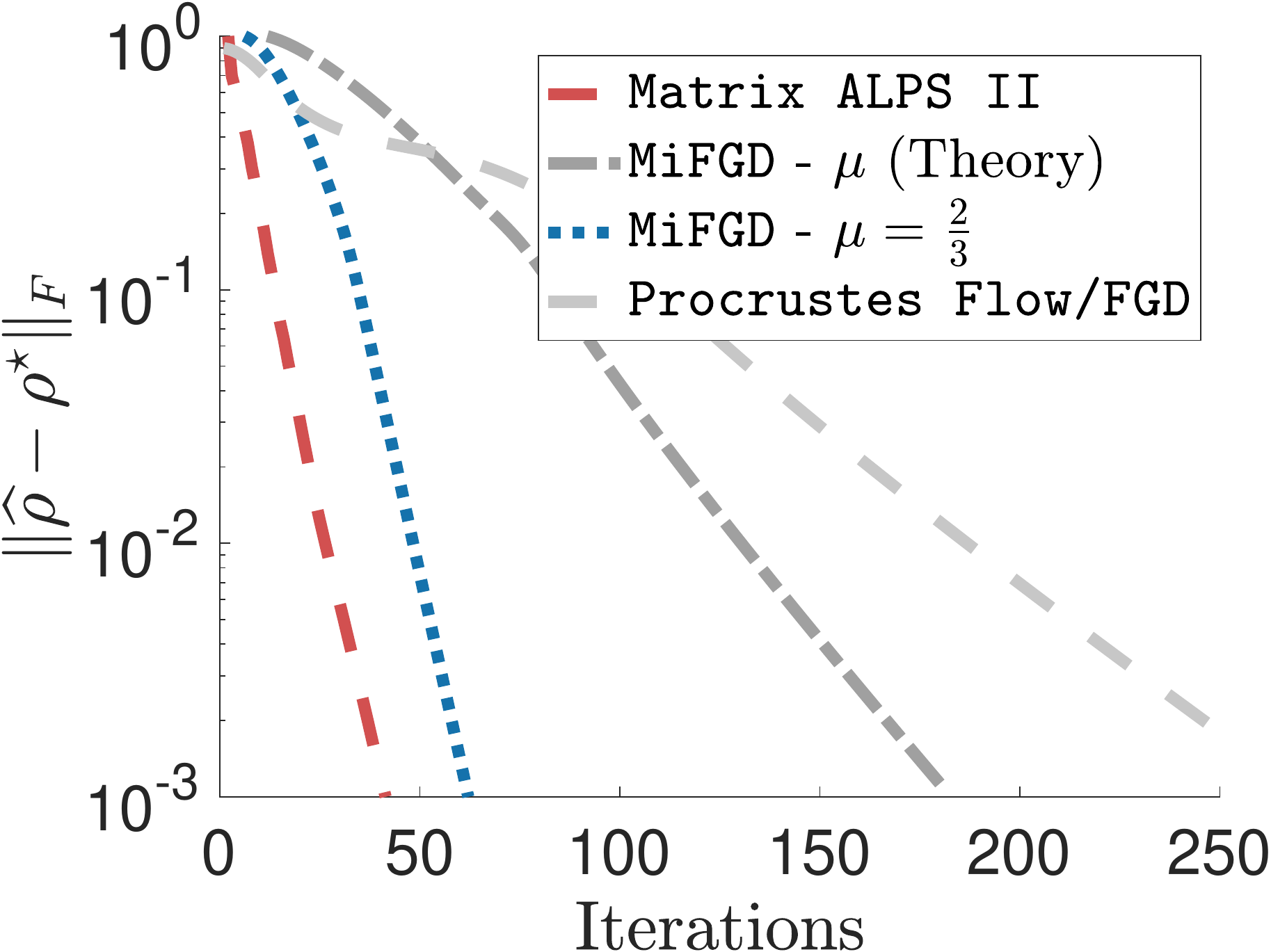} \includegraphics[width=0.32\textwidth]{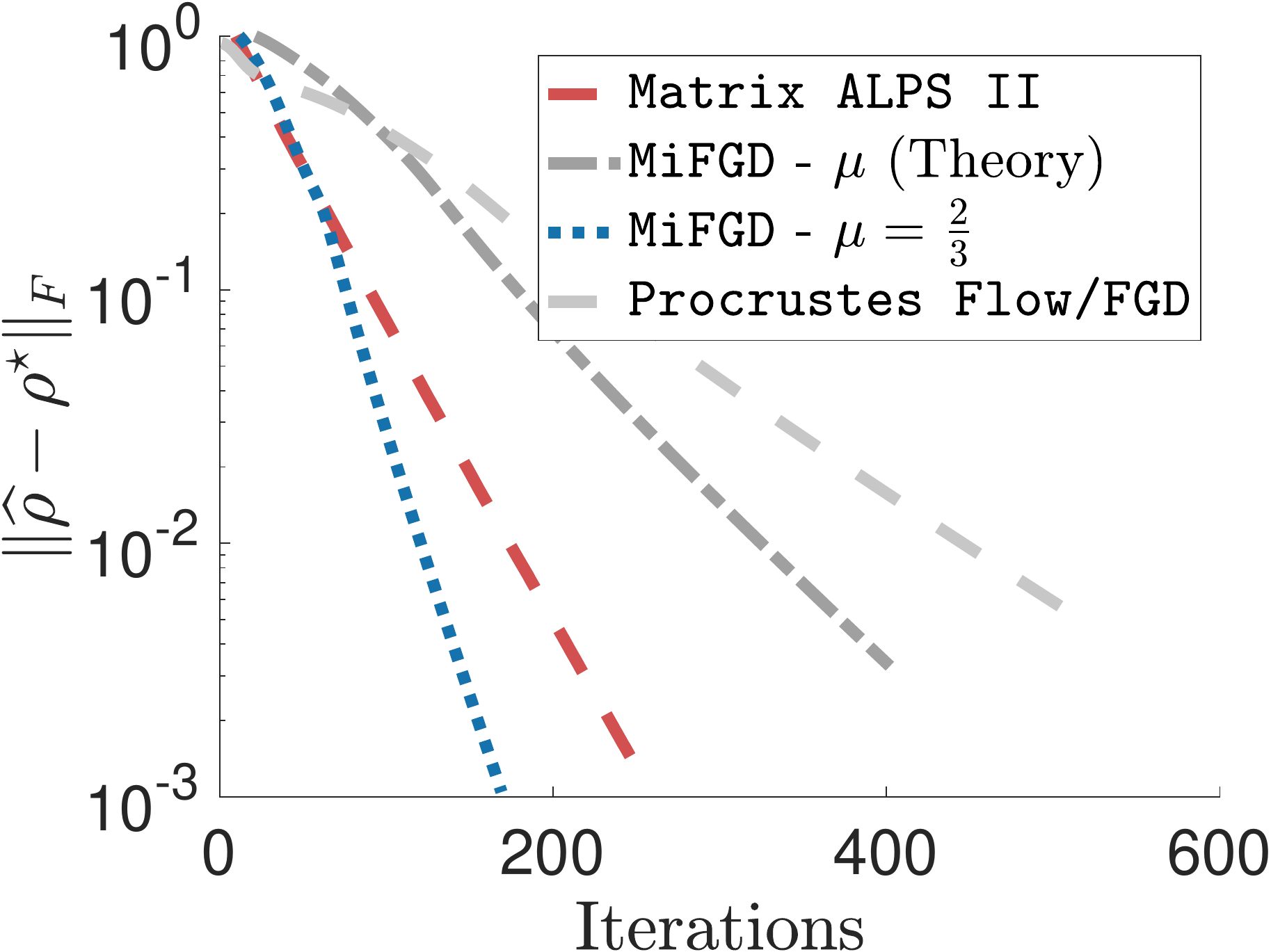} \includegraphics[width=0.32\textwidth]{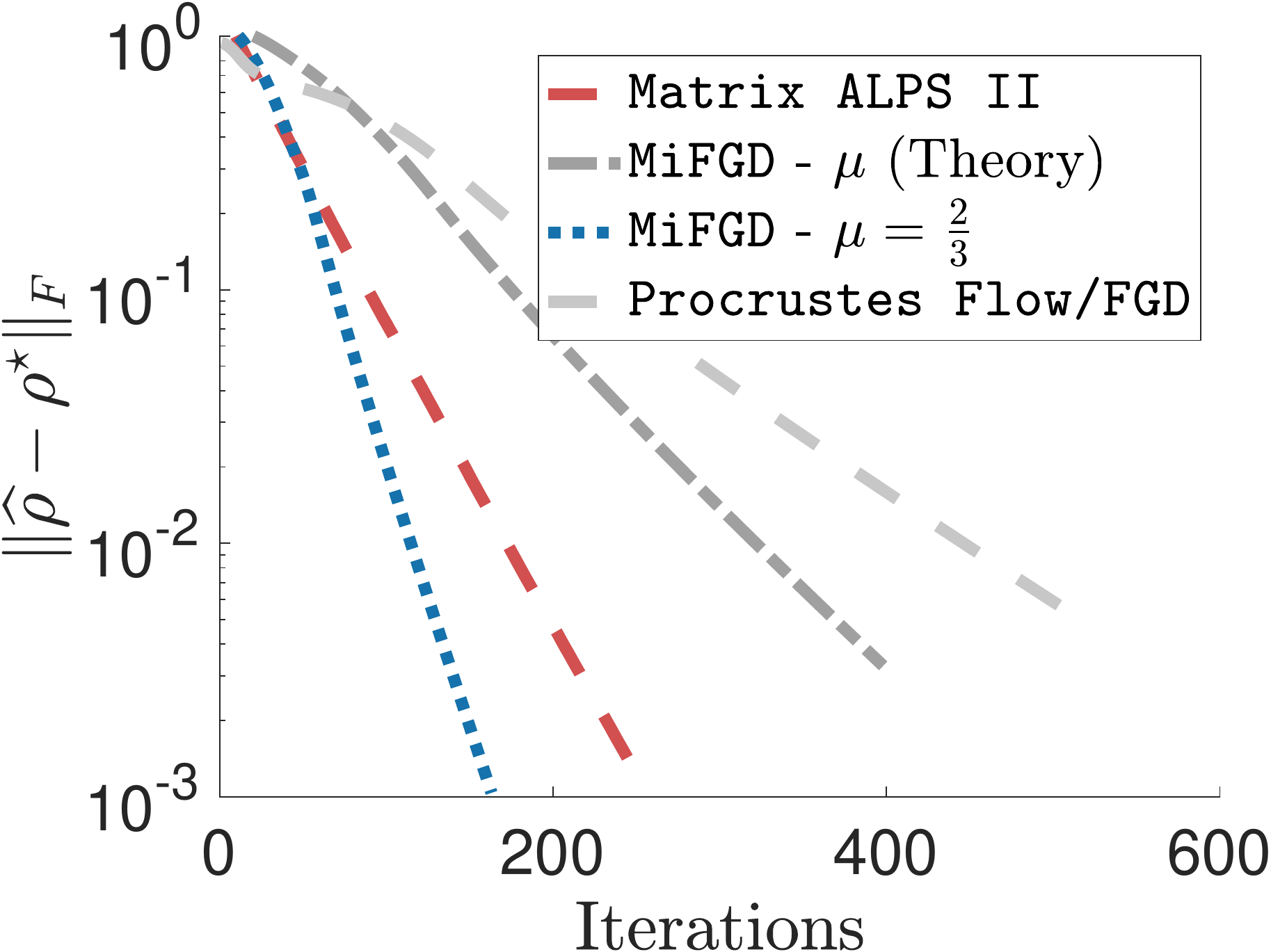} 
 \caption{Synthetic example results on low-rank matrix sensing in higher dimensions (equivalent to $n = 12$ qubits). \textbf{Top row}: Convergence behavior vs. time elapsed. \textbf{Bottom row:} Convergence behavior vs. number of iterations. \textbf{Left panel}: $c = 5$, noiseless case; \textbf{Center panel:} $c = 3$, noiseless case; \textbf{Right panel:} $c = 5$, noisy case, $\|w\|_2 = 0.01$.}
 \label{fig:00}
\end{figure*}

The procedure to generate synthetically the data is as follows:
The observations $y$ are set to
$y = \mathcal{A}\left(\rho^\star\right) + w$ for some noise vector $w$;
while the theory holds for the noiseless case, we show empirically that noisy cases are robustly handled by the same algorithm.
We use permuted and subsampled noiselets for the linear operator $\mathcal{A}$ \cite{waters2011sparcs}. 
The optimal matrix $\rho^\star$ is generated as the multiplication of a tall matrix $U^\star \in \mathbb{R}^{d \times r}$ such that $\rho^\star = U^\star U^{\star \top}$, and $\|\rho^\star \|_F = 1$, without loss of generality. 
The entries of $U^\star$ are drawn i.i.d. from a Gaussian distribution with zero mean and unit variance.
In the noisy case, $w$ has the same dimensions with $y$, its entries are drawn from a zero mean Gaussian distribution with norm $\|w\|_2 = 0.01$.
The random initialization is defined as $U_0$ drawn i.i.d. from a Gaussian distribution with zero mean and unit variance.

The results are shown in Figure \ref{fig:00}.
Some notable remarks: 
$i)$ While factorization techniques might take more iterations to converge compared to non-factorized algorithms, the per iteration time complexity is much less, such that overall, factorized gradient descent converges more quickly in terms of total execution time.
$ii)$ Our proposed algorithm, \emph{even under the restrictive assumptions on acceleration parameter $\mu$}, performs better than the non-accelerated factored gradient descent algorithms, such as Procrustes Flow.
$iii)$ Our theory is conservative: using a much larger $\mu$ we obtain a faster convergence; the proof for less strict assumptions for $\mu$ is an interesting future direction.
In all cases, our findings illustrate the effectiveness of the proposed schemes on different problem configurations.

 \subsection{
 Asymptotic complexity comparison of \texttt{lstsq}, \texttt{CVXPY}, and \texttt{MiFGD}}

{
%\color{blue}
We first note that \texttt{lstsq} can be only applied to the case we have a full tomographic set of measurements; this makes \texttt{lstsq} algorithm inapplicable in the compressed sensing scenario, where the number of measurements can be significantly reduced. Yet, we make the comparison by providing information-theoretically complete set of measurements to \texttt{lstsq} and \texttt{CVXPY}, as well as to \texttt{MiFGD}, to highlight the efficiency of our proposed method, even in the scenario that is not exactly intended in our work.
Given this, we compare in detail the asymptotic scailing of \texttt{MiFGD} with \texttt{lstsq} and \texttt{CVXPY} below:

\begin{itemize}
    \item \texttt{lstsq} is based on the computation of eigenvalues/eigenvector pairs (among other steps) of a matrix of size equal to the density matrix we want to reconstruct. Based on our notation, the density matrices are denoted as $\rho$ with dimensions $2^n \times 2^n$. Here, $n$ is the number of qubits in the quantum system. Standard libraries for eigenvalue/eigenvector calculations, like LAPACK, reduce a Hermitian matrix to tridiagonal form using the Householder method, which takes overall a $O\left((2^n)^3\right)$ computational complexity. The other steps in the \texttt{lstsq} procedure either take constant time, or $O(2^n)$ complexity. Thus, the actual run-time of an implementation depends on the eigensystem solver that is being used. %\emph{We note that \texttt{lstsq} only applies to the case we have a full tomographic set of measurements; this makes the algorithm inapplicable in the compressed scenario we consider in this work, in order to reduce the number of measurements required to perform tomography. \textcolor{blue}{AK: This sounds contradicting/strange, if lstsq is not applicable for CS, how/why do we use it in the paper?}} \textcolor{magenta}{Tasos: this discussion is about full tomography; we also have such results in the paper.}
    \item %For $n \geq 7$, Qiskit automatically chooses to utilize the \texttt{CVXPY} package, a python library with a collection of general-purpose solvers for convex optimization problems. Given that QST involves a positive-semidefinite constraint, this makes the optimization problem a semi-definite program (SDP). 
    \texttt{CVXPY} is distributed with the open source solvers; for the case of SDP instances, \texttt{CVXPY} utilizes the Splitting Conic Solver (SCS)\footnote{\url{https://github.com/cvxgrp/scs}}, a general numerical optimization package for solving large-scale convex cone problems. 
    SCS applies Douglas-Rachford splitting to a homogeneous embedding of the quadratic cone program. 
    %The algorithm involves multiple steps (including solving a large linear system of equations and performing projections onto the convex constraints). 
    Based on the PSD constraint, this again involves the computation of eigenvalues/eigenvector pairs (among other steps) of a matrix of size equal to the density matrix we want to reconstruct. This takes overall a $O\left((2^n)^3\right)$ computational complexity, \emph{not including the other steps performed within the SCS solver}. This is an iterative algorithm that requires such complexity per iteration. 
    Douglas-Rachford splitting methods enjoy $O(\tfrac{1}{\varepsilon})$ convergence rate in general \cite{he2015convergence, ocpb:16, odonoghue:21}. This leads to a rough $O((2^n)^3 \cdot \tfrac{1}{\varepsilon})$ overall iteration complexity.\footnote{This is an optimistic complexity bound since we have skipped several details within the Douglas-Rachford implementation of \texttt{CVXPY}.}
    \item For \texttt{MiFGD}, and for sufficiently small momentum value, we require $O(\sqrt{\kappa} \cdot \log(\tfrac{1}{\varepsilon}))$ iterations to get close to the optimal value. Per iteration, \texttt{MiFGD} does not involve any expensive eigensystem solvers, but relies only on matrix-matrix and matrix-vector multiplications. In particular, the main computational complexity per iteration origins from the iteration: 
    \begin{align}
U_{i+1} &= Z_{i} - \eta \mathcal{A}^\dagger \left(\mathcal{A}(Z_i Z_i^\dagger) - y\right) \cdot Z_i, \label{eq:MiFGD1}\\
Z_{i+1} &= U_{i+1} + \mu \left(U_{i+1} - U_i\right). \label{eq:MiFGD2}
\end{align}
Here, $U_i, Z_i \in \mathbb{R}^{2^n \times r}$ for all $i$. Observe that $\mathcal{A}(Z_i Z_i^\dagger) \in \mathbb{R}^m$ where each element is computed independently. For an index $j \in [m]$, $(\mathcal{A}(Z_i Z_i^\dagger))_j = \texttt{Tr}(A_j Z_i Z_i^\dagger)$ requires $O((2^n)^2 \cdot r)$ complexity, and thus computing $\mathcal{A}(Z_i Z_i^\dagger) - y$ requires $O((2^n)^2 \cdot r)$ complexity, overall. By definition the adjoing operation $\mathcal{A}^\dagger: \mathbb{R}^m \rightarrow \mathbb{C}^{2^n \times 2^n}$ satisfies: $\mathcal{A}^\dagger(x) = \sum_{i = 1}^m x_i A_i$; thus, the operation $\mathcal{A}^\dagger \left(\mathcal{A}(Z_i Z_i^\dagger) - y\right)$ is still dominated by $O((2^n)^2 \cdot r)$ complexity. Finally, we perform one more matrix-matrix multiplication with $Z_i$, which results into an additional $O((2^n)^2 \cdot r)$ complexity. The rest of the operations involve adding $2^n \times r$ matrices, which does not dominate the overall complexity. Combining the iteration complexity with the per-iteration computational complexity, \texttt{MiFGD} has a $O((2^n)^2 \cdot r \cdot \sqrt{\kappa} \cdot \log(\tfrac{1}{\varepsilon}))$ complexity.
\end{itemize}

Combining the above, we summarize the following complexities:
\begin{align*}
    \underbrace{O((2^n)^3)}_{\texttt{lstsq}} \quad \quad \text{vs} \quad \quad \underbrace{O((2^n)^3 \cdot \tfrac{1}{\varepsilon})}_{\texttt{CVXPY}} \quad \quad \text{vs} \quad \quad \underbrace{O((2^n)^2 \cdot r \cdot \sqrt{\kappa} \cdot \log(\tfrac{1}{\varepsilon}))}_{\texttt{MiFGD}}
\end{align*}

Observe that $i)$ \texttt{MiFGD} has the best dependence on the number of qubits and the ambient dimension of the problem, $2^n$; $ii)$ \texttt{MiFGD} applies to cases that \texttt{lstsq} is inapplicable; $iii)$ \texttt{MiFGD} has a better iteration complexity than other iterative algorithms, while has a better polynomial dependency on $2^n$.
}

\section{Detailed proof of Theorem \ref{thm:00}} 
We first denote $U_{+} \equiv U_{i+1}$, $U \equiv U_i$, $U_- \equiv U_{i - 1}$ and $Z \equiv Z_i$.
Let us start with the following equality. 
For $R_Z \in \mathcal{O}$ as the minimizer of $\min_{R \in \mathcal{O}} \|Z - U^\star R\|_F$,
we have:
\begin{align}
\|U_{+} - U^\star R_Z\|_F^2 &= \|U_{+} - Z + Z - U^\star R_Z\|_F^2 \\
					        &= \|U_{+} - Z\|_F^2 + \|Z - U^\star R_Z\|_F^2  - 2 \langle U_{+} - Z, U^\star R_Z - Z \rangle
\end{align}
The proof focuses on how to bound the last part on the right-hand side.
By definition of $U_{+}$, we get:
\begin{align}
\langle U_{+} - Z, U^\star R_Z - Z \rangle &= \left \langle Z - \eta \mathcal{A}^\dagger\left(\mathcal{A}(ZZ^\dagger) - y \right) Z - Z, U^\star R_Z - Z \right \rangle \\
						          &= \eta \left \langle \mathcal{A}^\dagger\left(\mathcal{A}(ZZ^\dagger) - y \right) Z, Z - U^\star R_Z \right \rangle
\end{align}
Observe the following:
\begin{small}
\begin{align}
\left \langle \mathcal{A}^\dagger\left(\mathcal{A}(ZZ^\dagger) - y \right) Z, Z - U^\star R_Z \right \rangle &= \left \langle \mathcal{A}^\dagger\left(\mathcal{A}(ZZ^\dagger) - y \right), ZZ^\dagger  - U^\star R_Z Z^\dagger \right \rangle \\
																		&= \left \langle \mathcal{A}^\dagger\left(\mathcal{A}(ZZ^\dagger) - y \right), ZZ^\dagger - \tfrac{1}{2} U^\star U^{\star \dagger} + \tfrac{1}{2} U^\star U^{\star \dagger} - U^\star R_Z Z^\dagger \right \rangle \\
																		&= \tfrac{1}{2} \left \langle \mathcal{A}^\dagger\left(\mathcal{A}(ZZ^\dagger) - y \right), ZZ^\dagger - U^\star U^{\star \dagger} \right \rangle \\
																		&\quad \quad + \left \langle \mathcal{A}^\dagger\left(\mathcal{A}(ZZ^\dagger) - y \right), \tfrac{1}{2} (ZZ^\dagger + U^\star U^{\star \dagger}) - U^\star R_Z Z^\dagger \right \rangle \\
																																				&= \tfrac{1}{2} \left \langle \mathcal{A}^\dagger\left(\mathcal{A}(ZZ^\dagger) - y \right), ZZ^\dagger - U^\star U^{\star \dagger} \right \rangle \\
																		&\quad \quad + \tfrac{1}{2} \left \langle \mathcal{A}^\dagger\left(\mathcal{A}(ZZ^\dagger) - y \right), (Z - U^\star R_Z)(Z - U^\star R_Z)^\dagger \right \rangle
\end{align}
\end{small}
By Lemmata \ref{lem:000} and \ref{lem:001}, we have:
\begin{small}
\begin{align}
\|U_+ - U^\star R_Z\|_F^2 &= \|U_{+} - Z\|_F^2 + \|Z - U^\star R_Z\|_F^2  - 2 \langle U_{+} - Z, U^\star R_Z - Z \rangle \\
					       &= \eta^2 \|\mathcal{A}^\dagger\left(\mathcal{A}(ZZ^\dagger) - y \right) Z \|_F^2 + \|Z - U^\star R_Z\|_F^2  \\ 
					       &\quad \quad -\eta \left \langle \mathcal{A}^\dagger\left(\mathcal{A}(ZZ^\dagger) - y \right), ZZ^\dagger - U^\star U^{\star \dagger} \right \rangle \\
					       &\quad \quad \quad \quad - \eta \left \langle \mathcal{A}^\dagger\left(\mathcal{A}(ZZ^\dagger) - y \right), (Z - U^\star R_Z)(Z - U^\star R_Z)^\dagger \right \rangle \\
					       &\leq \eta^2 \|\mathcal{A}^\dagger\left(\mathcal{A}(ZZ^\dagger) - y \right) Z \|_F^2 + \|Z - U^\star R_Z\|_F^2  \\ 
					       &\quad \quad - 1.0656 \eta^2 \left \|\mathcal{A}^\dagger( \mathcal{A}(ZZ^\dagger) - y)  Z \right \|_F^2 - \eta \tfrac{1 - \delta_{2r}}{2} \|U^\star U^{\star \dagger} - ZZ^\dagger\|_F^2 \\
					       &\quad \quad \quad \quad + \eta \Bigg( \theta \sigma_r(\rho^\star) \cdot \|Z - U^\star R_Z\|_F^2 \\ 
					       &\quad \quad \quad \quad \quad \quad + \tfrac{1}{200} \beta^2 \cdot \widehat{\eta} \cdot \tfrac{\left(\tfrac{3}{2} + 2|\mu|\right)^2}{\left(1  - \left(\tfrac{3}{2} + 2|\mu|\right)\tfrac{1}{10^3}\right)^2} \cdot \|\mathcal{A}^\dagger(\mathcal{A}(ZZ^\dagger) - y) \cdot Z\|_F^2\Bigg)
\end{align}
\end{small}

Next, we use the following lemma:
\begin{lemma}\cite[Lemma 5.4]{tu2016low}{\label{lem:tu}}
For any $W, V \in \mathbb{C}^{d \times r}$, the following holds:
\begin{small}
\begin{align}
\|WW^\dagger - VV^\dagger\|_F^2 \geq 2 (\sqrt{2} - 1) \cdot \sigma_r(VV^\dagger) \cdot \min_{R \in \mathcal{O}} \|W - VR\|_F^2.
\end{align}
\end{small}
\end{lemma}

From Lemma \ref{lem:tu}, the quantity  $\|U^\star U^{\star \dagger} - ZZ^\dagger\|_F^2$ satisfies:
\begin{small}
\begin{align}
\|U^\star U^{\star \dagger} - ZZ^\dagger\|_F^2 \geq 2 (\sqrt{2} - 1) \cdot \sigma_r(\rho^\star) \cdot \min_{R \in \mathcal{O}} \|Z - U^\star R\|_F^2 = 2 (\sqrt{2} - 1) \cdot \sigma_r(\rho^\star) \cdot \|Z - U^\star R_Z\|_F^2 ,
\end{align}
\end{small}
which, in our main recursion, results in:
\begin{small}
\begin{align}
\|U_+ - U^\star R_Z\|_F^2 &\leq \eta^2 \|\mathcal{A}^\dagger\left(\mathcal{A}(ZZ^\dagger) - y \right) Z \|_F^2 + \|Z - U^\star R_Z\|_F^2  \\ 
					       &\quad \quad - 1.0656 \eta^2 \left \|\mathcal{A}^\dagger( \mathcal{A}(ZZ^\dagger) - y)  Z \right \|_F^2 - \eta (\sqrt{2}-1)(1 - \delta_{2r}) \sigma_r(\rho^\star) \|Z - U^\star R_Z\|_F^2 \\
					       &\quad \quad \quad \quad + \eta \Bigg( \theta \sigma_r(\rho^\star) \cdot \|Z - U^\star R_Z\|_F^2 \\ 
					       &\quad \quad \quad \quad \quad \quad + \tfrac{1}{200} \beta^2 \cdot \widehat{\eta} \cdot \tfrac{\left(\tfrac{3}{2} + 2|\mu|\right)^2}{\left(1  - \left(\tfrac{3}{2} + 2|\mu|\right)\tfrac{1}{10^3}\right)^2} \cdot \|\mathcal{A}^\dagger(\mathcal{A}(ZZ^\dagger) - y) \cdot Z\|_F^2\Bigg) \\
					       &\stackrel{(i)}{\leq} \eta^2 \|\mathcal{A}^\dagger\left(\mathcal{A}(ZZ^\dagger) - y \right) Z \|_F^2 + \|Z - U^\star R_Z\|_F^2  \\ 
					       &\quad \quad - 1.0656 \eta^2 \left \|\mathcal{A}^\dagger( \mathcal{A}(ZZ^\dagger) - y)  Z \right \|_F^2 - \eta (\sqrt{2}-1)(1 - \delta_{2r}) \sigma_r(\rho^\star) \|Z - U^\star R_Z\|_F^2 \\
					       &\quad \quad \quad \quad + \eta \Bigg( \theta \sigma_r(\rho^\star) \cdot \|Z - U^\star R_Z\|_F^2 \\ 
					       &\quad \quad \quad \quad \quad \quad + \tfrac{1}{200} \beta^2 \cdot \tfrac{10}{9} \eta \cdot \tfrac{\left(\tfrac{3}{2} + 2|\mu|\right)^2}{\left(1  - \left(\tfrac{3}{2} + 2|\mu|\right)\tfrac{1}{10^3}\right)^2} \cdot \|\mathcal{A}^\dagger(\mathcal{A}(ZZ^\dagger) - y) \cdot Z\|_F^2\Bigg) \\					       
					       &\stackrel{(ii)}{=} \left(1 + \tfrac{1}{200} \beta^2 \cdot \tfrac{10}{9} \cdot \tfrac{\left(\tfrac{3}{2}+ 2|\mu|\right)^2}{\left(1  - \left(\tfrac{3}{2} + 2|\mu|\right)\tfrac{1}{10^3}\right)^2} - 1.0656\right) \eta^2 \|\mathcal{A}^\dagger(\mathcal{A}(ZZ^\dagger) - y) \cdot Z\|_F^2  \\ 
					       &\quad \quad \quad \quad +\left(1 + \eta \theta \sigma_r(\rho^\star) - \eta (\sqrt{2}-1)(1 - \delta_{2r}) \sigma_r(\rho^\star) \right) \|Z - U^\star R_Z\|_F^2 
\end{align}
\end{small}
where $(i)$ is due to Lemma \ref{lem:equiveta}, 
and $(ii)$ is due to the definition of $U_{+}$.

Under the facts that \begin{small}$\mu = \frac{\sigma_r(\rho^\star)^{1/2}}{10^3 \sqrt{\kappa\tau(\rho^\star)}} \cdot \frac{\varepsilon}{4 \cdot \sigma_1(\rho^\star)^{1/2} \cdot r}$\end{small}, for $\varepsilon \in (0, 1)$ user-defined, and $\delta_{2r} \leq \tfrac{1}{10}$, the main constant quantities in our proof so far simplify into:
\begin{align}
\beta = \frac{1 + \left(\tfrac{3}{2} + 2|\mu| \right) \cdot \tfrac{1}{10^3}}{1 - \left(\tfrac{3}{2} + 2|\mu| \right) \cdot \tfrac{1}{10^3}} = 1.003, \quad \text{and} \quad \beta^2 = 1.006,
\end{align}
by Corollary \ref{cor:1}. 
Thus:
\begin{align}
1 + \tfrac{1}{200} \beta^2 \cdot \tfrac{10}{9} \cdot \tfrac{\left(\tfrac{3}{2}+ 2|\mu|\right)^2}{\left(1  - \left(\tfrac{3}{2} + 2|\mu|\right)\tfrac{1}{10^3}\right)^2} - 1.0656 \leq -0.0516,
\end{align}
and our recursion becomes:
\begin{small}
\begin{align}
\|U_+ - U^\star R_Z\|_F^2 &\leq -0.0516 \cdot \eta^2 \cdot \|\mathcal{A}^\dagger(\mathcal{A}(ZZ^\dagger) - y) \cdot Z\|_F^2  \\ 
					       &\quad \quad \quad \quad +\left(1 + \eta \theta \sigma_r(\rho^\star) - \eta (\sqrt{2}-1)(1 - \delta_{2r}) \sigma_r(\rho^\star) \right) \|Z - U^\star R_Z\|_F^2 
\end{align}
\end{small}
Finally, 
\begin{small}
\begin{align}
\theta &= \tfrac{(1 - \delta_{2r}) \left(1 + \left(\tfrac{3}{2} + 2|\mu| \right)\tfrac{1}{10^3}\right)^2}{10^3} +  (1 + \delta_{2r})  \left(2+\left(\tfrac{3}{2} + 2|\mu|\right) \cdot \tfrac{1}{10^3}\right) \left(\tfrac{3}{2} + 2|\mu|\right) \cdot \tfrac{1}{10^3} \\ 
	  &\stackrel{(i)}{=} (1 - \delta_{2r}) \cdot \left( \tfrac{\left(1 + (\tfrac{3}{2} + 2|\mu|)\tfrac{1}{10^3}\right)^2}{10^3} + \kappa \left(2+\left(\tfrac{3}{2} + 2|\mu|\right) \cdot \tfrac{1}{10^3}\right) \left(\tfrac{3}{2} + 2|\mu|\right) \cdot \tfrac{1}{10^3} \right)\\
	  &\leq  0.0047 \cdot (1 - \delta_{2r}).
\end{align}
\end{small}
where $(i)$ is by the definition of $\kappa := \tfrac{1 + \delta_{2r}}{1 - \delta_{2r}} \leq 1.223$ for $\delta_{2r} \leq \tfrac{1}{10}$, by assumption.
Combining the above in our main inequality, we obtain:
\begin{small}
\begin{align}
\|U_+ - U^\star R_Z\|_F^2 &\leq -0.0516 \cdot \eta^2 \cdot \|\mathcal{A}^\dagger(\mathcal{A}(ZZ^\dagger) - y) \cdot Z\|_F^2  \nonumber \\ 
					       &\quad \quad \quad \quad +\left(1 + \eta \sigma_r(\rho^\star) (1 - \delta_{2r}) \cdot (0.0047 - \sqrt{2} + 1)\right) \|Z - U^\star R_Z\|_F^2 \nonumber \\
					       &\leq \left(1 - \tfrac{4\eta \sigma_r(\rho^\star)(1 - \delta_{2r})}{10} \right)\|Z - U^\star R_Z\|_F^2 \label{eq:0000}
\end{align}
\end{small}

% By Lemma \ref{lem:equiveta}, we know that $\eta \geq \tfrac{100}{102} \eta^\star$.
% Also, $\eta^\star = \tfrac{1}{4 (1 + \delta_{2r}) \|\rho^\star\|_2}$, since $\|\mathcal{A}^\dagger(\mathcal{A}(\rho^\star) - y)\|_2 = 0$, in the noiseless setting.
% Returning to \eqref{eq:0000}, we have:
% \begin{small}
% \begin{align}
% \|U_+ - U^\star R_Z\|_F^2 &\leq \left(1 - 0.393 \cdot \tfrac{(1 - \delta_{2r})\sigma_r(\rho^\star)}{(1 + \delta_{2r})\sigma_1(\rho^\star)} \right)\|Z - U^\star R_Z\|_F^2 \\
% 					       &= \left(1 - \tfrac{0.393}{\kappa \tau(\rho^\star)} \right)\|Z - U^\star R_Z\|_F^2
% \end{align}
% \end{small}

{
%\color{blue}
Taking square root on both sides, we obtain:
\begin{small}
\begin{align}
\|U_+ - U^\star R_Z\|_F &\leq \sqrt{1 - \tfrac{4\eta \sigma_r(\rho^\star)(1 - \delta_{2r})}{10}} \cdot \|Z - U^\star R_Z\|_F
\end{align}
\end{small}
%We know that, for $|x| \leq 1$, the following bounds hold:
%\begin{align}
%\sqrt{1 + x} \leq e^{x/2} \leq \sqrt{1 + x + x^2}.
%\end{align}
%Setting $x = - \tfrac{0.393}{\kappa \tau(\rho^\star)}$, we have:
%\begin{small}
%\begin{align}
%\|U_+ - U^\star R_Z\|_F &\leq e^{-\tfrac{0.393}{2 \kappa \tau(\rho^\star)} } \|Z - U^\star R_Z\|_F \leq e^{-\tfrac{4}{25 \tau(\rho^\star)}} \|Z - U^\star R_Z\|_F,
%\end{align}
%\end{small}
%where $\kappa = \tfrac{1 + \delta_{2r}}{1 - \delta_{2r}} \leq \tfrac{11}{9}$, assuming $\delta_{2r} \leq \tfrac{1}{10}$.
Let us define $\xi = \sqrt{1 - \tfrac{4\eta \sigma_r(\rho^\star)(1 - \delta_{2r})}{10}}$. 
%
%\textcolor{magenta}{Lyle: notice different definition on $\xi$  from previous version}
%
Using the definitions $Z = U + \mu (U - U_-)$ and $R_{Z} \in \arg\min\limits_{R \in \mathcal{O}} \|Z - U^\star R\|_{F}$, we get
\begin{small}
\begin{align}
\|U_+ - U^\star R_Z\|_F &\leq \xi \cdot \min_{R \in \mathcal{O}} \|Z - U^\star R\|_F = \xi \cdot \min_{R \in \mathcal{O}} \|U + \mu (U - U_-) - U^\star R\|_F \nonumber \\
				&= \xi \cdot \min_{R \in \mathcal{O}} \|U + \mu \left(U - U_-\right)  - (1 - \mu + \mu) U^\star R\|_F \nonumber \\
				&\stackrel{(i)}{\leq} \xi \cdot |1 + \mu| \cdot \min_{R \in \mathcal{O}} \|U - U^\star R \|_F + \xi \cdot |\mu| \cdot \min_{R \in \mathcal{O}}\|U_- - U^\star R\|_F + \xi \cdot |\mu| \cdot r \sigma_1(\rho^\star)^{1/2} \nonumber 
\end{align}
\end{small}
where $(i)$ follows from steps similar to those in Lemma \ref{lem:supp_02}. 
Further observe that $\min_{R \in \mathcal{O}}\|U_+ - U^\star R\|_F \leq \|U_+ - U^\star R_Z\|_F$, thus leading to:
\begin{small}
\begin{align}
\min_{R \in \mathcal{O}}\|U_+ - U^\star R\|_F \leq \xi \cdot|1 + \mu| \cdot \min_{R \in \mathcal{O}} \|U - U^\star R \|_F + \xi \cdot |\mu| \cdot \min_{R \in \mathcal{O}}\|U_- - U^\star R\|_F  + \xi \cdot |\mu| \cdot r \sigma_1(\rho^\star)^{1/2} \label{eq:decrease}
\end{align}
\end{small}
Including two subsequent iterations in a single two-dimensional first-order system, we get the following characterization: 
\begin{small}
\begin{align}
\begin{bmatrix}
\min_{R \in \mathcal{O}} \|U_{i+1} - U^\star R\|_F \\ \min_{R \in \mathcal{O}} \|U_i - U^\star R\|_F
\end{bmatrix} &\leq \begin{bmatrix}
\xi \cdot |1 + \mu| & \xi \cdot  |\mu| \\
1 & 0
\end{bmatrix} \cdot 
\begin{bmatrix}
\min_{R \in \mathcal{O}} \|U_i - U^\star R \|_F \\
\min_{R \in \mathcal{O}} \|U_{i-1} - U^\star R \|_F
\end{bmatrix}
\\
&\quad \quad \quad \quad + 
\begin{bmatrix}
1 \\
0
\end{bmatrix} \cdot \xi \cdot |\mu| \cdot \sigma_1(\rho^\star)^{1/2} \cdot r.
\end{align}
\end{small}
Now, let $x_j = \min_{R \in \mathcal{O}} \|U_j - U^\star R\|_F$. Then, we can write the above relation as 
\begin{align*}
\begin{bmatrix}
x_{i+1} \\ x_i
\end{bmatrix} &\leq \underbrace{\begin{bmatrix}
\xi \cdot |1 + \mu| & \xi \cdot  |\mu| \\
1 & 0
\end{bmatrix}}_{:= A} \cdot 
\begin{bmatrix}
x_i \\
x_{i-1}
\end{bmatrix}
 + 
\begin{bmatrix}
1 \\
0
\end{bmatrix} \cdot \xi \cdot |\mu| \cdot \sigma_1(\rho^\star)^{1/2} \cdot r,
\end{align*}
where we denote the ``contraction matrix'' by $A$.
Taking norms on both sides, we get
\begin{align}
\left\| \begin{bmatrix}
x_{i+1} \\ x_i
\end{bmatrix} \right\|_2 
&\leq \left\| A \cdot 
\begin{bmatrix}
x_i \\
x_{i-1}
\end{bmatrix}
 + 
\begin{bmatrix}
1 \\
0
\end{bmatrix} \cdot \xi \cdot |\mu| \cdot \sigma_1(\rho^\star)^{1/2} \cdot r \right\|_2 \nonumber \\
&\stackrel{(i)}{\leq} \left\| A \cdot 
\begin{bmatrix}
x_i \\
x_{i-1}
\end{bmatrix}
\right\|_2
 + 
 \left\|
\begin{bmatrix}
1 \\
0
\end{bmatrix} \cdot \xi \cdot |\mu| \cdot \sigma_1(\rho^\star)^{1/2} \cdot r \right\|_2 \nonumber\\
&\stackrel{(ii)}{\leq} \left\| A \right\|_2 \cdot 
\left\|
\begin{bmatrix}
x_i \\
x_{i-1}
\end{bmatrix}
\right\|_2
 + 
 \left\|
\begin{bmatrix}
1 \\
0
\end{bmatrix} \cdot \xi \cdot |\mu| \cdot \sigma_1(\rho^\star)^{1/2} \cdot r \right\|_2, \label{eq:one-step-bound}
\end{align}
where $(i)$ is by triangle inequality, and $(ii)$ is by Cauchy–Schwarz inequality.

Therefore, the convergence rate will be determined by the (maximum) eigenvalue of the contraction matrix $A$, which is given by
\begin{align*}
    \lambda_{1, 2} = \frac{\xi \cdot |1+\mu|}{2} \pm \sqrt{ \frac{\xi^2 (1+\mu)^2}{4} + \xi \cdot |\mu| }  \stackrel{(i)}{\implies} \max\{\lambda_1, \lambda_2 \} = \lambda_1 = \frac{\xi \cdot |1+\mu|}{2} + \sqrt{ \frac{\xi^2 (1+\mu)^2}{4} + \xi \cdot |\mu| },
\end{align*}
where $(i)$ follows since every term in $\lambda_{1, 2}$ is positive. 

% Thus, the maximum eigenvalue of the contraction matrix $A$ will be determined with 
% \begin{align*}
%     \lambda_1 = \frac{\xi \cdot |1+\mu|}{2} + \sqrt{ \frac{\xi^2 (1+\mu)^2}{4} + \xi \cdot |\mu| }.
% \end{align*}
To show accelerated convergence rate, we want the above eigenvalue (which determines the convergence rate) to be bounded by $1 - \sqrt{ \tfrac{1-\delta_{2r}}{1+\delta_{2r}} } $. To show this, first note that this term is bounded above as follows:
% To show this, we utilize the conventional bound on momentum: $0 \leq \mu < 1$:
%
\begin{align*}
    \lambda_1 = \frac{\xi \cdot |1+\mu|}{2} + \sqrt{ \frac{\xi^2 (1+\mu)^2}{4} + \xi \cdot |\mu| } &\stackrel{(i)}{\leq} \xi + \sqrt{\xi^2 + \xi} \\
    &\stackrel{(ii)}{\leq} \xi + \sqrt{2\xi} \\
    &\stackrel{(ii)}{\leq} (\sqrt{2} + 1) \sqrt{\xi},
\end{align*}
where $(i)$ is by the conventional bound on momentum: $0 < \mu < 1$, and $(ii)$ is by the relation $\xi^2 \leq \xi \leq \sqrt{\xi}$ for $0 \leq \xi \leq 1$.
Therefore, to show the accelerated rate of convergence, we want the following relation to hold: 
\begin{align}
\label{eq:xi-upperbound}
    (\sqrt{2} + 1)\sqrt{\xi} \leq 1 - \sqrt{ \tfrac{1-\delta_{2r}}{1+\delta_{2r}} } \iff \sqrt{\xi} \leq \frac{ \sqrt{1+\delta_{2r}} - \sqrt{1-\delta_{2r}} }{(\sqrt{2}+1) \sqrt{1+\delta_{2r}}}.
\end{align}
Recalling our definition of $\xi = \sqrt{1 - \tfrac{4\eta \sigma_r(\rho^\star)(1 - \delta_{2r})}{10}}$, the problem boils down to choosing the right step size $\eta$ such that the above inequality on $\xi$ in Eq. \eqref{eq:xi-upperbound} is satisfied. With simple algebra, we can show the following lower bound on $\eta$:
\begin{align*}
    \left[ 1 - \left( \frac{ \sqrt{1+\delta_{2r}} - \sqrt{1-\delta_{2r}} }{(\sqrt{2}+1) \sqrt{1+\delta_{2r}}} \right)^4 \right] \cdot \frac{10}{4 \sigma_r (\rho^\star) (1-\delta_{2r})} \leq \eta
\end{align*}
Finally, the argument inside the $\sqrt{\cdot}$ term of $\xi = \sqrt{1 - \tfrac{4\eta \sigma_r(\rho^\star)(1 - \delta_{2r})}{10}} > 0$ has to be non-negative, yielding the following upper bound on $\eta$:
\begin{align*}
    \eta \leq \frac{10}{4 \sigma_r (\rho^\star) (1-\delta_{2r})}.
\end{align*}
Combining two inequalities, and noting that the term $\left[ 1 - \left( \frac{ \sqrt{1+\delta_{2r}} - \sqrt{1-\delta_{2r}} }{(\sqrt{2}+1) \sqrt{1+\delta_{2r}}} \right)^4 \right]$ is bounded above by 1, we arrive at the following bound on $\eta$:
\begin{align}
\label{eq:step-size-bound}
    \left[ 1 - \left( \frac{ \sqrt{1+\delta_{2r}} - \sqrt{1-\delta_{2r}} }{(\sqrt{2}+1) \sqrt{1+\delta_{2r}}} \right)^4 \right] \cdot \frac{10}{4 \sigma_r (\rho^\star) (1-\delta_{2r})} \leq \eta \leq \frac{10}{4 \sigma_r (\rho^\star) (1-\delta_{2r})}.
\end{align}
In sum, for the specific $\eta$ satisfying the above bound, we have shown that
\begin{align*}
    \lambda_1 = \frac{\xi \cdot |1+\mu|}{2} + \sqrt{ \frac{\xi^2 (1+\mu)^2}{4} + \xi \cdot |\mu| } \leq 1 - \sqrt{\frac{1 - \delta_{2r}}{1 + \delta_{2r}} } 
    %= 1 - \sqrt{\kappa}.
\end{align*}

Above translates our original recursion in \eqref{eq:one-step-bound} as:
\begin{align}
\left\| 
\begin{bmatrix}
x_{i+1} \\
x_i
\end{bmatrix}  
\right\|_2
&\leq \left\| A \right\|_2 \cdot 
\left\|
\begin{bmatrix}
x_i \\
x_{i-1}
\end{bmatrix}
\right\|_2
 + 
 \left\|
\begin{bmatrix}
1 \\
0
\end{bmatrix} \cdot \xi \cdot |\mu| \cdot \sigma_1(\rho^\star)^{1/2} \cdot r \right\|_2 \nonumber \\
&\leq \left( 1 - \sqrt{ \tfrac{1-\delta_{2r}}{1+\delta_{2r}} } \right) \cdot 
\left\|
\begin{bmatrix}
x_i \\
x_{i-1}
\end{bmatrix}
\right\|_2
 + 
 \left\|
\begin{bmatrix}
1 \\
0
\end{bmatrix} \cdot \xi \cdot |\mu| \cdot \sigma_1(\rho^\star)^{1/2} \cdot r \right\|_2 \nonumber \\
&= \left( 1 - \sqrt{ \tfrac{1-\delta_{2r}}{1+\delta_{2r}} } \right) \cdot 
\left\|
\begin{bmatrix}
x_i \\
x_{i-1}
\end{bmatrix}
\right\|_2
+
\xi \cdot |\mu| \cdot \sigma_1(\rho^\star)^{1/2} \cdot r, 
\label{eq:acc-one-step}
\end{align}
where the last equality is by the definition of $\ell_2$-norm.

Unrolling the recursion in Eq.~\eqref{eq:acc-one-step} for $J$ iterations, we get
\begin{align*}
\left\| 
\begin{bmatrix}
x_{J+1} \\
x_J
\end{bmatrix}  
\right\|_2
&\leq \left( 1 - \sqrt{ \tfrac{1-\delta_{2r}}{1+\delta_{2r}} } \right)^{J+1}  
\left\|
\begin{bmatrix}
x_0 \\
x_{-1}
\end{bmatrix}
\right\|_2
+
\xi \cdot |\mu| \cdot \sigma_1(\rho^\star)^{1/2} \cdot r \cdot \sum_{i=0}^J \left(1 - \sqrt{ \tfrac{1-\delta_{2r}}{1+\delta_{2r}} }\right)^i \\
&= 
\left( 1 - \sqrt{ \tfrac{1-\delta_{2r}}{1+\delta_{2r}} } \right)^{J+1} 
\left\|
\begin{bmatrix}
x_0 \\
x_{-1}
\end{bmatrix}
\right\|_2
+
\xi \cdot |\mu| \cdot \sigma_1(\rho^\star)^{1/2} \cdot r \cdot \left( 1 - \left(1-\sqrt{ \tfrac{1-\delta_{2r}}{1+\delta_{2r}} }\right)^{J+1} \right) \left(1-\sqrt{ \tfrac{1-\delta_{2r}}{1+\delta_{2r}} }\right)^{-1}  \\
&= 
\left( 1 - \sqrt{ \tfrac{1-\delta_{2r}}{1+\delta_{2r}} } \right)^{J+1} 
\left\|
\begin{bmatrix}
x_0 \\
x_{-1}
\end{bmatrix}
\right\|_2 + O(\mu)  %c_1.
% \\
% &=
% \left( 1 - \sqrt{\kappa} \right)^{J+1}  
% \left\|
% \begin{bmatrix}
% x_0 \\
% x_{-1}
% \end{bmatrix}
% \right\|_2
% +
% O(\mu).
\end{align*}
Finally, computing the $\ell_2$-norm explicitly and resubstituting $x_j = \min_{R \in \mathbb{O}} \|U_j - U^\star R \|_F$, 
%and also assuming $U_0$ and $U_{-1}$ are initialized same, 
we get
\begin{align*}
\min_{R \in \mathcal{O}} \|U_{J+1} - U^\star R\|_F 
&\leq 
    \left(1-\sqrt{ \tfrac{1-\delta_{2r}}{1+\delta_{2r}} }\right)^{J+1} \left( \min_{R \in \mathcal{O}} \|U_0 - U^\star R\|_F^2 + \min_{R \in \mathcal{O}} \|U_{-1} - U^\star R\|_F^2 \right)^{1/2} +  O(\mu). % c_1  \\
    % &=
    % \left(1-\sqrt{ \tfrac{1-\delta_{2r}}{1+\delta_{2r}} }\right)^{J+1} \left( \min_{R \in \mathcal{O}} \|U_0 - U^\star R\|_F^2 + \min_{R \in \mathcal{O}} \|U_{-1} - U^\star R\|_F^2 \right)^{1/2} + c_1
\end{align*}
}

\section{Supporting lemmata}

In this section, we present a series of lemmata, used for the main result of the paper.

\begin{lemma}{\label{lem:supp_00}}
Let $U \in \mathbb{C}^{d \times r}$ and $U^\star \in \mathbb{C}^{d \times r}$, such that $\|U - U^\star R\|_F \leq \tfrac{\sigma_r(\rho^\star)^{1/2}}{10^3 \sqrt{\kappa \tau(\rho^\star)}}$ for some $R \in \mathcal{O}$, where $\rho^\star = U^\star U^{\star \dagger}$, $\kappa := \tfrac{1 + \delta_{2r}}{1 - \delta_{2r}} > 1$, for $\delta_{2r} \leq \tfrac{1}{10}$, and $\tau(\rho^\star) := \tfrac{\sigma_1(\rho^\star)}{\sigma_r(\rho^\star)} > 1$. Then:
\begin{align}
\sigma_1(\rho^\star)^{1/2} \left(1  - \tfrac{1}{10^3}\right) &\leq \sigma_1(U) \leq \sigma_1(\rho^\star)^{1/2} \left(1 + \tfrac{1}{10^3} \right) \\
\sigma_r(\rho^\star)^{1/2} \left(1  - \tfrac{1}{10^3}\right) &\leq \sigma_r(U) \leq \sigma_r(\rho^\star)^{1/2} \left(1 + \tfrac{1}{10^3} \right)
\end{align}
\end{lemma}

\begin{proof}
By the fact $\|\cdot\|_2 \leq \|\cdot\|_F$ and using Weyl's inequality for perturbation of singular values \cite[Theorem 3.3.16]{horn1990matrix}, we have:
\begin{align}
\left|\sigma_i(U) - \sigma_i(U^\star)\right| \leq \tfrac{\sigma_r(\rho^\star)^{1/2}}{10^3 \sqrt{\kappa \tau(\rho^\star)}} \leq \tfrac{\sigma_r(\rho^\star)^{1/2}}{10^3}, \quad 1 \leq i \leq r.
\end{align}
Then, 
\begin{align}
- \tfrac{\sigma_r(\rho^\star)^{1/2}}{10^3} &\leq \sigma_1(U) - \sigma_1(U^\star) \leq \tfrac{\sigma_r(\rho^\star)^{1/2}}{10^3} \Rightarrow \\
\sigma_1(\rho^\star)^{1/2} - \tfrac{\sigma_r(\rho^\star)^{1/2}}{10^3} &\leq \sigma_1(U) \leq \sigma_1(\rho^\star)^{1/2} + \tfrac{\sigma_r(\rho^\star)^{1/2}}{10^3} \Rightarrow \\
\sigma_1(\rho^\star)^{1/2} \left(1  - \tfrac{1}{10^3}\right) &\leq \sigma_1(U) \leq \sigma_1(\rho^\star)^{1/2} \left(1 + \tfrac{1}{10^3} \right).
\end{align}
Similarly:
\begin{align}
- \tfrac{\sigma_r(\rho^\star)^{1/2}}{10^3} &\leq \sigma_r(U) - \sigma_r(U^\star) \leq \tfrac{\sigma_r(\rho^\star)^{1/2}}{10^3} \Rightarrow \\
\sigma_r(\rho^\star)^{1/2} - \tfrac{\sigma_r(\rho^\star)^{1/2}}{10^3} &\leq \sigma_r(U) \leq \sigma_r(\rho^\star)^{1/2} + \tfrac{\sigma_r(\rho^\star)^{1/2}}{10^3} \Rightarrow \\
\sigma_r(\rho^\star)^{1/2} \left(1  - \tfrac{1}{10^3}\right) &\leq \sigma_r(U) \leq \sigma_r(\rho^\star)^{1/2} \left(1 + \tfrac{1}{10^3} \right).
\end{align}
In the above, we used the fact that $\sigma_i(U^\star) = \sigma_i(\rho^\star)^{1/2}$, for all $i$, and the fact that $\sigma_i(\rho^\star)^{1/2} \geq \sigma_j(\rho^\star)^{1/2}$, for $i \leq j$.
\end{proof}

\begin{lemma}{\label{lem:supp_02}}
Let $U \in \mathbb{C}^{d \times r}, U_{-} \in \mathbb{C}^{d \times r}$, and $U^\star \in \mathbb{C}^{d \times r}$, such that 
\begin{small}$\min_{R \in \mathcal{O}} \|U - U^\star R\|_F \leq \tfrac{\sigma_r(\rho^\star)^{1/2}}{10^3 \sqrt{\kappa\tau(\rho^\star)}}$\end{small} ~and~
\begin{small}$\min_{R \in \mathcal{O}} \|U_- - U^\star R\|_F \leq \tfrac{\sigma_r(\rho^\star)^{1/2}}{10^3 \sqrt{\kappa\tau(\rho^\star)}}$\end{small}, where $\rho^\star = U^\star U^{\star \dagger}$, and \begin{small}$\kappa := \tfrac{1 + \delta_{2r}}{1 - \delta_{2r}} > 1$\end{small}, for \begin{small}$\delta_{2r} \leq \tfrac{1}{10}$\end{small}, and \begin{small}$\tau(\rho^\star) := \tfrac{\sigma_1(\rho^\star)}{\sigma_r(\rho^\star)} > 1$\end{small}.
Set the momentum parameter as \begin{small}$\mu = \frac{\sigma_r(\rho^\star)^{1/2}}{10^3 \sqrt{\kappa\tau(\rho^\star)}} \cdot \frac{\varepsilon}{4 \cdot \sigma_1(\rho^\star)^{1/2} \cdot r}$\end{small}, for $\varepsilon \in (0, 1)$ user-defined.
Then, 
\begin{align}
\|Z - U^\star R_Z\|_F \leq \left(\tfrac{3}{2} + 2|\mu|\right) \cdot \tfrac{\sigma_r(\rho^\star)^{1/2}}{10^3\sqrt{\kappa \tau(\rho^\star)}}.
\end{align}
\end{lemma}

\begin{proof}
Let $R_{U} \in \arg\min_{R \in \mathcal{O}} \|U - U^\star\|_{F}$ and $R_{U_{-}} \in \arg\min_{R \in \mathcal{O}} \|U_{-} - U^\star R\|_{F}$.
By the definition of the distance function:
\begin{align}
\|Z - U^\star R_Z\|_F &= \min_{R \in \mathcal{O}} \|Z - U^\star R\|_F = \min_{R \in \mathcal{O}} \| U + \mu (U - U_{-}) - U^\star R \|_F \\
				     &= \min_{R \in \mathcal{O}} \|U + \mu (U - U_{-}) - (1 - \mu + \mu) U^\star R\|_F \\
				     &\leq |1 + \mu|\cdot\|U - U^{*}R_{U}\|_{F} + |\mu|\cdot||U_{-} - U^{*}R_{U_{-}}||_{F}\\
				     &= |1 + \mu|\cdot\|U - U^{*}R_{U}\|_{F} + |\mu|\cdot\|U_{-} - U^{*}R_{U} - U^{*}R_{U_{-}} + U^{*}R_{U_{-}}\|_{F}\\
				     &= |1 + \mu|\cdot\|U - U^{*}R_{U}\|_{F} + |\mu|\cdot\|(U_{-} - U^{*}R_{U_{-}}) + U^{*}(R_{U_{-}} - R_{U})\|_{F}\\
				     &\leq |1 + \mu| \cdot \min_{R \in \mathcal{O}} \|U - U^\star R\|_F + |\mu| \cdot \min_{R \in \mathcal{O}} \|U_- - U^\star R\|_F + |\mu| \cdot \|U^\star (R_U - R_{U_-})\|_F \\
				     &\leq |1 + \mu| \cdot \min_{R \in \mathcal{O}} \|U - U^\star R\|_F + |\mu| \cdot \min_{R \in \mathcal{O}} \|U_- - U^\star R\|_F + 2|\mu| \cdot \sigma_1(\rho^\star)^{1/2} r \\
				     &\stackrel{(i)}{\leq} \left(\tfrac{3}{2} + 2|\mu|\right) \cdot \tfrac{\sigma_r(\rho^\star)^{1/2}}{10^3\sqrt{\kappa \tau(\rho^\star)}}
\end{align}
where $(i)$ is due to the fact that $\mu \leq \frac{\sigma_r(\rho^\star)^{1/2}}{10^3 \sqrt{\kappa\tau(\rho^\star)}} \cdot \frac{1}{4 \cdot \sigma_1(\rho^\star)^{1/2} \cdot r}$.
We keep $\mu$ in the expression, but we use it for clarity for the rest of the proof.
%Using $\mu = \tfrac{1}{500}$, we get the additional result.
\end{proof}

\begin{corollary}{\label{cor:2}}
Let $Z \in \mathbb{C}^{d \times r}$ and $U^\star \in \mathbb{C}^{d \times r}$, such that 
$\|Z - U^\star R\|_F \leq \left(\frac{3}{2} + 2|\mu|\right) \cdot \tfrac{\sigma_r(\rho^\star)^{1/2}}{10^3\sqrt{\kappa \tau(\rho^\star)}}$ for some $R \in \mathcal{O}$, and $\rho^\star = U^\star U^{\star \dagger}$. Then:
\begin{align}
\sigma_1(\rho^\star)^{1/2} \left(1  - \left(\tfrac{3}{2} + 2|\mu|\right)\tfrac{1}{10^3}\right) &\leq \sigma_1(Z) \leq \sigma_1(\rho^\star)^{1/2} \left(1 + \left(\tfrac{3}{2} + 2|\mu|\right)\tfrac{1}{10^3} \right) \\
\sigma_r(\rho^\star)^{1/2} \left(1  - \left(\tfrac{3}{2} + 2|\mu|\right)\tfrac{1}{10^3}\right) &\leq \sigma_r(Z) \leq \sigma_r(\rho^\star)^{1/2} \left(1 + \left(\tfrac{3}{2} + 2|\mu|\right)\tfrac{1}{10^3} \right).
\end{align}
Given that  \begin{small}$\mu = \frac{\sigma_r(\rho^\star)^{1/2}}{10^3 \sqrt{\kappa\tau(\rho^\star)}} \cdot \frac{\varepsilon}{4 \cdot \sigma_1(\rho^\star)^{1/2} \cdot r} \leq \tfrac{1}{10^3}$\end{small}, we get:
\begin{align}
0.998\cdot \sigma_1(\rho^\star)^{1/2} &\leq \sigma_1(Z) \leq 1.0015 \cdot \sigma_1(\rho^\star)^{1/2} \\
0.998\cdot \sigma_r(\rho^\star)^{1/2}  &\leq \sigma_r(Z) \leq 1.0015 \cdot \sigma_r(\rho^\star)^{1/2}.
\end{align}
\end{corollary}

\begin{proof}
The proof follows similar motions as in Lemma \ref{lem:supp_00}.
\end{proof}

\begin{corollary}{\label{cor:01}}
Under the same assumptions of Lemma \ref{lem:supp_00} and Corollary \ref{cor:2}, and given the assumptions on $\mu$, we have:
\begin{align}
\tfrac{99}{100} \cdot \|\rho^\star\|_2 \leq \|ZZ^\dagger\|_2 \leq \tfrac{101}{100} \cdot \|\rho^\star\|_2 \\
\tfrac{99}{100} \cdot \|\rho^\star\|_2 \leq \|Z_0Z_0^\dagger\|_2 \leq \tfrac{101}{100} \cdot \|\rho^\star\|_2 
\end{align}
and 
\begin{align}
\tfrac{99}{101} \cdot \|Z_0Z_0^\dagger\|_2 \leq \|ZZ^\dagger\|_2 \leq \tfrac{101}{99} \cdot \|Z_0Z_0^\dagger\|_2
\end{align}
\end{corollary}
\begin{proof}
The proof is easily derived based on the quantities from Lemma \ref{lem:supp_00} and Corollary \ref{cor:2}.
\end{proof}

\begin{corollary}{\label{cor:1}}
Let $Z \in \mathbb{C}^{d \times r}$ and $U^\star \in \mathbb{C}^{d \times r}$, such that 
$\|Z - U^\star R\|_F \leq \left(\tfrac{3}{2} + 2|\mu|\right) \cdot \tfrac{\sigma_r(\rho^\star)^{1/2}}{10^3 \sqrt{\kappa\tau(\rho^\star)}}$ for some $R \in \mathcal{O}$, and $\rho^\star = U^\star U^{\star \dagger}$.
Define $\tau(W) = \frac{\sigma_1(W)}{\sigma_r(W)}$. 
Then:
\begin{align}
\tau(ZZ^\dagger) \leq \beta^2 \tau(\rho^\star),
\end{align}
where $\beta := \frac{1 + \left(\tfrac{3}{2} + 2|\mu|\right) \cdot \tfrac{1}{10^3}}{1 - \left(\tfrac{3}{2} + 2|\mu| \right) \cdot \tfrac{1}{10^3}} > 1$.
for $\mu \leq \frac{\sigma_r(\rho^\star)^{1/2}}{10^3 \sqrt{\kappa\tau(\rho^\star)}} \cdot \frac{1}{4 \cdot \sigma_1(\rho^\star)^{1/2} \cdot r}$.
\end{corollary}

\begin{proof}
The proof uses the definition of the condition number $\tau(\cdot)$ and the results from Lemma \ref{lem:supp_00} and and Corollary \ref{cor:2}.
\end{proof}

\begin{lemma}{\label{lem:equiveta}}
Consider the following three step sizes:
\begin{align}
\eta &= \frac{1}{4\left( (1 + \delta_{2r}) \|Z_0Z_0^\dagger\|_2 + \|\mathcal{A}^\dagger \left(\mathcal{A}( Z_0Z_0^\dagger ) - y\right)\|_2 \right) } \\
\widehat{\eta} &= \frac{1}{4\left( (1 + \delta_{2r}) \|ZZ^\dagger\|_2 + \|\mathcal{A}^\dagger \left(\mathcal{A}( ZZ^\dagger ) - y\right) Q_ZQ_Z^\dagger \|_2 \right) } \\
\eta^\star &= \frac{1}{4\left( (1 + \delta_{2r}) \|\rho^{\star \top}\|_2 + \|\mathcal{A}^\dagger \left(\mathcal{A}( \rho^{\star} ) - y\right) \|_2 \right) }.
\end{align}
Here, $Z_0 \in \mathbb{C}^{d \times r}$ is the initial point, $Z \in \mathbb{C}^{d \times r}$ is the current point, $\rho^\star \in \mathbb{C}^{d \times d}$ is the optimal solution, and $Q_Z$ denotes a basis of the column space of $Z$. 
Then, under the assumptions that $\min_{R \in \mathcal{O}} \|U - U^\star R\|_F \leq \tfrac{\sigma_r(\rho^\star)^{1/2}}{10^3 \sqrt{\kappa\tau(\rho^\star)}}$, 
and $\min_{R \in \mathcal{O}} \|Z - U^\star R\|_F \leq \left( \tfrac{3}{2} + 2|\mu|\right) \cdot \tfrac{\sigma_r(\rho^\star)^{1/2}}{10^3 \sqrt{\kappa \tau(\rho^\star)}}$, and assuming \begin{small}$\mu = \frac{\sigma_r(\rho^\star)^{1/2}}{10^3 \sqrt{\kappa\tau(\rho^\star)}} \cdot \frac{\varepsilon}{4 \cdot \sigma_1(\rho^\star)^{1/2} \cdot r}$\end{small}, for the user-defined parameter $\varepsilon \in (0, 1)$, we have:
\begin{align}
\tfrac{10}{9} \eta \geq \widehat{\eta} \geq \tfrac{10}{10.5} \eta, \quad \text{and} \quad \tfrac{100}{102} \eta^\star \leq \eta \leq \tfrac{102}{100} \eta^\star
\end{align}
\end{lemma}

\begin{proof}
The assumptions of the lemma are identical to that of Corollary \ref{cor:01}.
Thus, we have:
$\tfrac{99}{100} \cdot \|U^\star\|_2^2 \leq \|Z\|_2^2 \leq \tfrac{101}{100} \cdot \|U^\star\|_2^2$, % \\
$\tfrac{99}{100} \cdot \|U^\star\|_2^2 \leq \|Z_0\|_2^2 \leq \tfrac{101}{100} \cdot \|U^\star\|_2^2$,
and 
$\tfrac{99}{101} \cdot \|Z_0\|_2^2 \leq \|Z\|_2^2 \leq \tfrac{101}{99} \cdot \|Z_0\|_2^2.$
We focus on the inequality $\widehat{\eta} \geq \tfrac{10}{10.5} \eta$. 
Observe that:
\begin{small}
\begin{align}
\left\|\mathcal{A}^\dagger \left(\mathcal{A}(ZZ^\dagger) - y \right)Q_ZQ_Z^\dagger \right\|_2 &\leq \left\|\mathcal{A}^\dagger \left(\mathcal{A}(ZZ^\dagger) - y \right) \right\|_2 \\
&= \left\|\mathcal{A}^\dagger \left(\mathcal{A}(ZZ^\dagger) - y \right) - \mathcal{A}^\dagger \left(\mathcal{A}(Z_0Z_0^\dagger) - y \right) + \mathcal{A}^\dagger \left(\mathcal{A}(Z_0Z_0^\dagger) - y \right)\right\|_2 \\
&\stackrel{(i)}{\leq} (1 + \delta_{2r}) \left\|ZZ^\dagger - Z_0 Z_0^\dagger \right\|_F + \left \|\mathcal{A}^\dagger \left(\mathcal{A}(Z_0Z_0^\dagger) - y \right)\right\|_2 \\
&\leq (1 + \delta_{2r}) \left\|ZZ^\dagger - U^\star U^{\star \dagger} \right\|_F + (1 + \delta_{2r}) \left\|Z_0 Z_0^\dagger - U^\star U^{\star \dagger} \right\|_F \\
&\quad \quad \quad \quad + \left \|\mathcal{A}^\dagger \left(\mathcal{A}(Z_0Z_0^\dagger) - y \right)\right\|_2
\end{align}
\end{small}
where $(i)$ is due to smoothness via RIP constants of the objective and the fact $\|\cdot\|_2 \leq \|\cdot \|_F$.
For the first two terms on the right-hand side, where $R_Z$ is the minimizing rotation matrix for $Z$, we obtain:
\begin{align}
\|ZZ^\dagger - U^\star U^{\star \dagger}\|_F &= \|ZZ^\dagger - U^\star R_Z Z^\dagger + U^\star R_Z Z^\dagger - U^\star U^{\star \dagger}\|_F \\
&= \|(Z -U^\star R_Z) Z^\dagger + U^{\star} R_Z ( Z - U^\star R_Z)^\dagger\|_F \\
&\leq \|Z\|_2 \cdot \|Z - U^\star R_Z\|_F  + \|U^\star \|_2 \cdot \|Z - U^\star R_Z\|_F \\
&\leq \left(\|Z\|_2 + \|U^\star\|_2 \right) \cdot \|Z - U^\star R_Z\|_F \\ 
&\stackrel{(i)}{\leq} \left( \sqrt{\tfrac{101}{99}} + \sqrt{\tfrac{100}{99}}\right) \|Z_0\|_2 \cdot \|Z - U^\star R_Z\|_F \\
&\stackrel{(ii)}{\leq}  \left( \sqrt{\tfrac{101}{99}} + \sqrt{\tfrac{100}{99}}\right) \|Z_0\|_2 \cdot 0.001 \sigma_r(\rho^\star)^{1/2} \\
&\leq \left( \sqrt{\tfrac{101}{99}} + \sqrt{\tfrac{100}{99}}\right) \cdot 0.001 \cdot \sqrt{\tfrac{100}{99}} \cdot \|Z_0\|_2^2
\end{align}
where $(i)$ is due to the relation of $\|Z\|_2$ and $\|U^\star\|_2$ derived above,
$(ii)$ is due to Lemma \ref{lem:supp_02}.
Similarly:
\begin{align}
\|Z_0Z_0^\dagger - U^\star U^{\star \dagger}\|_F \leq \left( \sqrt{\tfrac{101}{99}} + \sqrt{\tfrac{100}{99}}\right) \cdot 0.001 \cdot \sqrt{\tfrac{100}{99}} \cdot \|Z_0\|_2^2
\end{align}
Using these above, we obtain:
\begin{small}
\begin{align}
\left\|\mathcal{A}^\dagger \left(\mathcal{A}(ZZ^\dagger) - y \right)Q_ZQ_Z^\dagger \right\|_2 &\leq \tfrac{4.1 (1 + \delta_{2r})}{10^3} \|Z_0Z_0^\dagger\|_2 + \left \|\mathcal{A}^\dagger \left(\mathcal{A}(Z_0Z_0^\dagger) - y \right)\right\|_2
\end{align}
\end{small}
Thus:
\begin{align}
\widehat{\eta} &= \frac{1}{4\left( (1 + \delta_{2r}) \|ZZ^\dagger\|_2 + \|\mathcal{A}^\dagger \left(\mathcal{A}( ZZ^\dagger ) - y\right) Q_ZQ_Z^\dagger \|_2 \right) } \\
		        &\geq \frac{1}{4\left( (1 + \delta_{2r}) \tfrac{101}{99} \|Z_0Z_0\|_2 + \right) + \tfrac{4.1 (1 + \delta_{2r})}{10^3} \|Z_0Z_0^\dagger\|_2 + \left \|\mathcal{A}^\dagger \left(\mathcal{A}(Z_0Z_0^\dagger) - y \right)\right\|_2} \\ 
%		        &= \frac{1}{4\left( (1 + \delta_{2r})\left(\tfrac{101}{99} + \tfrac{4}{100} \right) \|Z_0Z_0^\dagger\|_2 + \|\mathcal{A}^\dagger \left(\mathcal{A}( Z_0Z_0^\dagger ) - y\right) \|_2 \right) } \\ 
		        &\geq \frac{1}{4\left( \tfrac{10.5}{10} \cdot (1 + \delta_{2r})\|Z_0Z_0^\dagger\|_2 + \|\mathcal{A}^\dagger \left(\mathcal{A}( Z_0Z_0^\dagger ) - y\right) \|_2 \right) } \\ &\geq \tfrac{10}{10.5} \eta
\end{align}
Similarly, one gets $\widehat{\eta} \leq \tfrac{10}{9} \eta$.

For the relation between $\eta$ and $\eta^\star$, we will prove here the lower bound; similar motions lead to the upper bound also.
By definition, and using the relations in Corollary \ref{cor:01}, we get:
\begin{align}
\eta &= \frac{1}{4\left( (1 + \delta_{2r}) \|Z_0Z_0^\dagger\|_2 + \|\mathcal{A}^\dagger \left(\mathcal{A}( Z_0Z_0^\dagger ) - y\right)\|_2 \right) } \\
         &\geq \frac{1}{4\left( (1 + \delta_{2r}) \tfrac{101}{100} \|\rho^{\star \top}\|_2 + \|\mathcal{A}^\dagger \left(\mathcal{A}( Z_0Z_0^\dagger ) - y\right)\|_2 \right) }
\end{align}
For the gradient term, we observe:
\begin{align}
\left\|\mathcal{A}^\dagger \left(\mathcal{A}( Z_0Z_0^\dagger ) - y\right)\right\|_2 &\leq \left\|\mathcal{A}^\dagger \left(\mathcal{A}( Z_0Z_0^\dagger ) - y\right) -\mathcal{A}^\dagger \left(\mathcal{A}( \rho^\star ) - y\right) \right\|_2 + \left\|\mathcal{A}^\dagger \left(\mathcal{A}( \rho^\star ) - y\right)\right\|_2 \\
														     &\stackrel{(i)}{=} \left\|\mathcal{A}^\dagger \left(\mathcal{A}( Z_0Z_0^\dagger ) - y\right) -\mathcal{A}^\dagger \left(\mathcal{A}( \rho^\star ) - y\right) \right\|_2 \\
														     &\stackrel{(ii)}{\leq} (1 + \delta_{2r}) \left\|Z_0Z_0^\dagger - U^\star U^{\star \dagger} \right\|_F \\
														     &\stackrel{(iii)}{\leq} (1 + \delta_{2r}) \left( \|Z_0\|_2 + \|U^\star\|_2 \right) \cdot \|Z - U^\star R_Z\|_F \\
														     &\stackrel{(iv)}{\leq} (1 + \delta_{2r}) \left( \sqrt{\tfrac{101}{100}} + 1\right) \|U^\star\|_2 \cdot 0.001 \cdot \|U^\star\|_2^2 \\
														     &\leq 0.002 \cdot (1 + \delta_{2r}) \|\rho^\star\|_2
\end{align}
where $(i)$ is due to $\left\|\mathcal{A}^\dagger \left(\mathcal{A}( \rho^\star ) - y\right)\right\|_2 = 0$,
$(ii)$ is due to the restricted smoothness assumption and the RIP, 
$(iii)$ is due to the bounds above on $\left\|Z_0Z_0^\dagger - U^\star U^{\star \dagger} \right\|_F$,
$(iv)$ is due to the bounds on $\|Z_0\|_2$, w.r.t. $\|U^\star\|_2$, as well as the bound on $\|Z - U^\star R\|_F$.

Thus, in the inequality above, we get:
\begin{align}
\eta &\geq \frac{1}{4\left( (1 + \delta_{2r}) \tfrac{101}{100} \|\rho^{\star \top}\|_2 + \|\mathcal{A}^\dagger \left(\mathcal{A}( Z_0Z_0^\dagger ) - y\right)\|_2 \right) } \\
         &\geq \frac{1}{4\left( (1 + \delta_{2r}) \tfrac{101}{100} \|\rho^{\star \top}\|_2 + 0.001 \cdot (1 + \delta_{2r}) \|\rho^\star\|_2 + \|\mathcal{A}^\dagger \left(\mathcal{A}( \rho^\star ) - y\right)\|_2 \right) } \\
         &\geq \frac{1}{4\left( (1 + \delta_{2r}) \tfrac{102}{100} \|\rho^{\star \top}\|_2 + \|\mathcal{A}^\dagger \left(\mathcal{A}( \rho^\star ) - y\right)\|_2 \right) } \geq \tfrac{100}{102} \eta^\star
\end{align}
Similarly, one can show that $\frac{102}{100}\eta^{\star} \geq \eta$. %\footnote{\tajayi{Should this be ``Similarly, one can show that $\frac{102}{100}\eta^{*} \geq \eta$"? As written, it is the same as what is already proven.}}
\end{proof}

\begin{lemma}{\label{lem:000}}
Let $U \in \mathbb{C}^{d \times r}, U_{-} \in \mathbb{C}^{d \times r}$, and $U^\star \in \mathbb{C}^{d \times r}$, such that 
\begin{small}$\min_{R \in \mathcal{O}} \|U - U^\star R\|_F \leq \tfrac{\sigma_r(\rho^\star)^{1/2}}{10^3 \sqrt{\kappa\tau(\rho^\star)}}$\end{small} ~and~
\begin{small}$\min_{R \in \mathcal{O}} \|U_- - U^\star R\|_F  \leq \tfrac{\sigma_r(\rho^\star)^{1/2}}{10^3 \sqrt{\kappa\tau(\rho^\star)}}$\end{small}, %for some $R \in \mathcal{O}$, 
where $\rho^\star = U^\star U^{\star \dagger}$, and $\kappa := \tfrac{1 + \delta_{2r}}{1 - \delta_{2r}} > 1$, for $\delta_{2r} \leq \tfrac{1}{10}$, and $\tau(\rho^\star) := \tfrac{\sigma_1(\rho^\star)}{\sigma_r(\rho^\star)} > 1$.
By Lemma \ref{lem:supp_02}, the above imply also that: \begin{small}$\|Z - U^\star R_Z\|_F \leq \left(\tfrac{3}{2} + 2|\mu|\right) \cdot \tfrac{\sigma_r(\rho^\star)^{1/2}}{10^3\sqrt{\kappa\tau(\rho^\star)}}$\end{small}.
Then, under RIP assumptions of the mapping $\mathcal{A}$, we have:
\begin{small}
 \begin{align}
\Big \langle &\mathcal{A}^\dagger(\mathcal{A}(ZZ^\dagger) - y), (Z - U^\star R_Z)(Z - U^\star R_Z)^\dagger \Big \rangle \nonumber \\ 
		  &\geq - \Bigg( \theta \sigma_r(\rho^\star) \cdot \|Z - U^\star R_Z\|_F^2 + \tfrac{10.1}{100} \beta^2 \cdot \widehat{\eta} \cdot \tfrac{(1 + 2|\mu|)^2}{\left(1  - \left(1 + 2|\mu|\right)\tfrac{1}{200}\right)^2} \cdot \|\mathcal{A}^\dagger(\mathcal{A}(ZZ^\dagger) - y) \cdot Z\|_F^2\Bigg)
\end{align}
\end{small}
where $$\theta = \tfrac{(1 - \delta_{2r}) \left(1 + (1 + 2|\mu|)\tfrac{1}{200}\right)^2}{10^3} +  (1 + \delta_{2r})  \left(2+\left(1 + 2|\mu|\right) \cdot \tfrac{1}{200}\right) \left(1 + 2|\mu|\right) \cdot \tfrac{1}{200},$$
and $\widehat{\eta}= \tfrac{1}{4((1+\delta_r)\|ZZ^\dagger\|_2+\|\mathcal{A}^\dagger
\left(\mathcal{A}(ZZ^\dagger) - y \right)Q_ZQ_Z^\dagger\|_2)}$.
\end{lemma}

\begin{proof}
%The following steps follow Lemma 16 in \cite{bhojanapalli2016dropping}.
%We lower bound $\left \langle \mathcal{A}^\dagger(\mathcal{A}(ZZ^\dagger) - y), (Z - U^\star R)(Z - U^\star R)^\dagger \right \rangle$ as follows.
First, denote $\Delta := Z - U^\star R_Z$.
Then:
\begin{small}
 \begin{align}
\Big \langle &\mathcal{A}^\dagger(\mathcal{A}(ZZ^\dagger) - y), (Z - U^\star R_Z)(Z - U^\star R_Z)^\dagger \Big \rangle \nonumber \\ 
&\stackrel{(i)}{=} \left \langle \mathcal{A}^\dagger(\mathcal{A}(ZZ^\dagger) - y) \cdot Q_{\Delta} Q_{\Delta}^\dagger, \Delta_Z \Delta_Z^\dagger \right \rangle \nonumber \\
&\geq -  \left|\text{Tr}\left(\mathcal{A}^\dagger(\mathcal{A}(ZZ^\dagger) - y) \cdot Q_{\Delta} Q_{\Delta}^\dagger \cdot \Delta_Z \Delta_Z^\dagger\right)\right| \nonumber \\ 
&\stackrel{(ii)}{\geq} - \| \mathcal{A}^\dagger(\mathcal{A}(ZZ^\dagger) - y) \cdot Q_{\Delta} Q_{\Delta}^\dagger \|_2 \cdot \text{Tr}( \Delta_Z \Delta_Z^\dagger)  \nonumber \\
&\stackrel{(iii)}{\geq} - \left( \| \mathcal{A}^\dagger(\mathcal{A}(ZZ^\dagger) - y) \cdot Q_{Z} Q_{Z}^\dagger \|_2+  \|\mathcal{A}^\dagger(\mathcal{A}(ZZ^\dagger) - y) \cdot Q_{U^\star} Q_{U^\star}^\dagger \|_2 \right) \|Z - U^\star R_Z\|_F^2 \label{proofsr1:eq_11}
\end{align}
\end{small}
Note that  $(i)$ follows from the fact $\Delta_Z =\Delta_Z Q_{\Delta} Q_{\Delta}^\dagger $, for a matrix $Q$ that spans the row space of $\Delta_Z$, and $(ii)$ follows from $|\text{Tr}(AB)| \leq \|A\|_2 \trace(B)$, for PSD matrix $B$ (Von Neumann's trace inequality~\cite{mirsky1975trace}). 
For the transformation in $(iii)$, we use that fact that the row space of $\Delta_Z$, $\text{\textsc{Span}}(\Delta_Z)$, is a subset of $\text{\textsc{Span}}(Z \cup U^\star)$, as $\Delta_Z$ is a linear combination of $U$ and $U^\star$. % Finally $(iv)$ follows from the hypothesis of the Lemma. 

To bound the first term in equation~\eqref{proofsr1:eq_11}, we observe:
\begin{small}
\begin{align}
&\|\mathcal{A}^\dagger(\mathcal{A}(ZZ^\dagger) - y) \cdot Q_{Z} Q_{Z}^\dagger\|_2 \cdot \|Z - U^\star R_Z\|_F^2 \nonumber \\ 
&\quad \quad \stackrel{(i)}{=} \widehat{\eta} \cdot 4 \Big((1 + \delta_{2r})\| ZZ^\dagger \|_2 \nonumber \\ 
&\quad \quad \quad \quad + \|\mathcal{A}^\dagger(\mathcal{A}(ZZ^\dagger) - y) \cdot Q_{Z} Q_{Z}^\dagger\|_2 \Big) \cdot \|\mathcal{A}^\dagger(\mathcal{A}(ZZ^\dagger) - y)\cdot Q_{Z} Q_{Z}^\dagger\|_2 \cdot \|Z - U^\star R_Z\|_F^2  \nonumber \\
 &\quad \quad =  \underbrace{4 \widehat{\eta} (1 + \delta_{2r}) \|ZZ^\dagger\|_2 \| \mathcal{A}^\dagger(\mathcal{A}(ZZ^\dagger) - y) \cdot Q_{Z} Q_{Z}^\dagger\|_2 \cdot \|Z - U^\star R_Z\|_F^2}_{:=A} \nonumber \\ &\quad \quad \quad \quad \quad \quad \quad \quad + 4 \widehat{\eta} \|\mathcal{A}^\dagger(\mathcal{A}(ZZ^\dagger) - y) \cdot Q_{Z} Q_{Z}^\dagger\|_2^2 \cdot  \|Z - U^\star R_Z\|_F^2 \nonumber
\end{align} 
\end{small}
where $(i)$ is due to the definition of $\widehat{\eta}$.

To bound term $A$, we observe that $\|\mathcal{A}^\dagger(\mathcal{A}(ZZ^\dagger) - y) \cdot Q_{Z} Q_{Z}^\dagger\|_2 \leq \tfrac{(1 - \delta_{2r}) \sigma_r(ZZ^\dagger)}{10^3}$ or $\|\mathcal{A}^\dagger(\mathcal{A}(ZZ^\dagger) - y) \cdot Q_{Z} Q_{Z}^\dagger\|_2 \geq \tfrac{(1 - \delta_{2r}) \sigma_r(ZZ^\dagger)}{10^3}$.
This results into bounding $A$ as follows:
\begin{small}
\begin{align}
&4 \widehat{\eta} (1 + \delta_{2r}) \| ZZ^\dagger \|_2 \|\mathcal{A}^\dagger(\mathcal{A}(ZZ^\dagger) - y) \cdot Q_{Z} Q_{Z}^\dagger\|_2 \cdot \|Z - U^\star R_Z\|_F^2 \\ 
&\quad \quad \leq \max \Big\{\tfrac{4\cdot \widehat{\eta} \cdot (1 + \delta_{2r}) \| ZZ^\dagger \|_2 \cdot (1 - \delta_{2r}) \sigma_r(ZZ^\dagger)}{10^3} \cdot \|Z - U^\star R_Z\|_F^2 , \\ 
&\quad \quad \quad \quad \quad \quad \quad \quad \widehat{\eta} \cdot 4 \cdot 10^3 \kappa \tau(ZZ^\dagger) \|\mathcal{A}^\dagger(\mathcal{A}(ZZ^\dagger) - y) \cdot Q_{Z} Q_{Z}^\dagger\|_2^2 \cdot  \|Z - U^\star R_Z\|_F^2 \Big\} \\
&\quad \quad \leq \tfrac{4\cdot \widehat{\eta} \cdot (1 - \delta_{2r}^2) \| ZZ^\dagger \|_2 \cdot \sigma_r(ZZ^\dagger)}{10^3} \cdot \|Z - U^\star R_Z\|_F^2 \\ 
&\quad \quad \quad \quad \quad \quad \quad \quad + \widehat{\eta} \cdot 4 \cdot 10^3 \kappa \tau(ZZ^\dagger) \|\mathcal{A}^\dagger(\mathcal{A}(ZZ^\dagger) - y) \cdot Q_{Z} Q_{Z}^\dagger\|_2^2 \cdot  \|Z - U^\star R_Z\|_F^2.
\end{align}
\end{small}

Combining the above inequalities, we obtain:
\begin{small}
\begin{align}
\|\mathcal{A}^\dagger(\mathcal{A}(ZZ^\dagger) - y) &\cdot Q_{Z} Q_{Z}^\dagger\|_2 \cdot \|Z - U^\star R_Z\|_F^2 \nonumber \\ 
	&\stackrel{(i)}{\leq}  \tfrac{(1 - \delta_{2r}) \sigma_{r}(ZZ^\dagger)}{10^3} \cdot \|Z - U^\star R_Z\|_F^2 \nonumber \\ 
	&\quad \quad  \quad + (10^3 \kappa \tau(ZZ^\dagger)+1) \cdot  4 \cdot \widehat{\eta} \|\mathcal{A}^\dagger(\mathcal{A}(ZZ^\dagger) - y) \cdot Q_{Z} Q_{Z}^\dagger\|_2^2 \cdot \|Z - U^\star R_Z\|_F^2  \nonumber \\
	&\stackrel{(ii)}{\leq}  \tfrac{(1 - \delta_{2r}) \sigma_{r}(ZZ^\dagger)}{10^3} \cdot \|Z - U^\star R_Z\|_F^2 \nonumber \\
	&\quad \quad  \quad + (10^3 \beta^2 \kappa \tau(\rho^\star)+1) \cdot  4 \cdot \widehat{\eta} \|\mathcal{A}^\dagger(\mathcal{A}(ZZ^\dagger) - y) \cdot Q_{Z} Q_{Z}^\dagger\|_2^2 \cdot \tfrac{\left(\tfrac{3}{2}+2|\mu|\right)^2}{\kappa \tau(\rho^\star)} \tfrac{1}{10^6} \sigma_{r}(\rho^\star)  \nonumber \\
	&\stackrel{(iii)}{\leq}  \tfrac{(1 - \delta_{2r}) \sigma_{r}(ZZ^\dagger)}{10^3} \cdot \|Z - U^\star R_Z\|_F^2 \nonumber \\ 
	&\quad\quad \quad + 4 \cdot 1001 \beta^2 \cdot \widehat{\eta} \cdot \|\mathcal{A}^\dagger(\mathcal{A}(ZZ^\dagger) - y) \cdot Q_{Z} Q_{Z}^\dagger\|_2^2 \cdot \tfrac{\left(\tfrac{3}{2} + 2|\mu|\right)^2}{10^6\left(1  - \left(\tfrac{3}{2} + 2|\mu|\right)\tfrac{1}{10^3}\right)^2}\sigma_{r}(ZZ^\dagger)   \nonumber \\ 
	&\stackrel{(iv)}{\leq}  \tfrac{(1 - \delta_{2r}) \sigma_{r}(ZZ^\dagger)}{10^3} \cdot \|Z - U^\star R_Z\|_F^2 \nonumber \\ 
	&\quad\quad \quad + 4 \cdot 1001 \beta^2 \cdot \widehat{\eta} \cdot \tfrac{\left(\tfrac{3}{2} + 2|\mu|\right)^2}{10^6 \left(1  - \left(\tfrac{3}{2} + 2|\mu|\right)\tfrac{1}{10^3}\right)^2} \cdot \|\mathcal{A}^\dagger(\mathcal{A}(ZZ^\dagger) - y) \cdot Z\|_F^2   \nonumber \\ 
	&\stackrel{(v)}{\leq}  \tfrac{(1 - \delta_{2r}) \left(1 + (\tfrac{3}{2} + 2|\mu|)\tfrac{1}{10^3}\right)^2 \sigma_{r}(\rho^\star)}{10^3} \cdot \|Z - U^\star R_Z\|_F^2 \nonumber \\ 
	&\quad\quad \quad + \tfrac{1}{200} \beta^2 \cdot \widehat{\eta} \cdot \tfrac{\left(\tfrac{3}{2} + 2|\mu|\right)^2}{\left(1  - \left(\tfrac{3}{2} + 2|\mu|\right)\tfrac{1}{10^3}\right)^2} \cdot \|\mathcal{A}^\dagger(\mathcal{A}(ZZ^\dagger) - y) \cdot Z\|_F^2   \nonumber 
% 											   &\stackrel{(iv)}{\leq} \tfrac{m \sigma_{r}(\rho)}{40} \cdot \|Z - U^\star R\|_F^2 + \tfrac{2\widehat{\eta}}{25} \cdot \|\gradf(\\rho)U\|_F^2,   \label{proofsr1:eq_12}
\end{align}
\end{small}
where $(i)$ follows from $\widehat{\eta} \leq \tfrac{1}{4 (1+\delta_{2r}) \| ZZ^\dagger \|_2}$, 
$(ii)$ is due to Corollary \ref{cor:1}, bounding $\|Z -  U^\star R_Z\|_F \leq \rho \sigma_{r}(\rho^\star)^{1/2}$, where $\rho := \left(\tfrac{3}{2} + 2|\mu|\right) \tfrac{1}{10^3 \sqrt{\kappa \tau(\rho^\star)}}$ by Lemma \ref{lem:supp_02},
$(iii)$ is due to $(10^3 \beta^2 \kappa \tau(\rho^\star)+1) \leq 1001 \beta^2 \kappa \tau(\rho^\star)$, and by Corollary \ref{cor:2}, 
$(iv)$ is due to the fact $\sigma_{r}(ZZ^\dagger)\|\mathcal{A}^\dagger(\mathcal{A}(ZZ^\dagger) - y) \cdot Q_{Z} Q_{Z}^\dagger\|_2^2 \leq  \| \mathcal{A}^\dagger(\mathcal{A}(ZZ^\dagger) - y) Z\|_F^2$,
and $(v)$ is due to Corollary \ref{cor:2}.

Next, we bound the second term in equation~\eqref{proofsr1:eq_11}:
\begin{align}
\|\mathcal{A}^\dagger(\mathcal{A}(ZZ^\dagger) - y) &\cdot Q_{U^\star} Q_{U^\star}^\dagger\|_2 \cdot \|Z - U^\star R_Z\|_F^2 \nonumber \\ 
&\stackrel{(i)}{\leq}  \|\mathcal{A}^\dagger(\mathcal{A}(ZZ^\dagger) - y) - \mathcal{A}^\dagger(\mathcal{A}(\rho^\star) - y) \|_2 \cdot \|Z - U^\star R_Z\|_F^2 \nonumber \\
&\stackrel{(ii)}{\leq} (1 + \delta_{2r}) \cdot \|ZZ^\dagger - U^\star U^{\star \dagger}\|_F \cdot \|Z - U^\star R_Z\|_F^2 \nonumber \\
&\stackrel{(iii)}{\leq} (1 + \delta_{2r})  (2+\rho) \cdot \rho \cdot \sigma_1(U^\star) \cdot \sigma_r(U^\star) \cdot \|Z - U^\star R_Z\|_F^2 \nonumber \\
&\stackrel{(iv)}{\leq}  (1 + \delta_{2r})  (2+\rho) \left(\tfrac{3}{2} + 2|\mu|\right) \cdot \tfrac{1}{10^3} \sigma_r(\rho^\star) \cdot \|Z - U^\star R_Z\|_F^2 \nonumber \\
&\leq (1 + \delta_{2r})  \left(2+\left(\tfrac{3}{2} + 2|\mu|\right) \cdot \tfrac{1}{10^3}\right) \left(\tfrac{3}{2} + 2|\mu|\right) \cdot \tfrac{1}{10^3} \sigma_r(\rho^\star) \cdot \|Z - U^\star R_Z\|_F^2, \nonumber
\end{align}
where $(i)$ follows from $\|\mathcal{A}^\dagger(\mathcal{A}(ZZ^\dagger) - y) \cdot Q_{U^\star} Q_{U^\star}^\dagger\|_2 \leq \|\mathcal{A}^\dagger(\mathcal{A}(ZZ^\dagger) - y)\|_2$ and $\mathcal{A}^\dagger(\mathcal{A}(\rho^\star) - y) =0$, 
$(ii)$ is due to smoothness of $f$ and the RIP constants, $(iii)$ follows from \cite[Lemma 18]{bhojanapalli2016dropping}, for $\rho = \left(\tfrac{3}{2} + 2|\mu|\right) \cdot \tfrac{1}{10^3 \sqrt{\kappa \tau(\rho^\star)}}$, 
$(iv)$ follows from substituting $\rho$ above, and observing that $\tau(\rho^\star) = \sigma_1(U^\star)^2 / \sigma_r(U^\star)^2 > 1$ and $\kappa = (1+\delta_{2r})/(1-\delta_{2r}) > 1$.

Combining the above we get:
\begin{small}
 \begin{align}
\Big \langle &\mathcal{A}^\dagger(\mathcal{A}(ZZ^\dagger) - y), (Z - U^\star R_Z)(Z - U^\star R_Z)^\dagger \Big \rangle \nonumber \\ 
		  &\geq - \Bigg( \theta \sigma_r(\rho^\star) \cdot \|Z - U^\star R_Z\|_F^2 + \tfrac{1}{200} \beta^2 \cdot \widehat{\eta} \cdot \tfrac{\left(\tfrac{3}{2} + 2|\mu| \right)^2}{\left(1  - \left(\tfrac{3}{2} + 2|\mu|\right)\tfrac{1}{10^3}\right)^2} \cdot \|\mathcal{A}^\dagger(\mathcal{A}(ZZ^\dagger) - y) \cdot Z\|_F^2\Bigg)
\end{align}
\end{small}
where $\theta = \tfrac{(1 - \delta_{2r}) \left(1 + \left(\tfrac{3}{2} + 2|\mu| \right)\tfrac{1}{10^3}\right)^2}{10^3} +  (1 + \delta_{2r})  \left(2+\left(\tfrac{3}{2} + 2|\mu|\right) \cdot \tfrac{1}{10^3}\right) \left(\tfrac{3}{2} + 2|\mu|\right) \cdot \tfrac{1}{10^3}$.
\end{proof}

\begin{lemma}{\label{lem:001}}
Under identical assumptions with Lemma \ref{lem:000}, the following inequality holds: 
\begin{small}
\begin{align}
\Big \langle \mathcal{A}^\dagger( \mathcal{A}(ZZ^\dagger) - y), ZZ^\dagger - U^\star U^{\star \dagger}\Big \rangle \geq 1.1172 \eta \left \|\mathcal{A}^\dagger( \mathcal{A}(ZZ^\dagger) - y)  Z \right \|_F^2 + \tfrac{1 - \delta_{2r}}{2} \|U^\star U^{\star \dagger} - ZZ^\dagger\|_F^2
\end{align}
\end{small}
\end{lemma}

\begin{proof}
By smoothness assumption of the objective, based on the RIP assumption, we have:
\begin{small}
\begin{align}
\tfrac{1}{2} \|\mathcal{A}(ZZ^\dagger) - y\|_2^2 &\geq \tfrac{1}{2} \|\mathcal{A}(U_{+} U_{+}^\dagger) - y\|_2^2 \\ 
								   &\quad \quad - \left \langle \mathcal{A}^\dagger( \mathcal{A}(ZZ^\dagger) - y), U_{+} U_{+}^\dagger - ZZ^\dagger \right \rangle - \tfrac{1+\delta_{2r}}{2} \|U_{+} U_{+}^\dagger - ZZ^\dagger\|_F^2 \Rightarrow \\
\tfrac{1}{2} \|\mathcal{A}(ZZ^\dagger) - y\|_2^2 &\geq \tfrac{1}{2} \|\mathcal{A}(U^\star U^{\star \dagger}) - y\|_2^2 \\ 
								   &\quad \quad - \left \langle \mathcal{A}^\dagger( \mathcal{A}(ZZ^\dagger) - y), U_{+} U_{+}^\dagger - ZZ^\dagger \right \rangle - \tfrac{1+\delta_{2r}}{2} \|U_{+} U_{+}^\dagger - ZZ^\dagger\|_F^2								   
\end{align}
\end{small}
due to the optimality $\|\mathcal{A}(U^\star U^{\star \dagger}) - y\|_2^2 = 0 \leq \|\mathcal{A}(VV^\dagger) - y\|_2^2$, for any $V \in \mathbb{C}^{d \times r}$.
Also, by the restricted strong convexity with RIP, we get:
\begin{small}
\begin{align}
\tfrac{1}{2} \|\mathcal{A}(U^\star U^{\star \dagger}) - y\|_2^2 &\geq \tfrac{1}{2} \|\mathcal{A}(ZZ^\dagger) - y\|_2^2 \\ 
								   &\quad \quad + \left \langle \mathcal{A}^\dagger( \mathcal{A}(ZZ^\dagger) - y), U^\star U^{\star \dagger} - ZZ^\dagger \right \rangle + \tfrac{1-\delta_{2r}}{2} \|U^\star U^{\star \dagger} - ZZ^\dagger\|_F^2					   
\end{align}
\end{small}
Adding the two inequalities, we obtain:
\begin{small}
\begin{align}
\left \langle \mathcal{A}^\dagger( \mathcal{A}(ZZ^\dagger) - y), ZZ^\dagger - U^\star U^{\star \dagger}\right \rangle &\geq \left \langle \mathcal{A}^\dagger( \mathcal{A}(ZZ^\dagger) - y), ZZ^\dagger - U_{+} U_{+}^\dagger\right \rangle \\
&\quad \quad - \tfrac{1+\delta_{2r}}{2} \|U_+ U_+^\dagger - ZZ^\dagger \|_F^2 + \tfrac{1 - \delta_{2r}}{2} \|U^\star U^{\star \dagger} - ZZ^\dagger\|_F^2
\end{align}
\end{small}
To proceed we observe:
\begin{small}
\begin{align}
U_+ U_+^\dagger &= \left(Z - \eta \mathcal{A}^\dagger \left(\mathcal{A}(ZZ^\dagger) - y \right)Z \right) \cdot \left(Z - \eta \mathcal{A}^\dagger \left(\mathcal{A}(ZZ^\dagger) - y \right)Z \right)^\dagger \\
		      &= ZZ^\dagger - \eta ZZ^\dagger \cdot \mathcal{A}^\dagger \left(\mathcal{A}(ZZ^\dagger) - y \right) - \eta \mathcal{A}^\dagger \left(\mathcal{A}(ZZ^\dagger) - y \right) \cdot ZZ^\dagger \\
		      & \quad \quad + \eta^2 \mathcal{A}^\dagger \left(\mathcal{A}(ZZ^\dagger) - y \right) \cdot ZZ^\dagger \cdot \mathcal{A}^\dagger \left(\mathcal{A}(ZZ^\dagger) - y \right) \\
		      &\stackrel{(i)}{=} ZZ^\dagger - \left(I - \tfrac{\eta}{2} Q_Z Q_Z^\dagger  \mathcal{A}^\dagger \left(\mathcal{A}(ZZ^\dagger) - y \right) \right) \cdot \eta ZZ^\dagger \cdot \mathcal{A}^\dagger \left(\mathcal{A}(ZZ^\dagger) - y \right) \\ &\quad \quad - \eta \mathcal{A}^\dagger \left(\mathcal{A}(ZZ^\dagger) - y \right) \cdot ZZ^\dagger  \cdot \left(I - \tfrac{\eta}{2} Q_Z Q_Z^\dagger \mathcal{A}^\dagger \left(\mathcal{A}(ZZ^\dagger) - y \right) \right) 
\end{align}
\end{small}
where $(i)$ is due to the fact $\mathcal{A}^\dagger \left(\mathcal{A}(ZZ^\dagger) - y \right) \cdot ZZ^\dagger \cdot \mathcal{A}^\dagger \left(\mathcal{A}(ZZ^\dagger) - y \right) = \mathcal{A}^\dagger \left(\mathcal{A}(ZZ^\dagger) - y \right) \cdot Q_ZQ_Z^\dagger \cdot ZZ^\dagger \cdot Q_ZQ_Z^\dagger \cdot \mathcal{A}^\dagger \left(\mathcal{A}(ZZ^\dagger) - y \right)$, for $Q_Z$ a basis matrix whose columns span the column space of $Z$; also, $I$ is the identity matrix whose dimension is apparent from the context.
%From Lemma \ref{lem:equiveta}, we know that $\eta \leq \widehat{\eta}$.
%
Thus:
\begin{small}
\begin{align}
\tfrac{\eta}{2} Q_Z Q_Z^\dagger \mathcal{A}^\dagger \left(\mathcal{A}(ZZ^\dagger) - y \right) \preceq \tfrac{10.5}{10} \tfrac{\widehat{\eta}}{2} Q_Z Q_Z^\dagger \mathcal{A}^\dagger \left(\mathcal{A}(ZZ^\dagger) - y \right),
\end{align}
\end{small}
and, hence,
\begin{small}
\begin{align}
I - \tfrac{\eta}{2} Q_Z Q_Z^\dagger \mathcal{A}^\dagger \left(\mathcal{A}(ZZ^\dagger) - y \right) \succeq I - \tfrac{10.5}{10} \tfrac{\widehat{\eta}}{2} Q_Z Q_Z^\dagger \mathcal{A}^\dagger \left(\mathcal{A}(ZZ^\dagger) - y \right).
\end{align}
\end{small}
Define $\Psi = I - \tfrac{\eta}{2} Q_Z Q_Z^\dagger \mathcal{A}^\dagger \left(\mathcal{A}(ZZ^\dagger) - y \right).$
Then, using the definition of $\widehat{\eta}$, we know that $\widehat{\eta} \leq \tfrac{1}{4\|Q_Z Q_Z^\dagger \mathcal{A}^\dagger \left(\mathcal{A}(ZZ^\dagger) - y \right)\|_2}$, and thus:
\begin{small}
\begin{align}
\Psi \succ 0, \quad \sigma_1(\Psi) \leq 1 + \tfrac{21}{160}, \quad \text{and} \quad \sigma_n(\Psi) \geq 1 - \tfrac{21}{160}.
\end{align}
\end{small}
Going back to the main recursion and using the above expression for $U_+ U_+^\dagger$, we have:
\begin{small}
\begin{align}
\Big \langle \mathcal{A}^\dagger( \mathcal{A}(ZZ^\dagger) - y), &ZZ^\dagger - U^\star U^{\star \dagger}\Big \rangle - \tfrac{1 - \delta_{2r}}{2} \|U^\star U^{\star \dagger} - ZZ^\dagger\|_F^2 \\
			&\geq \left \langle \mathcal{A}^\dagger( \mathcal{A}(ZZ^\dagger) - y), ZZ^\dagger - U_{+} U_{+}^\dagger \right \rangle - \tfrac{1+\delta_{2r}}{2} \|U_{+} U_{+}^\dagger - ZZ^\dagger\|_F^2 \\ 
			&\stackrel{(i)}{\geq} 2 \eta \left \langle \mathcal{A}^\dagger( \mathcal{A}(ZZ^\dagger) - y),  \mathcal{A}^\dagger( \mathcal{A}(ZZ^\dagger) - y) \cdot ZZ^\dagger \cdot \Psi \right \rangle \\
			 &\quad \quad - \tfrac{1+\delta_{2r}}{2} \|2\eta  \mathcal{A}^\dagger( \mathcal{A}(ZZ^\dagger) - y) \cdot ZZ^\dagger \cdot \Psi\|_F^2 \\			
			 &\stackrel{(ii)}{\geq} 2\left(1 - \tfrac{21}{160}\right) \eta \left \|\mathcal{A}^\dagger( \mathcal{A}(ZZ^\dagger) - y)  Z \right \|_F^2 \\ 
			 &\quad \quad - 2(1 + \delta_{2r}) \eta^2 \left \|\mathcal{A}^\dagger( \mathcal{A}(ZZ^\dagger) - y)  Z \right \|_F^2 \cdot \|Z\|_2^2 \cdot \|\Psi\|_2^2 \\
			 &\stackrel{(iii)}{\geq} 2\left(1 - \tfrac{21}{160}\right) \eta \left \|\mathcal{A}^\dagger( \mathcal{A}(ZZ^\dagger) - y)  Z \right \|_F^2 \\ 
			 &\quad \quad - 2(1 + \delta_{2r}) \eta^2 \left \|\mathcal{A}^\dagger( \mathcal{A}(ZZ^\dagger) - y)  Z \right \|_F^2 \cdot \|Z\|_2^2 \cdot \left(1 + \tfrac{21}{160}\right)^2 \\
			 &= 2\left(1 - \tfrac{21}{160}\right) \eta \left \|\mathcal{A}^\dagger( \mathcal{A}(ZZ^\dagger) - y)  Z \right \|_F^2 \cdot \left(1 - 2(1 + \delta_{2r}) \eta  \cdot \|Z\|_2^2 \cdot \left(1 + \tfrac{21}{160}\right)^2 \cdot \tfrac{1}{2(1 - \tfrac{21}{160})}\right) \\
			 &\stackrel{(iv)}{\geq} 2\left(1 - \tfrac{21}{160}\right) \eta \left \|\mathcal{A}^\dagger( \mathcal{A}(ZZ^\dagger) - y)  Z \right \|_F^2 \cdot \left(1 - 2(1 + \delta_{2r}) \tfrac{10.5}{10} \widehat{\eta}  \cdot \|Z\|_2^2 \cdot \left(1 + \tfrac{21}{160}\right)^2 \cdot \tfrac{1}{2(1 - \tfrac{21}{160})}\right) \\
			 &\stackrel{(v)}{\geq} 2\left(1 - \tfrac{21}{160}\right) \eta \left \|\mathcal{A}^\dagger( \mathcal{A}(ZZ^\dagger) - y)  Z \right \|_F^2 \cdot \left(1 - \tfrac{10.5}{10} \tfrac{ \left(1 + \tfrac{21}{160}\right)^2}{4(1 - \tfrac{21}{160})}\right) \\
			 &= 1.0656 \eta \left \|\mathcal{A}^\dagger( \mathcal{A}(ZZ^\dagger) - y)  Z \right \|_F^2
\end{align}
\end{small}
where $(i)$ is due to the symmetry of the objective;
$(ii)$ is due to Cauchy-Schwarz inequality and the fact:
\begin{small}
\begin{align}
\Big \langle  \mathcal{A}^\dagger( \mathcal{A}(ZZ^\dagger) - y), & ~\mathcal{A}^\dagger( \mathcal{A}(ZZ^\dagger) - y) \cdot ZZ^\dagger \cdot \Psi \Big \rangle \\ 
											   &= \Big \langle  \mathcal{A}^\dagger( \mathcal{A}(ZZ^\dagger) - y),  \mathcal{A}^\dagger( \mathcal{A}(ZZ^\dagger) - y) \cdot ZZ^\dagger \Big \rangle  \\ &\quad \quad - \tfrac{\eta}{2} \Big \langle  \mathcal{A}^\dagger( \mathcal{A}(ZZ^\dagger) - y),  \mathcal{A}^\dagger( \mathcal{A}(ZZ^\dagger) - y) \cdot ZZ^\dagger \cdot \mathcal{A}^\dagger( \mathcal{A}(ZZ^\dagger) - y) \Big \rangle \\
											   &\stackrel{(i)}{\geq} \Big \langle  \mathcal{A}^\dagger( \mathcal{A}(ZZ^\dagger) - y),  \mathcal{A}^\dagger( \mathcal{A}(ZZ^\dagger) - y) \cdot ZZ^\dagger \Big \rangle  \\ &\quad \quad - \tfrac{10.5}{10} \tfrac{\widehat{\eta}}{2} \Big \langle  \mathcal{A}^\dagger( \mathcal{A}(ZZ^\dagger) - y),  \mathcal{A}^\dagger( \mathcal{A}(ZZ^\dagger) - y) \cdot ZZ^\dagger \cdot \mathcal{A}^\dagger( \mathcal{A}(ZZ^\dagger) - y) \Big \rangle \\
											   &\geq \left( 1 - \tfrac{10.5}{10} \tfrac{\widehat{\eta}}{2} \|Q_Z Q_Z^\dagger \mathcal{A}^\dagger( \mathcal{A}(ZZ^\dagger) - y) \|_2^2 \right) \cdot \left \|\mathcal{A}^\dagger( \mathcal{A}(ZZ^\dagger) - y)  Z \right \|_F^2 \\
											   &\geq \left( 1 - \tfrac{21}{160}\right) \left \|\mathcal{A}^\dagger( \mathcal{A}(ZZ^\dagger) - y)  Z \right \|_F^2
\end{align}
\end{small}
where $(i)$ is due to $\eta \leq \tfrac{10.5}{10} \widehat{\eta}$, and the last inequality comes from the definition of the $\widehat{\eta}$ and its upper bound;
$(iii)$ is due to the upper bound on $\|\Psi\|_2$ above;
$(iv)$ is due to $\eta \leq \tfrac{10.5}{10}\widehat{\eta}$;
$(v)$ is due to $\widehat{\eta} \leq \tfrac{1}{4 (1 + \delta_{2r}) \|ZZ^\dagger\|_2}$.
The above lead to the desiderata:
\begin{small}
\begin{align}
\Big \langle \mathcal{A}^\dagger( \mathcal{A}(ZZ^\dagger) - y), ZZ^\dagger - U^\star U^{\star \dagger}\Big \rangle \geq 1.0656 \eta \left \|\mathcal{A}^\dagger( \mathcal{A}(ZZ^\dagger) - y)  Z \right \|_F^2 + \tfrac{1 - \delta_{2r}}{2} \|U^\star U^{\star \dagger} - ZZ^\dagger\|_F^2
\end{align}
\end{small}
\end{proof}

\end{document}